\documentclass[11pt]{article}
\usepackage{bbm}
\usepackage{geometry}
\usepackage{color}
\usepackage{amsfonts}
\usepackage{algorithm}
\usepackage{algorithmic}
\usepackage{latexsym, amssymb, amsmath, amscd, amsthm, amsxtra}
\usepackage{mathtools}
\usepackage{enumerate}
\usepackage[all]{xy}
\usepackage{mathrsfs}
\usepackage{fancyhdr}
\usepackage{listings}
\usepackage{hyperref}
\usepackage{enumitem}
\usepackage{multicol}

\def\P{\mathbb{P}}

\def\E{\mathbb{E}}

\def\11{\mathbbm{1}}

\def\G{\mathsf G}

\def\ER{Erd\H{o}s--R\'enyi\ }

\newtheorem{Theorem}{Theorem}[section]

\newtheorem{Definition}[Theorem]{Definition}
\newtheorem{Lemma}[Theorem]{Lemma}
\newtheorem{Corollary}[Theorem]{Corollary}
\newtheorem{Proposition}[Theorem]{Proposition}
\newtheorem{Remark}[Theorem]{Remark}
\newtheorem{Claim}[Theorem]{Claim}

\numberwithin{equation}{section}

\makeatletter
\newenvironment{breakablealgorithm}
{
		\begin{center}
			\refstepcounter{algorithm}
			\hrule height.8pt depth0pt \kern2pt
			\renewcommand{\caption}[2][\relax]{
				{\raggedright\textbf{\ALG@name~\thealgorithm} ##2\par}%
				\ifx\relax##1\relax 
				\addcontentsline{loa}{algorithm}{\protect\numberline{\thealgorithm}##2}%
				\else 
				\addcontentsline{loa}{algorithm}{\protect\numberline{\thealgorithm}##1}%
				\fi
				\kern2pt\hrule\kern2pt
			}
		}{
		\kern2pt\hrule\relax
	\end{center}
}
\makeatother

\hypersetup{colorlinks=true,linkcolor=blue}

\title{A polynomial-time iterative algorithm for random graph matching with non-vanishing correlation}

\author{Jian Ding \\ Peking University  \and  Zhangsong Li \\ Peking University}
\date{\today}

\begin{document}

\maketitle

\begin{abstract}
We propose an efficient algorithm for matching two correlated \ER graphs with $n$ vertices whose edges are correlated through a latent vertex correspondence. When the edge density $q= n^{-\alpha+o(1)}$ for a constant $\alpha \in [0,1)$, we show that our algorithm has polynomial running time and succeeds to recover the latent matching as long as the edge correlation is non-vanishing. This is closely related to our previous work on a polynomial-time algorithm that matches two Gaussian Wigner matrices with non-vanishing correlation, and provides the first polynomial-time random graph matching algorithm (regardless of the regime of $q$) when the edge correlation is below the square root of the Otter's constant  (which is $\approx 0.338$).
\end{abstract}

\section{Introduction}

In this paper, we study the algorithmic perspective for recovering the latent matching between two correlated \ER graphs. To be mathematically precise, we first need to choose a model for a pair of correlated \ER graphs, and a natural choice is that the two graphs are independently sub-sampled from a common \ER graph, as described more precisely next. For two vertex sets $V$ and $\mathsf V$ with cardinality $n$, let $E_0$ be the set of unordered pairs $(u,v)$ with $u,v\in V, u\neq v$ and define $\mathsf E_0$ similarly with respect to $\mathsf V$. For some model parameters $p,s\in (0,1)$, we generate a pair of correlated random graphs $G=(V,E)$ and $\mathsf G=(\mathsf V,\mathsf E)$ with the following procedure: sample a uniform bijection $\pi:V\to \mathsf V$, independent Bernoulli variables $I_{(u,v)}$ with parameter $p$ for $(u,v)\in E_0$ as well as independent Bernoulli variables $J_{(u,v)}, \mathsf J_{(\mathsf u,\mathsf v)}$ with parameter $s$ for $(u,v)\in E_0$ and $(\mathsf u,\mathsf v)\in \mathsf E_0$. Let 
\begin{equation}{\label{eq:def-of-G_e}}
    G_{(u,v)}=I_{(u,v)}J_{(u,v)},\forall (u,v)\in E_0\,, \quad \G_{(\mathsf u,\mathsf v)}=I_{\left(\pi^{-1}(\mathsf u),\pi^{-1}(\mathsf v)\right)}\mathsf J_{(\mathsf u,\mathsf v)},\forall (\mathsf u,\mathsf v)\in \mathsf E_0\,,
\end{equation}
and let $E=\{e\in E_0: G_e=1\}, \mathsf E=\{\mathsf e\in \mathsf E_0:\G_{\mathsf e}=1\}$. It is obvious that marginally $G$ is an \ER graph with edge density $q = ps$  and so is $\mathsf G$. In addition, the edge correlation $\rho$ is given by $\rho = s(1-p)/(1-ps)$.

An important question is to recover the \emph{latent} matching $\pi$ from the observation of $(G, \mathsf G)$. More precisely, we wish to find an estimator $\Hat \pi$ which is measurable with respect to $(G, \mathsf G)$ such that $\Hat \pi = \pi$. Our main contribution is an efficient matching algorithm as incorporated in the theorem below.
\begin{Theorem}{\label{Main-Thm}}
Suppose that $q = n^{-\alpha+o(1)} \leq 1/2$ for some constant $\alpha\in [0, 1)$ and that $\rho\in (0, 1]$ is a constant. Then there exist a constant $C=C(\alpha,\rho)$ and an algorithm (see Algorithm~\ref{algo:matching} below) with time-complexity $O(n^{C})$ that recovers the latent matching with probability $1-o(1)$. That is, this polynomial-time algorithm takes $(G, \mathsf G)$ as the input and outputs an estimator $\Hat{\pi}$ such that $\Hat{\pi} = \pi$ with probability tending to 1 as $n\to \infty$.
\end{Theorem}

\subsection{Backgrounds and related works}

The random graph matching problem is motivated by questions from various applied fields such as social network analysis \cite{NS08, NS09}, computer vision \cite{CSS07, BBM05}, computational biology \cite{SXB08, VCP15, ECK15} and natural language processing \cite{HNM05}.  
In biology, an important problem is to identify proteins with similar structures/functions across different species. Toward this goal, directly comparing amino acid sequences that constitute proteins is often complicated, since genetic mutations of the species can result in significant variations of such a sequence. However, despite these variations, proteins typically maintain similar functions within each species' metabolism. In light of this, biologists employ graph-based representations, such as Protein-Protein Interaction (PPI) graphs, for each species. Under the assumption that the topological structures of PPI graphs are similar across species, researchers can then effectively match proteins with similar functions by taking advantage of such similarity (and by possibly taking advantage of some other domain knowledge). This approach turns out to be successful and offers a nuanced understanding of phenomena such as the evolution of protein complexes. We refer the reader to \cite{ECK15} for more information on the topic.
In the domain of social networks, data privacy is a fundamental and complicated problem. One complication arises from the fact that a user may have accounts in multiple social network platforms, where the user shares similar content or engages in comparable activities. Thus, from such similarities in different platforms, it is possible to infer their identities by aligning the respective graphs representing user interactions and content consumption.  That is to say, it is possible to use graph matching techniques to deanonymize private social networks using information gleaned from public social platforms. 
In this scope, well-known examples include the deanonymization of Netflix using data from IMDb \cite{NS08} and the deanonymization of Twitter using Flickr \cite{NS09}.
Viewing this problem from the opposite perspective, a deeper understanding on the random graph matching problem may offer insights on better mechanism to protect data privacy. With aforementioned applications in mind, the graph matching problem has been extensively studied recently from a theoretical point of view.

It is natural that different applications have their own distinct features, which mathematically boils down to a careful choice of underlying random graph model suitable for the desired application. Similar to most previous works on the random graph matching problem, in this paper we consider the correlated \ER random graph model, which is possibly an over idealization for any realistic network but nevertheless offers a good playground to develop insights and methods for this problem in general.
By the collective efforts as in \cite{CK16, CK17, HM20, WXY20+,WXY21+, GML21, DD22+, DD22+b}, it is fair to say that we have a fairly complete understanding on the information thresholds for the problem of correlation detection and vertex matching. In contrast, the understanding on the computational aspect is far from being complete, and in what follows we briefly review the progress on this front.

A huge amount of efforts have been devoted to developing efficient algorithms for graph matching \cite{PG11, YG13, LFP14, KHG15, FQRM+16, SGE17, BCL19, DMWX21, FMWX22a, FMWX22b, BSH19, CKMP19,DCKG19, MX20, GM20, MRT21, MRT21+, MWXY21+, GML22, GMS22+, MWXY22+}. Previous to this work, arguably the best result is the recent work \cite{MWXY22+} (see also \cite{GML22, GMS22+} for a remarkable result on partial recovery of similar flavor when the average degree is $O(1)$), where the authors substantially improved a groundbreaking work \cite{MRT21+} and obtained a polynomial time algorithm which succeeds as long as the correlation is above the square root of the Otter's constant (the Otter's constant is around ${0.338}$). In terms of methods, the present work seems drastically different from these existing works; instead, this work is closely related to our previous work \cite{DL22+} on a polynomial-time iterative algorithm that matches two correlated Gaussian Wigner matrices. One is encouraged to see \cite[Section 1.1]{MWXY22+} and \cite[Section 1.1]{DL22+} for a more elaborated review on previous algorithms. In addition, one is encouraged to see \cite[Section 1.2]{DL22+} for discussions on novel features of this iterative algorithm, especially in comparison with the message-passing algorithm \cite{GMS22+, PSSZ22} and a recent greedy algorithm for aligning two independent graphs \cite{DDG22+}.

\subsection{Our contributions}
While the present work can be viewed as an extension of \cite{DL22+}, we do feel that we have conquered substantial obstacles and have made substantial contributions to this topic, as we describe below. 
\begin{itemize}
\item While in \cite{MWXY22+} (see also \cite{GMS22+}) polynomial-time matching algorithms were obtained when the correlation is above the threshold from the Otter's constant, our work establishes a polynomial-time algorithm as long as the average degree grows polynomially and the correlation is non-vanishing. In addition, the power in the running time only tends to $\infty$ as the correlation tends to $0$ (for each fixed $\alpha<1$), and we are under the feeling that this is the best possible; such feeling is supported by a recent work on the complexity of low-degree polynomials for graph matching \cite{DDL23+}.
\item From a conceptual point of view, our work demonstrates the robustness of the iterative matching algorithm proposed in \cite{DL22+}. This type of ``algorithmic universality'' is closely related to the universality phenomenon in random matrix theory, which roughly speaking postulates that the particular distribution for entries of random matrices is often irrelevant for spectral properties under investigation. Our work also encourages future study on even more ambitious perspectives for robustness, for instance algorithms that are robust with assumptions on underlying random graph models. This is of major interest since realistic networks are usually captured better by more structured graph models, such as the random geometric graph model \cite{WWXY22+}, the random growing graph model \cite{RS20+} and the stochastic block model \cite{RS21}. 
\item In terms of techniques, our work employs a method which argues that Gaussian approximation is valid typically. There are a couple of major challenges : (1) Gaussian approximation is valid only in a rather weak sense that the Radon-Nikodym derivative is not \emph{too large}; (2) we can not simply ignore the atypical cases when Gaussian approximation is invalid since we have to analyze the conditional behavior given \emph{all} the information the algorithm has exploited so far. The latter raises a major obstacle in our proof for theoretical guarantee; see Section~\ref{sec:proof-outline} for more detailed discussions on how this obstacle is addressed. 
\end{itemize}
In addition, it is natural to suspect that our work brings us one step closer to understanding computational phase transitions for random graph matching problems as well as understanding algorithmic perspectives for other matching problems (see, e.g., \cite{CJMNZ22+}). We refer an interested reader to \cite[Section 1.3]{DL22+} and omit further discussions here.

\subsection{Notations}
We record in this subsection some notation conventions, and we point out that a list of commonly used notations is included at the end of the paper for better reference.

Denote the identity matrix by $\mathrm{I}$ and the zero matrix by $\mathrm{O}$. For a $d\!*\!m$ matrix $\mathrm{A}$, we use $\mathrm{A}^{*}$ to denote its transpose, and let $\| \mathrm{A} \|_{\mathrm{HS}}$ denote the Hilbert-Schmidt norm of $\mathrm{A}$. 
For $1 \leq s \leq \infty$, define the $s$-norm of $\mathrm{A}$ by $\|\mathrm{A}\|_{s} = \sup \{ \| \mathrm{A} x^* \|_{s}: \|x\|_{s} =1 \}$. Note that when $\mathrm{A}$ is a symmetric square matrix, we have for $\tfrac{1}{s}+\tfrac{1}{t}=1$
\[
\|\mathrm{A}\|_{s} = \sup \{ y \mathrm{A} x^*: \|x\|_{s} =1, \|y\|_{t}=1 \} = \sup \{ \| y \mathrm{A} \|_{t} : \|y\|_{t}=1 \} = \| \mathrm{A}^* \|_t = \| \mathrm{A} \|_t \,.
\]
We further denote the operator norm of $\mathrm{A}$ by $\|\mathrm{A}\|_{\mathrm{op}} = \|\mathrm{A}\|_2$. If $m=d$, denote $\mathrm{det(A)}$ and $\mathrm{tr(A)}$ the determinant and the trace of $\mathrm{A}$, respectively. For $d$-dimensional vectors $x,y$ and a $d\!*\!d$ symmetric matrix $\Sigma$, let $\langle x,y \rangle_{\Sigma} = x \Sigma y^{*}$ be the ``inner product'' of $x,y$ with respect to $\Sigma$, and we further denote $\| x \|_{\Sigma}^2 = x \Sigma x^{*}$. For two vectors $\gamma,\mu \in \mathbb{R}^d$, we say $\gamma \geq \mu$ (or equivalently $\mu \leq \gamma$) if their entries satisfy $\gamma(i) \geq \mu(i)$ for all $1 \leq i \leq d$. The indicator function of a set $A$ is denoted by $\mathbf{1}_{A}$.

Without further specification, all asymptotics are taken with respect to $n \to \infty$. We also use standard asymptotic notations: for two sequences $\{ a_n \}$ and $\{ b_n \}$, we write $a_n = O(b_n)$ or $a_n \lesssim b_n$, if $|a_n| \leq C |b_n|$ for some absolute constant $C$ and all $n$. We write $a_n = \Omega(b_n)$ or $a_n \gtrsim b_n$, if $b_n = O(a_n)$; we write $a_n = \Theta(b_n)$ or $a_n \asymp b_n$, if $a_n = O(b_n)$ and $a_n = \Omega(b_n)$; we write $a_n = o(b_n)$ or $b_n = \omega(a_n)$, if $\frac{a_n}{b_n} \to 0$ as $n \to \infty$. We write $a_n \sim b_n$ if $\frac{a_n}{b_n} \to 1$.

We denote by $\mathrm{Ber}(p)$ the Bernoulli distribution with parameter $p$, denote by $\mathrm{Bin}(n,p)$ the Binomial distribution with $n$ trials and success probability $p$, denote by $\mathcal{N}(\mu, \sigma^2)$ the normal distribution with mean $\mu$ and variance $\sigma^2$, and denote by $\mathcal{N}(\mu,\Sigma)$ the multi-variate normal distribution with mean $\mu$ and covariance matrix $\Sigma$. We say $(X,Y)$ is a pair of correlated binomial random variables, denoted as $\mathrm{CorBin}(N,M,p;m,\rho)$ for $m\leq \min\{N,M\}$ and $\rho \in [0,1]$, if $(X,Y) \sim \big( \sum_{k=1}^{N} b_k, \sum_{k=1}^{M} b^{\prime}_k \big)$ with $b_k, b^{\prime}_l \sim \mathrm{Ber}(p)$ such that  $\{ b_k, b^{\prime}_k \}$ are independent with $\{ b_1, \ldots, b_N, b^{\prime}_1,\ldots,b^{\prime}_{M} \} \setminus \{ b_k, b^{\prime}_k \}$ and the covariance between $b_k$ and $b^{\prime}_k$ is $\rho$ when $k \leq m$, and that  $b_k$ is independent with $\{ b_1, \ldots, b_N, b^{\prime}_1,\ldots,b^{\prime}_{M} \} \setminus \{ b_k \}$ and $b^{\prime}_l$ is independent with $\{ b_1, \ldots, b_N, b^{\prime}_1,\ldots,b^{\prime}_{M} \} \setminus \{ b^{\prime}_l \}$ when $k,l>m$. We say $X$ is a sub-Gaussian variable, if there exists a positive constant $C$ such that $\mathbb{P}( |X| \geq t ) \leq 2 e^{-{t^2}/{C^2}}$, and we use $\| X \|_{\psi_2}= \inf \big\{ C >0 : \mathbb{E}[ \exp \{ \frac{X^2}{C^2} \} ] \leq 2 \big\}$ to denote its sub-Gaussian norm.

\smallskip

\noindent {\bf Acknowledgment.} We thank Zongming Ma, Yihong Wu, Jiaming Xu and Fan Yang for stimulating discussions on random graph matching problems. J. Ding is partially supported by NSFC Key Program Project No. 12231002 and the Xplorer Prize.

\section{An iterative matching algorithm} \label{sec:algorithm-description}

We first describe the underlying heuristics of our algorithm (the reader is strongly encouraged to consult \cite[Section 2]{DL22+} for a description on the iterative matching algorithm for correlated Gaussian Wigner matrices).  Since we expect that Wigner matrices and \ER graphs (with sufficient edge density) should belong to the same algorithmic universality class, it is natural to try to extend the algorithm proposed in \cite{DL22+} to the case for correlated random graphs. As in \cite{DL22+}, our wish is to iteratively construct a sequence of paired sets $\big( \Gamma^{(t)}_k,\Pi^{(t)}_k \big)_{1 \leq k \leq K_t}$ for $t\geq 0$ (with $\Gamma^{(t)}_k \subset V$ and $\Pi^{(t)}_k \subset \mathsf V$), where each $\big( \Gamma^{(t)}_k, \Pi^{(t)}_k \big)$ contains more true pairs of the form $(v, \pi(v))$ than the case when the two sets are sampled uniformly and independently. In addition, we may further require $|\Gamma^{(t)}_k|,|\Pi^{(t)}_k| \approx \mathfrak{a}_t n$ for convenience of analysis later.

For initialization in \cite{DL22+}, we obtain $K_0$ true pairs via brute-force search, and provided with $K_0$ true pairs we then for each such pair define $\big( \Gamma^{(0)}_k,\Pi^{(0)}_k \big)$ to be the collections of their neighbors such that the corresponding edge weights exceed a certain threshold. In this work, however, due to the sparsity of \ER graphs (when $\alpha>0$) we cannot produce an efficient initialization by simply looking at the 1-neighborhoods of some true pairs. In order to address this, we instead look at their $\chi$-neighborhoods with carefully chosen $\chi$ (see the definition of $\big( \Gamma^{(0)}_k, \Pi^{(0)}_k \big)$ in \eqref{equ-def-initial-set} below). This would require a significantly more complicated analysis since this initialization will have influence on iterations later. The idea to address this is to argue that in the initialization we have only used information on a small fraction of the edges; this is why $\chi$ will be chosen carefully.

Provided with the initialization, the iteration of the algorithm is similar to that in \cite{DL22+} (although we will introduce some modifications in order to facilitate our analysis later). Since each pair $\big( \Gamma^{(t)}_k, \Pi^{(t)}_k \big)$ carries some signal, we then hope to construct \emph{more} paired sets at time $t+1$ by considering various linear combinations of vertex degrees restricted to each $\Gamma^{(t)}_k$ (or to $\Pi^{(t)}_k$). As a key novelty of this iterative algorithm, as in \cite{DL22+} we will use the increase on the number of paired sets to compensate the decrease on the signal carried in each pair. As we hope, once the iteration progresses to time $t = t^*$ for some well chosen $t^*$ (see \eqref{eq-def-t-*} below) we would have accumulated enough total signal so that we can just complete the matching directly in the next step, as described in Section~\ref{sec:ending-iteration}. 

However, controlling the correlation among different iterative steps is a much more sophisticated job in this setting. In \cite{DL22+} we used Gaussian projection to remove the influence of conditioning on information obtained in previous steps. This is indeed a powerful technique but it crucially relies on the property of Gaussian process. Although there are examples that the universality of iterative algorithms have been established (see, e.g., \cite{BLM15, CL21, FWZ22+} for development on this front for approximate-message-passing), we are not sure how their techniques can help solving our problem since dealing with a pair of correlated matrices seems of substantial and novel challenge. Instead we try to compare the Bernoulli process we obtained in the iterative algorithm with the corresponding Gaussian process when we replace  $\{G_{u,v}, \mathsf{G}_{\mathsf{u,v}}\}$ by a Gaussian process with the same mean and covariance structure. In order to facilitate such comparison, we also apply a Gaussian smoothing to our Bernoulli process in our algorithm below (see \eqref{equ-def-iter-sets} where we introduce external Gaussian noise for smoothing purpose). However, since we need to analyze the conditional behavior of two processes, we need to compare their densities; this is much more demanding than controlling e.g. the transportation distance between two processes, and actually the density ratio of these two processes are fairly large. In order to address this, on the one hand, we will show that if we ignore a vanishing fraction of vertices (which is a highly non-trivial step as we will elaborate in Section~\ref{sec:proof-outline}), the density ratio is then under control while still being fairly large; on the other hand, we show that in the Gaussian setting some really bad event occurs with tiny probability (and thus still with small probability even after multiplying by this fairly large density ratio). We refer to Sections~\ref{sec:density-compare} and \ref{sec:Gaussian-analysis} for more detailed discussions on this point. 

Finally, due to the aforementioned complications we are only able to show that our iterative algorithm constructs an almost exact matching. To obtain an exact matching, we will employ the method of seeded graph matching, as developed in previous works \cite{BCL19, MX20, MWXY22+}.

In the rest of this section, we will describe in detail our iterative algorithm, which consists of a few steps including preprocessing (see Section~\ref{sec:preprocessing}), initialization (see Section~\ref{sec:initialization}), iteration (see Section~\ref{sec:iteration}), finishing (see Section~\ref{sec:ending-iteration}) and seeded graph matching (see Section~\ref{sec:seeded-matching}). We formally present our algorithm in Section~\ref{sec:formal-algorithm}. In Section~\ref{sec:runtime-analysis} we analyze the time complexity of the algorithm.

\subsection{Preprocessing}\label{sec:preprocessing}

Similar to \cite{DL22+}, we make appropriate preprocessing on random graphs such that we only need to consider graphs with directed edges. We first make a technical assumption that we only need to consider the case when $\rho$ is a sufficiently small constant, which can be easily achieved by keeping each edge independently with a sufficiently small constant probability. 

Now, we define $\overrightarrow G$ from $G$. For any $u \neq v \in V$, we do the following: 
\begin{itemize}
    \item if $(u,v) \in E(G)$, then independently among all such $(u,v)$:
    \begin{align*}
        & \mbox{with probability } \frac{1}{2} - \frac{q}{4}, \mbox{ set } \overrightarrow{(u,v)} \in \overrightarrow{G}, \overrightarrow{(v,u)} \not \in \overrightarrow{G}\,, \\
        & \mbox{with probability } \frac{1}{2} - \frac{q}{4}, \mbox{ set } \overrightarrow{(u,v)} \not \in \overrightarrow{G}, \overrightarrow{(v,u)} \in \overrightarrow{G} \,,\\
        & \mbox{with probability }  \frac{q}{4}, \mbox{ set } \overrightarrow{(u,v)} \in \overrightarrow{G}, \overrightarrow{(v,u)}  \in \overrightarrow{G} \,, \\
        & \mbox{with probability }  \frac{q}{4}, \mbox{ set } \overrightarrow{(u,v)} \not \in \overrightarrow{G}, \overrightarrow{(v,u)} \not \in \overrightarrow{G} \,;
    \end{align*}
    \item if $(u,v) \not \in E(G)$, then set $\overrightarrow{(u,v)} \not \in \overrightarrow{G}, \overrightarrow{(v,u)} \not \in \overrightarrow{G}$.
\end{itemize}
We define $\overrightarrow {\mathsf G}$ from $\mathsf G$ in the same manner such that $\overrightarrow G$ and $\overrightarrow {\mathsf G}$ are conditionally independent given $(G, \mathsf G)$. We continue to use the convention that $\overrightarrow G_{u, v} = \mathbf 1_{ \{ \overrightarrow{(u,v)} \in \overrightarrow G \}}$.
It is then straightforward to verify that $\{ \overrightarrow{G}_{u,v} : u \neq v \}$ and $\{ \mathsf{\overrightarrow{\mathsf{G}}_{u,v} : u \neq v}  \}$ are two families of i.i.d.\ Bernoulli random variables with parameter $\frac{q}{2}$. In addition, we have
\begin{align*}
    \mathbb{E}[ \overrightarrow{G}_{u,v} {\overrightarrow{\mathsf{G}}_{\pi(u),\pi(v)}} ] = \mathbb{E}[ \overrightarrow{G}_{u,v} {\overrightarrow{\mathsf{G}}_{\pi(v),\pi(u)}} ]= \frac{q(q+ \rho(1-q))}{4}\,.
\end{align*}
Thus, $\overrightarrow{G},\overrightarrow{\mathsf{G}}$ are edge-correlated directed graphs, denoted as $\overrightarrow{\mathcal{G}}(n,\Hat{q},\Hat{\rho})$, such that $\Hat{q} =\frac{q}{2}$ and $\Hat{\rho}= \frac{1-q}{2-q} \rho$. Also note that $\Hat{q} \geq n^{-\alpha+o(1)}$ and $\Hat{\rho} \in [ \frac{\rho}{3} ,\frac{\rho}{2})$ since $q \leq 1/2$. From now on we will work on the directed graph  $(\overrightarrow{G}, \overrightarrow{\mathsf{G}})$.

\subsection{Initialization}\label{sec:initialization}

For a pair of standard bivariate normal variables $(X, Y)$ with correlation $u$, we define $\phi: [-1, 1] \mapsto [0, 1]$ by (below the number 10 is somewhat arbitrarily chosen)
\begin{align}
    \phi(u) = \mathbb{P} (|X| \geq 10 ,|Y| \geq 10) \,.
    \label{equ-def-func-phi}
\end{align}
In addition, we define 
\begin{equation}{\label{eq-def-iota-ub-lb}}
    \iota_{\mathrm{ub}} = \sup_{x \in (0,1]} \Big\{ \frac{\phi(x)-\phi(0)}{x^2} \Big\} \mbox{ and } \iota_{\mathrm{lb}} = \inf_{x \in (0,1]} \Big\{ \frac{\phi(x)-\phi(0)}{x^2} \Big\} \,.
\end{equation}
From the definition we know $\phi$ is strictly increasing and by \cite[Claims 2.6 and 2.8]{DL22+} we have $\phi^{\prime} (0) =0, \phi^{\prime \prime} (0) >0$, and thus both $\iota_{\mathrm{ub}}$ and $\iota_{\mathrm{lb}}$ are positive and finite. Also we write $\mathfrak{a}=\phi(1) = \mathbb{P}(|X| \geq 10)$. Recalling in Subsection~\ref{sec:preprocessing} it was shown that $\rho$ can be assumed to be a sufficiently small constant, from now on we will assume that
\begin{equation}{\label{eq-assumetion-rho}}
    \Hat{\rho} \leq \rho \leq \min \big\{ \mathfrak{a} - \mathfrak{a}^2 , \iota_{\mathrm{ub}}^{-1} , \tfrac{1}{10} \big\} \,. 
\end{equation}

Let $\kappa = \kappa(\Hat{\rho})$ be a sufficiently large constant depending on $\Hat{\rho}$ whose exact value will be decided later in \eqref{eq-kappa-choice}. Set $K_0 = \kappa$. We then arbitrarily choose a sequence $A = (u_1, u_2, \ldots , u_{K_0})$ where $u_i$'s are distinct vertices in $V$, and we list all the sequences of length $K_0$ with distinct elements in $\mathsf{V}$  as $\mathsf{A}_1, \mathsf{A}_2, \ldots , \mathsf{A}_{\mathtt{M}}$ where $\mathtt{M} = \mathtt{M}(n,\Hat{\rho},\Hat{q}) = \prod_{i=0}^{K_0-1}(n-i) $. As in \cite{DL22+}, for each $1 \leq \mathtt{m} \leq \mathtt{M}$, we will run a procedure of initialization and iteration and clearly for one of them (although \emph{a priori} we do not know which one it is) we are running the algorithm as if we have $K_0$ true pairs as seeds. For convenience, when describing the initialization and the iteration we will drop $\mathtt m$ from notations, but we emphasize that this procedure will be applied to each $\mathsf A_{\mathtt m}$. Having this clarified, we take a fixed $\mathtt m$ and denote $\mathsf{A}_{\mathtt{m}} = ( \mathsf{u}_1, \ldots, \mathsf{u}_{K_0} )$. In what follows, we abuse the notation and write $V \setminus A$ when regarding $A$ as a set (similarly for $\mathsf A_{\mathtt m}$). 

We next describe our initialization procedure. As discussed earlier, to the contrary of the case for Wigner matrices, we have to investigate the neighborhood around a seed up to a certain depth in order to get information for a large number of vertices. To this end, we choose an integer $\chi \leq \frac{1}{1-\alpha}$ as the depth such that 
\begin{equation} \label{eq-def-chi}
    (n \Hat{q})^{\chi} = o\big( n e^{-(\log \log n)^{100}} \big) \mbox{ and } (n \hat{q})^{\chi+1} = \Omega \big(n e^{-(\log \log n)^{100}} \big)
\end{equation} 
(this is possible since $\alpha<1$). We choose such $\chi$ since on the one hand we wish to see a large number of vertices near the seed and on the other hand we want to reveal only a vanishing fraction of the edges. Now for $1 \leq k \leq K_0$, define the seeds
\begin{align}
    \aleph^{(0)}_k = \{ u_{k} \}, \Upsilon^{(0)}_k = \{ \mathsf{u}_{k} \} \,,
    \mbox{ and }  \vartheta_0 = \varsigma_0 = 1 / n \,.    \label{equ-def-initial-aleph-Upsilon}
\end{align}
Then for $1 \leq a \leq \chi$, we iteratively define the $a$-neighborhood of each seed by
\begin{equation}
\begin{aligned}
    & \aleph^{(a)}_k = \Big\{ v \in V \setminus \big( \cup_{1 \leq k \leq K_0, 0 \leq j \leq a-1} \aleph^{(j)}_k \big) :  \overrightarrow{G}_{v,u} = 1 \mbox{ for some } u \in \aleph^{(a-1)}_k \Big\} \,, \\
    & \Upsilon^{(a)}_k = \Big\{ \mathsf{v} \in \mathsf{V} \setminus \big( \cup_{1 \leq k \leq K_0, 0 \leq j \leq a-1} \Upsilon^{(j)}_k \big) :  \overrightarrow{\mathsf{G}}_{\mathsf{v},\mathsf{u}} = 1 \mbox{ for some } \mathsf{u} \in \Upsilon^{(a-1)}_k \Big\} \,.
    \label{equ-def-iter-aleph-Upsilon}
\end{aligned}
\end{equation} 
Also, for $1 \leq a \leq \chi$ we iteratively define
\begin{equation}
\begin{aligned}
    & \vartheta_a = \mathbb{P}( X \geq 1 ) \mbox { where } X \sim \mathrm{Bin}(\vartheta_{a-1}n,\Hat{q})  \,, \\ 
    & \varsigma_{a} = \mathbb{P} ( X \geq 1,Y \geq 1 ) \mbox{ where } (X,Y) \sim \mathrm{CorBin}(\vartheta_{a-1}n,\vartheta_{a-1}n, \Hat{q}; \varsigma_{a-1}n, \Hat{\rho})  \,. \label{equ-def-iter-vartheta-varsigma}
\end{aligned}
\end{equation}
We will show in Subsection~\ref{sec:priliminary-events} that actually we have 
$$
    |\aleph^{(a)}_k|/n, |\Upsilon^{(a)}_k|/n \approx \vartheta_a \mbox{ and } |\aleph^{(a)}_k| \cap \Upsilon^{(a)}_k|/n \approx \varsigma_a \,.
$$
Let $\mathtt{d}_{\chi}=\mathtt{d}_{\chi}(n,\Hat{q})$ be the minimal integer such that $\mathbb{P}( \mathrm{Bin}(n \vartheta_{\chi}, \Hat{q}) \geq \mathtt{d}_{\chi} ) < \frac{1}{2}$, and set 
\begin{equation}
\begin{aligned}
    & \Gamma^{(0)}_k = \aleph^{(\chi+1)}_k = \Big\{ v \in V \setminus \big( \cup_{1 \leq k \leq K_0, 0 \leq j \leq \chi} \aleph^{(j)}_k \big) : \sum_{u \in \aleph^{(\chi)}_k} \overrightarrow{G}_{v,u} \geq \mathtt{d}_{\chi} \Big\} \,,  \\
    & \Pi^{(0)}_k = \Upsilon^{(\chi+1)}_k = \Big\{ \mathsf{v} \in \mathsf{V} \setminus \big( \cup_{1 \leq k \leq K_0, 0 \leq j \leq \chi} \Upsilon^{(j)}_k \big) : \sum_{\mathsf{u} \in \Upsilon^{(\chi)}_k}  \overrightarrow{\mathsf{G}}_{\mathsf{v},\mathsf{u}} \geq \mathtt{d}_{\chi} \Big\} \,.
    \label{equ-def-initial-set}
\end{aligned}
\end{equation}
And we further define
\begin{equation}
\begin{aligned}
    & \vartheta = \vartheta_{\chi+1} = \mathbb{P}( X \geq \mathtt{d}_{\chi}) \mbox{ where } X \sim \mathrm{Bin}(\vartheta_{\chi}n,\Hat{q})  \,,  \\ 
    & \varsigma = \varsigma_{\chi+1} = \mathbb{P} ( X,Y \geq \mathtt{d}_{\chi}) \mbox{ where } X,Y \sim \mathrm{CorBin} (\vartheta_{\chi}n, \vartheta_{\chi}n, \Hat{q}; \varsigma_{\chi}n, \Hat{\rho})\,. \label{equ-def-vartheta-varsigma}
\end{aligned}
\end{equation}
We may then choose $K_0 = \kappa$ sufficiently large such that $K_0 \geq \frac{10^{34} \iota_{\mathrm{ub}}^2 \Hat{\rho}^{-20}}{\iota_{\mathrm{lb}}^2 (\mathfrak{a}- \mathfrak{a}^2)^2}$ and
\begin{equation}
    \label{eq-kappa-choice}
    \frac{ \log ( K_0 \iota_{\mathrm{lb}}^2 (\mathfrak{a}-\mathfrak{a}^2)^2 \Hat{\rho}^{20} / 10^{30} \iota_{\mathrm{ub}}^2 ) }{ \log ( { K_0 \iota_{\mathrm{lb}}^4 \Hat{\rho}^{24} \varepsilon_0^2 } / 16 \cdot 10^{30} \iota_{\mathrm{ub}}^2 ) } \leq 1.01 \mbox{ where } \varepsilon_{0} = \frac{\varsigma-\vartheta^2} {2(\vartheta-\vartheta^2)}  \,.
\end{equation} 
In addition, we define $\Phi^{(0)}, \Psi^{(0)}$ to be $K_0\!*\!K_0$ matrices by
\begin{align}
    \Phi^{(0)} = \mathrm{I} \mbox{ and }
    \Psi^{(0)} = \frac{\varsigma-\vartheta^2} {2(\vartheta-\vartheta^2)} \mathrm{I}\,,
    \label{equ-initial-matrix}
\end{align}
and in the iterative steps we will also construct $K_t,\varepsilon_t,\Gamma^{(t)}_k,\Pi^{(t)}_k$ and $\Phi^{(t)},\Psi^{(t)}$ for $t \geq 1$. 
Similarly as in \cite{DL22+}, the matrices $\Phi^{(t)}$ and $\Psi^{(t)}$ are supposed to approximate the cardinalities of the sets $\big\{ \Gamma^{(t)}_k, \Pi^{(t)}_k \big\}$ in the following sense. Write
\begin{equation}{\label{eq-def-mathfrak-a-t}}
    \mathfrak{a}_t=
    \begin{cases}
        \mathfrak{a}, & t \geq 1 \,; \\
        \vartheta , & t=0 \,.
    \end{cases}
\end{equation}
Then somewhat informally, we expect that 
\begin{align}
    & \frac{1}{n}|\Gamma^{(t)}_i|, \frac{1}{n}|\Pi^{(t)}_i| \approx \mathfrak{a}_t \,, \label{eq-intuition-concentration-0} \\
    &\frac{ \frac{1}{n}|\Gamma^{(t)}_i \cap \Gamma^{(t)}_j| - \frac{1}{n}|\Gamma^{(t)}_i| - \frac{1}{n} |\Gamma^{(t)}_j| + \mathfrak{a}_t^2 }{ \mathfrak{a}_t-\mathfrak{a}_t^2 } \approx \frac{ \frac{1}{n}|\Gamma^{(t)}_i \cap \Gamma^{(t)}_j| - \mathfrak{a}_t^2 }{ \mathfrak{a}_t-\mathfrak{a}_t^2 } \approx \Phi^{(t)}(i,j) \,, \label{eq-intuition-concentration-1} \\
    &\frac{ \frac{1}{n}|\Pi^{(t)}_i \cap \Pi^{(t)}_j| - \frac{1}{n}|\Pi^{(t)}_i| - \frac{1}{n} |\Pi^{(t)}_j| + \mathfrak{a}_t^2 }{ \mathfrak{a}_t-\mathfrak{a}_t^2 } \approx \frac{ \frac{1}{n} |\Pi^{(t)}_i \cap \Pi^{(t)}_j| - \mathfrak{a}_t^2 }{ \mathfrak{a}_t-\mathfrak{a}_t^2 } \approx \Phi^{(t)}(i,j) \,, \label{eq-intuition-concentration-2} \\
    &\frac{ \frac{1}{n}|\Gamma^{(t)}_i \cap \Pi^{(t)}_j| - \frac{1}{n}|\Gamma^{(t)}_i| - \frac{1}{n} |\Pi^{(t)}_j| + \mathfrak{a}_t^2 }{ \mathfrak{a}_t-\mathfrak{a}_t^2 } \approx \frac{ \frac{1}{n}|\Gamma^{(t)}_i \cap \Pi^{(t)}_j| - \mathfrak{a}_t^2 }{ \mathfrak{a}_t-\mathfrak{a}_t^2 } \approx \Psi^{(t)}(i,j) \,. \label{eq-intuition-concentration-3}
\end{align}
As in \cite[Lemma 2.1]{DL22+}, in order to facilitate our analysis later we will also need an important property on the eigenvalues of $\Phi^{(t)}$ and $\Psi^{(t)}$:
\begin{align}
    &\Phi^{(t)} \mbox{ has at least } \frac{3}{4}K_t \mbox{ eigenvalues between $0.9$ and $1.1$} \,, \label{eq-intuition-spectral-1} \\
    \mbox{and } &\Psi^{(t)} \mbox{ has at least } \frac{3}{4}K_t \mbox{ eigenvalues between $ 0.9 \varepsilon_t$ and $1.1 \varepsilon_t$} \,. \label{eq-intuition-spectral-2}
\end{align}
We will show in Subsection~\ref{sec:priliminary-events} that \eqref{eq-intuition-concentration-0}--\eqref{eq-intuition-spectral-2} are satisfied for $t=0$. The main challenge is to construct $(\Gamma_k^{(t+1)}, \Pi_k^{(t+1)})$ and $\Phi^{(t+1)}, \Psi^{(t+1)}$ such that \eqref{eq-intuition-concentration-0}--\eqref{eq-intuition-spectral-2} hold for $t+1$, \emph{under the inductive assumption} that \eqref{eq-intuition-concentration-0}--\eqref{eq-intuition-spectral-2} hold for $t$.
We conclude this subsection with some bounds on $(\vartheta_k, \varsigma_k)$.
\begin{Lemma}
\label{lemma-property-vartheta-varsigma}
$\vartheta_k,\varsigma_k = \Theta( n^{-1} (n\Hat{q})^{k} )$ for $0 \leq k \leq \chi$ and $\vartheta_{\chi+1},\varsigma_{\chi+1} = \Omega( e^{- (\log \log n)^{100}} )$. Also, we have $\varsigma_k -\vartheta_k^2 = \Theta(\vartheta_k)$ for $0 \leq k \leq \chi+1$. In addition, we have either $\vartheta_{\chi+1} = \Theta(1)$ or $\vartheta_{\chi} \leq n^{-\alpha+o(1)}$.
\end{Lemma}
\begin{proof}
We prove the first claim by induction. The claim is trivial for $k=0$. Now suppose the claim holds up to some $k \leq \chi-1$. Using \eqref{equ-def-iter-vartheta-varsigma} and Poisson approximation (note that when $k \leq \chi-1$ we have $n\Hat{q} \vartheta_{k} = \Theta(n^{-1}(n\Hat{q})^{k+1}) = o(1)$)
\begin{align*}
    \vartheta_{k+1} = \Theta((\vartheta_k n \Hat{q}) = \Theta(n^{-1} (n \Hat{q})^{k+1}) \mbox{ and } \vartheta_{k+1} \geq \varsigma_{k+1} \geq \Theta( \Hat{\rho} \varsigma_k n \Hat{q} ) = \Theta(\vartheta_{k+1}) \,,
\end{align*}
which verifies the claim for $k+1$ and thus verifies the first claim (for $ 0 \leq k\leq \chi$). If $\vartheta_{\chi} n \Hat{q} = \Theta ( (n \Hat{q})^{\chi + 1} ) \ll 1$, we have $\mathtt{d}_{\chi} = 1$ and thus $\vartheta_{\chi+1}, \varsigma_{\chi+1} = \Theta( (n \Hat{q})^{\chi+1} ) = \Omega( e^{-(\log \log n)^{100}} )$; if $\vartheta_{\chi} n \Hat{q} = \Omega(1)$, we have $\varsigma_{\chi} n \Hat{q} = \Omega(1)$ and thus by the choice of $\mathtt{d}_{\chi}$ we have $\vartheta_{\chi+1},\varsigma_{\chi+1} = \Theta(1)$ and $\varsigma_{\chi+1}-\vartheta_{\chi+1}^2 = \Theta(1)$ using Poisson approximation. Thus, we have $\varsigma_k - \vartheta_k^2 = \Theta(\vartheta_k)$ for $k = \chi + 1$ (note that the case for $1 \leq k \leq \chi$ can be checked in a straightforward manner). In addition,  if $\vartheta_{\chi} = n^{-\alpha+\epsilon+o(1)}$ for some arbitrarily small but fixed $\epsilon>0$, then $n\Hat{q} \vartheta_{\chi} \gg 1$ and thus $\vartheta_{\chi+1} = \Theta(1)$. This completes the proof of the lemma.
\end{proof}

\subsection{Iteration}\label{sec:iteration}
 
We reiterate that in this subsection we are describing the iteration for a fixed $1 \leq \mathtt m \leq \mathtt M$ and eventually this iterative procedure will be applied to each $\mathtt m$. Define
\begin{align}
    K_{t+1} = \frac{1}{ \varkappa } K_{t}^{2} \mbox{ where } \varkappa = \varkappa(\Hat{\rho}) =  \frac{10^{30} \iota_{\mathrm{ub}}^2 \Hat{\rho}^{-20}}{\iota_{\mathrm{lb}}^2 (\mathfrak{a}- \mathfrak{a}^2)^2}
    \label{equ-def-iter-K}
\end{align}
for $t \geq 0$. Since we have assumed $K_0 \geq 10^4 \varkappa$, we can then prove by induction that
\begin{equation}{\label{eq-increasing-dimension}}
    10^{30} \Hat{\rho}^{20} (\mathfrak{a}-\mathfrak{a}^2)^2 K_t^2 \geq K_{t+1} \geq 10^4 K_t \,.
\end{equation}
We now suppose that $( \Gamma^{(s)}_k,\Pi^{(s)}_k)_{1 \leq k \leq K_s}$ and $\Phi^{(s)},\Psi^{(s)}$ have been constructed for $s \leq t$ (which will be implemented inductively via \eqref{equ-def-iter-sets} as we describe next). Recall that we are working under the assumption that \eqref{eq-intuition-concentration-0}--\eqref{eq-intuition-spectral-2} hold for $s \leq t$. For $v \in V$ and $\mathsf{v} \in \mathsf{V} $, define $D^{(t)}_v, \mathsf{D}^{(t)}_{\mathsf{v}} \in \mathbb{R}^{K_t}$ to be the ``normalized degrees'' of $v$ in $\Gamma^{(t)}_k$ and of $\mathsf{v}$ in $\Pi^{(t)}_k$ as follows:
\begin{equation}
    \begin{aligned}
        D^{(t)}_v(k) &= \frac{1}{\sqrt{( \mathfrak{a}_t - \mathfrak{a}_t^2 )n \Hat{q}(1-\Hat{q})}} \sum_{u \in V}( \mathbf{1}_{u \in \Gamma^{(t)}_k} - \mathfrak{a}_t) (\overrightarrow{G}_{v,u}-\Hat{q})\,, \\
        \mathsf{D}^{(t)}_{\mathsf{v}}(k) &= \frac{1}{\sqrt{(\mathfrak{a}_t-\mathfrak{a}_t^2)n \Hat{q}(1-\Hat{q})}} \sum_{\mathsf{u} \in \mathsf{V} }(\mathbf{1}_{\mathsf{u} \in \Pi^{(t)}_k} -\mathfrak{a}_t) (\overrightarrow{\mathsf{G}}_{\mathsf{v,u}}- \Hat{q})\,.
        \label{equ-def-degree}
    \end{aligned}
\end{equation}
Recalling \eqref{eq-def-mathfrak-a-t}, we note that there is a difference between the definition \eqref{equ-def-degree} for $t = 0$ and $t\geq 1$; this is because $\Gamma^{(0)}_k$ and $\Pi^{(0)}_k$ may only contain a vanishing fraction of vertices. We also point out that similar to \cite{DL22+}, in the above definition we used the “centered” version of $\mathbf{1}_{u \in \Gamma^{(t)}_k}$ and $\mathbf{1}_{\mathsf u \in \Pi^{(t)}_k}$ since \eqref{eq-intuition-concentration-0} suggests that intuitively each vertex $u$ (respectively, $\mathsf u$) has probability approximately $\mathfrak{a}_t$ to belong to $\Gamma^{(t)}_k$ (respectively, $\Pi^{(t)}_k$); such centering will be useful for our proof later as it leads to additional cancellation.

Assuming Lemma~\ref{lemma_matrix_eigenvalue}, we can then write $\Phi^{(t)}$ and $\Psi^{(t)}$ as their spectral decompositions:
\begin{equation}
    \label{eq-spectral-decomposition}
    \Phi^{(t)}=\sum^{K_t}_{i=1} \lambda^{(t)}_i
    \left({\nu^{(t)}_i} \right)^{*} \left(\nu^{(t)}_i \right) \mbox{ and }
    \Psi^{(t)}= \sum_{i=1}^{K_t} \mu^{(t)}_i \left({\xi^{(t)}_i} \right)^{*} \left(\xi^{(t)}_i \right)
\end{equation}
where 
\begin{equation}
    \label{eq-lambda-mu-bound}
    \lambda^{(t)}_i \in (0.9,1.1), \mu^{(t)}_i \in ( 0.9 \varepsilon_t, 1.1 \varepsilon_t ) \mbox{ for } 1 \leq i \leq \frac{3K_t}{4}
\end{equation} and $\nu_i^{(t)},\xi_i^{(t)}$ are the unit eigenvectors with respect to $\lambda_i^{(t)},\mu_i^{(t)}$ respectively. 
Next, for $s, t$ we define
$\mathrm{M}_{\Gamma}^{(t,s)},\mathrm{M}_{\Pi}^{(t,s)},\mathrm{P}_{\Gamma,\Pi}^{(t,s)}$ to be $K_t\!*\!K_s$ matrices as follows:
\begin{equation}
    \begin{aligned}
        \mathrm{M}_{\Gamma}^{(t,s)}(i,j) & = \frac{  |\Gamma^{(t)}_i \cap \Gamma^{(s)}_j | - \mathfrak{a}_s |\Gamma^{(t)}_i | - \mathfrak{a}_t |\Gamma^{(s)}_j | + \mathfrak{a}_s \mathfrak{a}_t n }{ \sqrt{(\mathfrak{a}_s-\mathfrak{a}_s^2) (\mathfrak{a}_t-\mathfrak{a}_t^2)} n} \,,\\
        \mathrm{M}_{\Pi}^{(t,s)} (i,j)  & = \frac{ |\Pi^{(t)}_i \cap \Pi^{(s)}_j| - \mathfrak{a}_s |\Pi^{(t)}_i| - \mathfrak{a}_t |\Pi^{(s)}_j| + \mathfrak{a}_s \mathfrak{a}_t n} { \sqrt{(\mathfrak{a}_s-\mathfrak{a}_s^2) (\mathfrak{a}_t-\mathfrak{a}_t^2)} n} \,, \\
        \mathrm{P}_{\Gamma,\Pi}^{(t,s)}(i,j) & = \frac{ |\pi(\Gamma^{(t)}_i) \cap \Pi^{(s)}_j| - \mathfrak{a}_s |\Gamma^{(t)}_i| - \mathfrak{a}_t |\Pi^{(s)}_j| + \mathfrak{a}_s \mathfrak{a}_t n }{ \sqrt{(\mathfrak{a}_s-\mathfrak{a}_s^2) (\mathfrak{a}_t-\mathfrak{a}_t^2)} n} \,.
        \label{equ_martix_M_P}
    \end{aligned}
\end{equation}
These matrices actually represent the covariance matrices for random vectors of the form $D^{(t)}_v$ and $\mathsf{D}^{(s)}_{\pi(v)}$. To get a rough intuition of this, we (formally incorrectly) regard $D_v^{(s)}$ as a linear combination of $\{G_{u,v}\}$ with deterministic coefficients (and the same applies to $\mathsf D^{(s)}_{\mathsf v}$). Then we can see that for all $v \in V$, the ``correlation'' between $D^{(t)}_v(i)$ and $D^{(s)}_v(j)$ equals
$$
\frac{1}{\sqrt{(\mathfrak{a}_t-\mathfrak{a}_t^2) (\mathfrak{a}_s-\mathfrak{a}_s^2)} n} \sum_{u \in V \setminus A} ( \mathbf{1}_{u \in \Gamma^{(t)}_i} - \mathfrak{a}_t ) ( \mathbf{1}_{u \in \Gamma^{(s)}_j} - \mathfrak{a}_s ) = \mathrm{M}_{\Gamma}^{(t,s)} (i,j) \,.
$$
This justifies our definition of $\mathrm{M}_{\Gamma}^{(t,s)}$ which aims to record the correlation between $D^{(t)}_v$ and $D^{(s)}_v$.
Similarly under the same simplification we have $\mathrm{M}_{\Pi}^{(t,s)}$ (respectively, $\Hat{\rho} \mathrm{P}_{\Gamma,\Pi}^{(t,s)}$) is the correlation matrix between $\mathsf{D}^{(t)}_{\mathsf v}$ and $\mathsf{D}^{(s)}_{\mathsf v}$ (respectively, between $D^{(t)}_v$ and $\mathsf{D}^{(s)}_{\pi(v)}$). In addition, from \eqref{eq-intuition-concentration-0}--\eqref{eq-intuition-concentration-3} we expect that $\mathrm{M}_{\Gamma}^{(t,t)}, \mathrm{M}_{\Pi}^{(t,t)} \approx \Phi^{(t)}$ and $\mathrm{P}_{\Gamma,\Pi}^{(t,t)} \approx \Psi^{(t)}$.
Note that $\mathrm{M}_{\Gamma},\mathrm{M}_{\Pi}$ are accessible by the algorithm but $\mathrm{P}_{\Gamma,\Pi}$ is not (since it relies on the latent matching). We further define two linear subspaces as follows:
\begin{equation}
    \begin{aligned}
        \mathrm{W}^{(t)} & \overset{\triangle}{=}  \big \{ x \in \mathbb{R}^{K_t} :
        x \mathrm{M}_{\Gamma}^{(t,s)} = 0, 
        x \mathrm{M}_{\Pi}^{(t,s)} = 0,  \mbox{ for all } s<t \big\} \,, \\
        \mathrm{V}^{(t)} & \overset{\triangle}{=}  \mathrm{span} \big\{ \nu^{(t)}_1, \nu^{(t)}_2, \ldots, \nu^{(t)}_{\frac{3}{4}K_t} \big\}  \cap \mathrm{span} \big\{ \xi^{(t)}_1, \xi^{(t)}_2, \ldots, \xi^{(t)}_{\frac{3}{4}K_t}  \big\} \cap \mathrm{W}^{(t)} \,.
        \label{equ-linear-space}
    \end{aligned}
\end{equation}
We refer to \cite[Remark 3.3]{DL22+} for underlying reasons of the definition above. Note that the number of linear constraints posed on $\mathrm{W}^{(t)}$ are at most $ 2\sum_{i=1}^{t}K_{i-1}$. So
\begin{equation*}
    \dim ( \mathrm{V}^{(t)} ) \geq \frac{3}{4}K_t + \frac{3}{4} K_t + \dim (\mathrm{W}_t) - 2K_t \geq \frac{1}{2} K_t - 2\sum_{i=1}^{t}K_{i-1} \overset{\eqref{eq-increasing-dimension}}{\geq} 0.49 K_t \,. 
\end{equation*}
As proved in \cite[(2.10) and (2.11)]{DL22+}, we can choose $\eta^{(t)}_1,\eta^{(t)}_2,\ldots,\eta^{(t)}_{\frac{1}{12}K_t}$ from $\mathrm{V}^{(t)}$ such that
\begin{align}
    & \eta^{(t)}_i \mathrm{M}_{\Gamma}^{(t,t)} \big(\eta^{(t)}_j \big)^{*}
    =\eta^{(t)}_i \mathrm{M}_{\Pi}^{(t,t)} \big( \eta^{(t)}_j \big)^{*}
    =\eta^{(t)}_i \Psi^{(t)} \big(\eta^{(t)}_j \big)^{*} =0 \,, \label{equ-vector-orthogonal} \\
    & \eta^{(t)}_i \Phi^{(t)} \big(\eta^{(t)}_i \big)^{*} =1, \quad  2 \varepsilon_t \geq   \eta^{(t)}_i \Psi^{(t)} \big(\eta^{(t)}_i \big)^{*} \geq 0.5 \varepsilon_t\,. \label{equ-vector-unit}
\end{align}
Furthermore, we must have $\big\| \eta^{(t)}_i \big\|^2 \in (\frac{1}{2},2)$. As in \cite{DL22+}, we will project the degrees $D^{(t)}_v, \mathsf{D}^{(t)}_{\mathsf v}$ to a set of carefully chosen directions in the space spanned by all $\eta_i$'s. These directions are defined as follows: 
we sample $\beta^{(t)}_k(j)$ as i.i.d.\ uniform variables on $\{-1, 1\}$. By \cite[Proposition 2.4]{DL22+}, these $\beta^{(t)}_k(j)$'s satisfy \cite[(2.21)--(2.24)]{DL22+} with probability at least 0.5. As in \cite{DL22+}, we will keep resampling until these requirements are satisfied. Define
\begin{align}
    \sigma_k^{(t)}= \sqrt{\frac{12}{K_t}}   \sum_{j=1}^{\frac{1}{12}K_t} \beta^{(t)}_k(j)  \eta_j^{(t)} \mbox{ for } k=1,2,\ldots,K_{t+1}\,.
    \label{equ-def-sigma}
\end{align}
We sample i.i.d.\ standard normal variables $\{ W^{(t)}_v(i), \mathsf{W}^{(t)}_{\mathsf{v}}(i): 1 \leq i \leq \frac{K_t}{12} \}$ and complete our iteration by setting
\begin{equation}
    \begin{aligned}
        &\Gamma^{(t+1)}_k= \Big\{  v \in V : \frac{1}{\sqrt{2}}  \Big| \sqrt{\frac{12}{K_t}} \langle \beta^{(t)}_k, W^{(t)}_v \rangle + \langle \sigma^{(t)}_k,D^{(t)}_v \rangle \Big| \geq 10  \Big\}\,,  \\
        &\Pi^{(t+1)}_k= \Big\{  \mathsf{v} \in \mathsf{V} : \frac{1}{\sqrt{2}} \Big| \sqrt{\frac{12}{K_t}} \langle \beta^{(t)}_k, \mathsf{W}^{(t)}_v \rangle + \langle \sigma^{(t)}_k, \mathsf{D}^{(t)}_{\mathsf{v}} \rangle \Big| \geq 10 \Big\}\,.
        \label{equ-def-iter-sets}
    \end{aligned}
\end{equation}
In the above, we introduced a Gaussian smoothing $\{ W^{(t)}_v(i), \mathsf{W}^{(t)}_{\mathsf{v}}(i): 1 \leq i \leq K_t \}$. We believe this is not essential but provides technical convenience: on the one hand it probably simply reduces the efficiency of the algorithm since it weakens the signal, but on the other hand it facilitates the analysis since it brings the distribution closer to Gaussian. In addition, we have used the absolute value of a random variable instead of a random variable itself, with the purpose of introducing more symmetry as in \cite{DL22+} (e.g., to bound \eqref{equ-tail-quadratic-part-I} below). Recall that we have explained that $\mathrm{M}_{\Gamma}^{(t,t)}$ records the covariance matrix of $D^{(t)}_v$ for all $v \in V$. Thus, we expect that the correlation between $\langle \sigma^{(t)}_k, D^{(t)}_v \rangle $ and $\langle \sigma^{(t)}_l, D^{(t)}_v \rangle$ is approximately
$$
    \frac{12}{K_t} \sum_{i,j=1}^{\frac{1}{12}K_t} \beta^{(t)}_k(i) \beta^{(t)}_l(j) \eta^{(t)}_i \mathrm{M}_{\Gamma}^{(t,t)} \big(\eta^{(t)}_j\big)^{*} \overset{\eqref{equ-vector-orthogonal} ,\eqref{equ-vector-unit}}{=} \frac{12}{K_t} \sum_{i=1}^{\frac{1}{12}K_t} \beta^{(t)}_k(i) \beta^{(t)}_l(i) = \frac{12}{K_t} \big\langle \beta^{(t)}_k, \beta^{(t)}_l \big\rangle \,.
$$
In particular, the variance of each $\langle \sigma^{(t)}_i, D^{(t)}_v \rangle$ is approximately $1$.
Similarly, we can show the correlation between $\langle \sigma^{(t)}_i, \mathsf{D}^{(t)}_{\mathsf v} \rangle$ and $\langle \sigma^{(t)}_j, \mathsf{D}^{(t)}_{\mathsf v} \rangle$ is approximately $\frac{12}{K_t} \big\langle \beta^{(t)}_k, \beta^{(t)}_l \big\rangle$, and the correlation between $\langle \sigma^{(t)}_i, D^{(t)}_v \rangle$ and $\langle \sigma^{(t)}_j, \mathsf{D}^{(t)}_{\pi(v)} \rangle$ is approximately $\Hat{\rho} \cdot \frac{12}{K_t} \big\langle \Hat{\beta}^{(t)}_i , \Hat{\beta}^{(t)}_j \big\rangle$, where 
\begin{equation}\label{eq-def-hat-beta}
    \Hat{\beta}^{(t)}_k(j)= \Big(\eta^{(t)}_j \Psi^{(t)} \big( \eta^{(t)}_j \big)^{*}\Big)^{1/2} \cdot \beta^{(t)}_k(j) \,.
\end{equation}
(Here we also used that \eqref{eq-intuition-concentration-3} implies that $\mathrm{P}_{\Gamma,\Pi}^{(t,t)} \approx \Psi^{(t)}$). Recall our desire for \eqref{eq-intuition-concentration-0}--\eqref{eq-intuition-concentration-3} to hold for $t+1$.
Thus the signal contained in each pair at time $t+1$ is approximately
\begin{align}
    \varepsilon_{t+1} = \frac{1}{ (\mathfrak{a} - \mathfrak{a}^2) } \Big( \phi \Big( \frac{\Hat{\rho}}{2} \frac{12}{K_t} \sum_{j=1}^{\frac{K_t}{12}} \eta^{(t)}_j \Psi^{(t)} \big( \eta^{(t)}_j \big)^{*} \Big) - \phi(0) \Big)  \,.
    \label{equ-def-iter-varepsilon}
\end{align}
By \eqref{equ-vector-unit}, we have that
\begin{equation}
    \label{equ-epsilon-t-bound}
    \varepsilon_{t+1} \in \big[ \frac{ \iota_{\mathrm{lb}} \Hat{\rho}^{2} } { 4(\mathfrak{a}-\mathfrak{a}^2) }  (0.5 \varepsilon_t)^2 , \frac{\iota_{\mathrm{ub}} \Hat{\rho}^{2}} { 4(\mathfrak{a}-\mathfrak{a}^2)}  (2 \varepsilon_t)^2 \big]\,.
\end{equation}
Recalling \eqref{eq-assumetion-rho}, we have $\varepsilon_{t+1} \leq \varepsilon_t^2$, and thus (recall from \eqref{eq-kappa-choice} that $\varepsilon_0<\tfrac{1}{2}$)
\begin{equation}{\label{eq-decrease-varepsilon}}
    \varepsilon_{t+1} \leq \varepsilon_t \leq \ldots \leq \varepsilon_0 \leq \tfrac{1}{2} \,,
\end{equation}
which verifies our statement that the signal $\varepsilon_t$ in each pair is decreasing.
We can then finish the iteration by defining $\Phi^{(t+1)},\Psi^{(t+1)}$ to be $K_{t+1}\!*\!K_{t+1}$ matrices such that
\begin{equation}
    \begin{aligned}
    & \Phi^{(t+1)}(i,j) = (\mathfrak{a}-\mathfrak{a}^2)^{-1} \Big\{  \phi \Big( \frac{12}{K_t} \langle {\beta}^{(t)}_i,{\beta}^{(t)}_j \rangle \Big) - \mathfrak{a}^2 \Big\}  \,,  \\
    & \Psi^{(t+1)}(i,j)= (\mathfrak{a}-\mathfrak{a}^2)^{-1} \Big\{ \phi \Big( \frac{\Hat{\rho}}{2} \frac{12}{K_t} \langle \Hat{\beta}^{(t)}_i , \Hat{\beta}^{(t)}_j \rangle \Big) - \mathfrak{a}^2    \Big\}\,.
    \label{equ-def-iter-matrix}
    \end{aligned}
\end{equation}
Next, we state a lemma which then inductively justifies \eqref{eq-intuition-spectral-1} and \eqref{eq-intuition-spectral-2}.
\begin{Lemma}{\label{lemma_matrix_eigenvalue}}
Let $(\Phi^{(t)}, \Psi^{(t)})$ be initialized as in \eqref{equ-initial-matrix} and inductively defined as in \eqref{equ-def-iter-matrix}, also let $\varepsilon_t$ be initialized in \eqref{equ-def-iter-K} and iteratively defined in \eqref{equ-def-iter-varepsilon}. Then, $\Phi^{(t)}$ has $\frac{3}{4}K_t$ eigenvalues between $0.9$ and $1.1$, and $\Psi^{(t)}$ has $\frac{3}{4}K_t$ eigenvalues between $ 0.9 \varepsilon_t$ and $1.1 \varepsilon_t$.
\end{Lemma}
We note that the definition of \eqref{equ-def-iter-matrix} is identical to that of \cite[(2.15)]{DL22+} and thus Lemma~\ref{lemma_matrix_eigenvalue} is identical to \cite[Lemma 2.1]{DL22+}.

\subsection{Almost exact matching} \label{sec:ending-iteration}
In this subsection we describe how we get an almost exact matching once we accumulate enough signal along the iteration. To this end, define
\begin{equation}\label{eq-def-t-*}
    t^{*}=\min \{  t \geq 0: K_t \geq (\log n)^{2} \}\,.
\end{equation}
Obviously $K_{t^*} \leq (\log n)^4$. By \eqref{equ-def-iter-K}, we have $K_t = (K_0^{2^t})/ (\varkappa^{2^t-1})$ and as a result we have $t^* = O( \log \log \log n )$. 
In addition, recalling \eqref{equ-epsilon-t-bound} we have
\begin{align*}
    K_{t+1} \varepsilon_{t+1}^2 &\geq \frac{ \Hat{\rho}^{20} \iota_{\mathrm{lb}}^2 (\mathfrak{a}-\mathfrak{a}^2) }{ 10^{30} \iota_{\mathrm{ub}}^2 } K_t^2 \cdot \Big( \frac{ \iota_{\mathrm{lb}} \Hat{\rho}^2 }{ 16(\mathfrak{a} - \mathfrak{a}^2) } \varepsilon_t^2 \Big)^2 \\
    &= \frac{ \Hat{\rho}^{24} \iota_{\mathrm{lb}}^4 }{ 16^2 \cdot 10^{30} \iota_{\mathrm{ub}}^2 (\mathfrak{a}-\mathfrak{a}^2) } (K_t \varepsilon_t^2)^2 \,.
\end{align*}
Using the choice of $K_0=\kappa$ in \eqref{eq-kappa-choice} we see that the total signal $K_t \varepsilon_t^2$ is increasing in $t$.
We also have that
\begin{align}
    K_{t^*} \varepsilon_{t^*}^2 \geq \Big( \frac{ K_0 \iota_{\mathrm{lb}}^2 \Hat{\rho}^{4} \varepsilon^2_0 }{ 16 (\mathfrak{a}-\mathfrak{a}^2)^2 \varkappa } \Big)^{2^{t^*}} \overset{\eqref{eq-kappa-choice}}{\geq} \Big( \frac{K_0}{\varkappa} \Big)^{ 2^{t^*} / 1.01 } \geq K_{t_*}^{1/1.01} \geq (\log n)^{1.9} \,.
    \label{equ-K-Varepsilon-bound}
\end{align}
For each $1 \leq \mathtt m \leq \mathtt M$, we run the procedure of initialization and then run the iteration up to time $t^{*}$, and then we construct a permutation $\pi_{\mathtt{m}}$ (with respect to  $\mathsf{A}_{\mathtt{m}}$) as follows. For $A= ( u_1, \ldots, u_{K_0} )$ and $\mathsf{A}_{\mathtt{m}} = ( \mathsf{u}_1 , \ldots, \mathsf{u}_{K_0} )$, set $\pi_{\mathtt{m}}(u_j) = \mathsf{u}_j$ for $1 \leq j \leq K_0$. We set a prefixed ordering of $V \setminus A$ and $\mathsf{V} \setminus \mathsf{A}_{\mathtt{m}}$ as $V \setminus A = \{ v_1, \ldots, v_{n-K_0} \}$ and $\mathsf{V} \setminus \mathsf{A}_{\mathtt{m}} = \{ \mathsf{v}_1,\ldots, \mathsf{v}_{n-K_0} \} $, initialize the set $\mathrm{CAND}$ to be $\mathsf{V} \setminus \mathsf{A}_{\mathtt{m}}$, and initialize the sets $\mathrm{SUC}$, $\mathrm{PAIRED}$ and $\mathrm{FAIL}$ to be empty sets. The algorithm processes $v_k$ in the increasing order of $k$: for each $v_k$, we find the minimal $\mathsf{k}$ such that $\mathsf{v}_{\mathsf{k}} \in \mathrm{CAND}$ and
\begin{equation}
    \sum_{j=1}^{\frac{1}{12}K_{t^*}} \big(  W^{(t^*)}_{v_k}(j) + \langle \eta^{(t^*)}_j, D^{(t^*)}_{v_k} \rangle \big) \big(  \mathsf{W}^{(t^*)}_{\mathsf{v}_{\mathsf{k}}}(j) + \langle \eta^{(t^*)}_j, \mathsf{D}^{(t^*)}_{\mathsf{v}_{\mathsf{k}}} \rangle \big) \geq \frac{1}{100} K_{t^{*}} \varepsilon_{t^{*}}\,.
    \label{equ-def-matching-ver}
\end{equation}
We then define $\pi_{\mathtt{m}}(v_k) = \mathsf{v}_{\mathsf{k}}$, put $v_k$ into $\mathrm{SUC}$ and move $\mathsf{v}_{\mathsf{k}}$ from $\mathrm{CAND}$ to $\mathrm{PAIRED}$. If there is no $\mathsf{k}$ satisfying \eqref{equ-def-matching-ver}, we put $v_k$ into $\mathrm{FAIL}$. Having processed all vertices in $V \setminus A$, we pair the vertices in $\mathrm{FAIL}$ and the (remaining) vertices in $\mathrm{CAND}$ in an arbitrary but pre-fixed manner to get the matching $\pi_{\mathtt{m}}$.

We say a pair of sequences $A = (u_1,u_2,\ldots, u_{K_0})$ and $\mathsf A = (\mathsf{u}_1, \mathsf{u}_2, \ldots, \mathsf{u}_{K_0} )$ is a good pair if
\begin{equation}\label{eq-correct-seeds}
    \mathsf{u}_j = \pi(u_j) \mbox{ for } 1 \leq j \leq K_0\,.
\end{equation}
The success of our algorithm lies in the following proposition which states that starting from a good pair we have that $\pi_{\mathtt{m}}$ correctly recovers almost all vertices. 
\begin{Proposition}{\label{prop-almost-exact-matching}}
    For a pair $(A, \mathsf A)$, define $\pi(A, \mathsf A) = \pi_{\mathtt m}$ if $\mathsf A = \mathsf A_{\mathtt m}$. If $(A, \mathsf A)$ is a good pair, then with probability $1-o(1)$, we have 
   \begin{align*}
       | \{ v: {\pi}(A,\mathsf{A})(v) = \pi(v) \} | \geq (1-\frac{10}{\log n}) n \,.
   \end{align*}
\end{Proposition}

\subsection{From almost exact matching to exact matching}
\label{sec:seeded-matching}

In this subsection, we employ a seeded matching algorithm \cite{BCL19} (see also \cite{MX20, YXL21}) to enhance an almost exact matching (which we denote as $\Tilde {\pi}$ in what follows) to an exact matching.   Our matching algorithm is a simplified version of \cite[Algorithm 4]{BCL19}.
\begin{breakablealgorithm}
			\label{algo:seeded-matching}
			\caption{Seeded Matching Algorithm}
	\begin{algorithmic}[1]
		\STATE {\bf Input:} A triple $(G,\mathsf{G},\Tilde{\pi})$ where $(G,\mathsf{G}) \sim \mathcal{G}(n,q,\rho)$ and $\Tilde{\pi}$ agrees with $\pi$ on $1-o(1)$ fraction of vertices.
		\STATE For $u \in V(G), \mathsf{v \in V(G)}$, define their 1-neighborhood $N(u,\mathsf{v})= |\{ w \in V: u \sim w, \mathsf v \sim \Tilde{\pi}(w) \}| $. 
		\STATE Define $\Delta=\frac{\rho^2 n q}{100}$ and set $\Hat{\pi} = \Tilde{\pi}$.
		\STATE Repeat the following: if there exists a pair $u,\mathsf{v}$ such that $N(u,\mathsf{v}) \geq \Delta$ and $N(u,\Hat{\pi}(u))$, $  N(\Hat{\pi})^{-1}(\mathsf{v}),\mathsf{v}) < \frac{1}{10} \Delta$, then modify $\Hat{\pi}$ to map $u$ to $\mathsf{v}$ and map ${\Hat{\pi}}^{-1}(\mathsf{v})$ to $\Hat{\pi}(u)$; otherwise, move to Step 5.
		\STATE {\bf Output:} $\Hat{\pi}$.
	\end{algorithmic}
\end{breakablealgorithm}
At this point, we can run  Algorithm~\ref{algo:seeded-matching} for each $\pi_{\mathtt{m}}$ (which serves as the input $\Tilde{\pi}$), and obtain the corresponding refined matching $\Hat{\pi}_{\mathtt{m}}$ (which is the output $\Hat{\pi}$). By \cite[Lemma 4.2]{BCL19}  and Proposition~\ref{prop-almost-exact-matching}, we see that  $\Hat{\pi}_{\mathtt{m}}= \pi$  with probability $1-o(1)$ (note that \cite[Lemma 4.2]{BCL19} applies to an adversarially chosen input $\Tilde{\pi}$ as long as  $\Tilde{\pi}$ agrees with $\pi$ on $1-o(1)$ fraction of vertices). Finally, we set 
\begin{align}
    \Hat{\pi}_\diamond = \arg \max_{ \Hat{\pi}_{\mathtt{m}} } \Big \{  \sum_{(u,v) \in E(V)} G_{u,v} \mathsf{G}_{\Hat{\pi}_{\mathtt{m}}(u),\Hat{\pi}_{\mathtt{m}}(v)} \Big \}\,.
    \label{equ-def-final-pi-hat}
\end{align}
Combined with \cite[Theorem 4]{WXY21+}, it yields the following theorem.
\begin{Theorem}{\label{main-thm}}
   With probability $1-o(1)$, we have $\Hat{\pi}_\diamond = \pi$.
\end{Theorem}

\subsection{Formal description of the algorithm} \label{sec:formal-algorithm}
We are now ready to present our algorithm formally. 
\begin{breakablealgorithm}
			\label{algo:matching}
			\caption{Random Graph  Matching Algorithm}
	\begin{algorithmic}[1]
		\STATE Define $\overrightarrow{G}, \overrightarrow{\mathsf{G}}, \Hat{q}, \Hat{\rho}, A, \phi, \mathtt{M}, \iota_{\mathrm{lb}}, \iota_{\mathrm{ub}}, \mathfrak{a}, \varkappa, \kappa, \chi$ and $\Phi^{(0)}, \Psi^{(0)}$ as above.
		\STATE List all sequences with $K_0$ distinct elements in $\mathsf{V}$ by $\mathsf{A}_1, \mathsf{A}_2, \ldots, \mathsf{A}_{\mathtt{M}}$.
		\FOR{$\mathtt{m}=1, \ldots, \mathtt{M}$}
		\STATE Define $\aleph^{(a)}_k,\Upsilon^{(a)}_k$ for $0 \leq a \leq \chi, 1 \leq k \leq K_0$ as in \eqref{equ-def-initial-aleph-Upsilon} and \eqref{equ-def-iter-aleph-Upsilon}.
		\STATE Define $\Gamma^{(0)}_k, \Pi^{(0)}_k$  for $1 \leq k \leq K_0$ as in \eqref{equ-def-initial-set}.
		\STATE Define $\varepsilon_0, K_0$ as above.
		\STATE Set $\pi_{\mathtt{m}}(v_j) = \mathsf{v}_j$ where $v_j, \mathsf{v}_j$ are the $j$-th coordinate of $A, \mathsf{A}_{\mathtt{m}}$ respectively.
		\WHILE{ $K_t \leq (\log n)^2 $ }
		\STATE Calculate $K_{t+1}$ according to \eqref{equ-def-iter-K}.
		\STATE Calculate $\mathrm{M}^{(t,s)}_{\Gamma}, \mathrm{M}^{(t,s)}_{\Pi}$ for $0 \leq s \leq t$ according to \eqref{equ_martix_M_P}.
		\STATE Calculate the eigenvalues and eigenvectors of $\Phi^{(t)}, \Psi^{(t)}$, as in \eqref{eq-spectral-decomposition}.
		\STATE Define $\eta^{(t)}_1, \eta^{(t)}_2, \ldots, \eta^{(t)}_{ \frac{K_t}{12} }$ according to \eqref{equ-vector-orthogonal} and \eqref{equ-vector-unit}.
		\STATE Calculate $\varepsilon_{t+1}$ according to \eqref{equ-def-iter-varepsilon}.
		\STATE Sample random vectors $\beta^{(t)}_k$ for $1\leq k\leq K_{t+1}$ as described below \eqref{eq-lambda-mu-bound}.
		\STATE Define $\sigma^{(t)}_k$ for $1 \leq k \leq K_{t+1}$ according to \eqref{equ-def-sigma}.
		\STATE Define $\Phi^{(t+1)}, \Psi^{(t+1)}$ according to \eqref{equ-def-iter-matrix}.
		\STATE Define $\Gamma^{(t+1)}_k, \Pi^{(t+1)}_k$ for $1 \leq k \leq K_{t+1}$ according to \eqref{equ-def-iter-sets}.
		\ENDWHILE
		\STATE Suppose we stop at $t=t^{*}$.
		\STATE Define $\eta^{(t)}_1, \eta^{(t)}_2, \ldots, \eta^{(t)}_{ \frac{K_{t^{*}}}{12} }$ according to \eqref{equ-vector-orthogonal} and \eqref{equ-vector-unit}.
		\STATE List $V\setminus A$ and $\mathsf{V} \setminus \mathsf{A}_{\mathtt{m}} $ in a prefixed order $V\setminus A = \{ v_1,\ldots,v_{n-K_0} \}$ and $\mathsf{V} \setminus \mathsf{A}_{\mathtt{m}} = \{ \mathsf{v}_1,\ldots,\mathsf{v}_{n-K_0} \}$.
		\STATE Set $\mathrm{SUC}, \mathrm{PAIRED}, \mathrm{FAIL} = \emptyset$ and $\mathrm{CAND}= \mathsf{V} \setminus \mathsf{A}_{\mathtt{m}}$.
		\FOR{ $1 \leq k \leq n-K_0$ }
		\STATE Define $\textup{S}_u =0$.
		\FOR{ $1 \leq \mathsf{k} \leq n-K_0$ }
		\IF{ $\mathsf{v}_{\mathsf{k}} \in \mathrm{CAND}$ and $(v_k,\mathsf{v}_{\mathsf{k}})$ satisfies (\ref{equ-def-matching-ver})}
		\STATE Define $\pi_{\mathtt{m}}(v_k) = \mathsf{v}_{\mathsf{k}}$.
		\STATE Set $\textup{S}_u =1$.
		\STATE Put $v_k$ into $\mathrm{SUC}$ and move $\mathsf{v}_{\mathsf{k}}$ into $\mathrm{PAIRED}$.
		\ENDIF
		\ENDFOR
		\IF{ $\textup{S}_u =0$ }
		\STATE Put $v_k$ into $\mathrm{FAIL}$.
		\ENDIF
		\ENDFOR
		\STATE Complete $\pi_{\mathtt{m}}$ into an entire matching by mapping $\mathrm{FAIL}$ to $\mathrm{CAND}$ in an arbitrary and prefixed manner.
		\STATE Run Algorithm~\ref{algo:seeded-matching} with the input $(G,\mathsf{G},{\pi}_{\mathtt{m}})$ and denote the output as $\Hat{\pi}_{\mathtt{m}}$.
		\ENDFOR 
		\STATE Find $\Hat{\pi}_{\mathtt{m}^{*}}$ which maximizes $\sum_{(u,v) \in E(V)} G_{u,v} \mathsf{G}_{\pi(u), \pi(v)}$ among $\{ \Hat{\pi}_{\mathtt{m}} : 1 \leq \mathtt{m} \leq \mathtt{M}\}$.
		\RETURN $\Hat{\pi}_\diamond = \Hat{\pi}_{\mathtt{m}^{*}}$.
	\end{algorithmic}
\end{breakablealgorithm}

\subsection{Running time analysis}
\label{sec:runtime-analysis}
In this subsection, we analyze the running time for Algorithm~\ref{algo:matching}. 
\begin{Proposition}{\label{prop_time_complexity}}
     The running time for computing each $\pi_{\mathtt m}$ is $O(n^{3})$. Furthermore, the running time for Algorithm \ref{algo:matching}  is $O(n^{\kappa+3})$.
\end{Proposition}
\begin{proof} 
We first prove the first claim. For each $\mathtt{m}$, it takes $O(n^2)$ time to compute all $\Gamma^{(0)}_k, \Pi^{(0)}_k$ and \cite[Proposition 2.13] {DL22+} can be easily adapted to show that computing $\pi_{\mathtt{m}}$ based on the initialization takes time $O(n^{2+o(1)})$. In addition, it is easy to see that Algorithm~\ref{algo:seeded-matching} runs in time $O(n^3)$. Altogether, this yields the claim.

We now prove the second claim. Since $\mathtt M \leq n^{\kappa}$, the running time for computing all $\pi_{\mathtt{m}}$ is $O(n^{\kappa+3})$. In addition, finding $\Hat{\pi}_\diamond$ from $\{ \pi_{\mathtt{m}} \}$ takes $O(n^{\kappa+2})$ time. So the total running time is $O(n^{\kappa+3})$.
\end{proof}

We complete this section by pointing out that Theorem~\ref{Main-Thm} follows directly from Theorem~\ref{main-thm} and Proposition~\ref{prop_time_complexity}.

\section{Analysis of the matching algorithm}
\label{sec:analysis}
The main goal of this section is to prove Proposition~\ref{prop-almost-exact-matching}.
\subsection{Outline of the proof}
\label{sec:proof-outline}
We fix a good pair $(A, \mathsf A)$. As in \cite{DL22+}, the basic intuition is that each pair of $( \Gamma^{(t)}_k,\Pi^{(t)}_k)$ carries signal of strength at least $\varepsilon_t$, and thus the total signal strength of all $K_t$ pairs will grow in $t$ (recall \eqref{equ-K-Varepsilon-bound}). A natural attempt is to prove this via induction, for which a key challenge is to control correlations among different iterative steps. As a related challenge, we also need to show that the signals carried by different pairs are essentially non-repetitive. To this end, we will (more or less) follow \cite{DL22+} and propose the following \emph{admissible} conditions on $( \Gamma^{(t)}_k,\Pi^{(t)}_k)$ and we hope that this will allow us to verify these admissible conditions by induction. Define the targeted approximation error at time $t$ by 
\begin{align}
    \Delta_t = \Delta_t(\Hat{q},\Hat{\rho}) = e^{ - (\log \log n)^{10} } (\log n)^{10t} \prod_{i \leq t} K_i^{100} \,.
    \label{equ-def-delta}
\end{align}
Since $K_t \leq K_{t^*} \leq (\log n)^4$ and $2^{t^*} \leq 2 \log \log n$, we have 
\begin{equation}{\label{eq-bound-Delta-t}}
    \Delta_t \leq \Delta_{t^*} \leq e^{- (\log \log n)^8} \ll 1 \,.
\end{equation}

\begin{Definition}{\label{def-admissible}}
For $t \geq 0$ and a collection of pairs $( \Gamma^{(s)}_k, \Pi^{(s)}_k)_{1 \leq k \leq K_s,0 \leq s \leq t}$ with $\Gamma^{(s)}_k \subset V$ and $\Pi^{(s)}_k \subset \mathsf{V}$, we say $( \Gamma^{(s)}_k,\Pi^{(s)}_k )_{1 \leq k \leq K_s,0 \leq s \leq t}$ is $t$-admissible if the following hold:
\begin{enumerate}
    \item[(i.)] $\Big| \frac{|\Gamma^{(0)}_k|}{n} - \vartheta \Big|, \Big| \frac{|\Pi^{(0)}_k|}{n} - \vartheta \Big| < \vartheta \Delta_0$ for $1 \leq k \leq K_0$;
    \item[(ii.)] $\Big| \frac{|\Gamma^{(0)}_k \cap \Gamma^{(0)}_{l}|}{n} - \vartheta^2 \Big|, \Big| \frac{|\Pi^{(0)}_k \cap \Pi^{(0)}_{l}|}{n} - \vartheta^2 \Big| < \vartheta \Delta_0$ for $1\leq k \neq l \leq K_0$;
    \item[(iii.)] $\Big| \frac{| \pi(\Gamma^{(0)}_k) \cap \Pi^{(0)}_{k}|}{n} - \varsigma  \Big| < \vartheta \Delta_0$ and $\Big| \frac{| \pi(\Gamma^{(0)}_k) \cap \Pi^{(0)}_{l}|}{n} - \vartheta^2 \Big| < \vartheta \Delta_0$ for $1\leq k \neq l \leq K_0$;
    \item[(iv.)] $\Big| \frac{|\Gamma^{(s)}_k|}{n} -\mathfrak{a} \Big|, \Big| \frac{|\Pi^{(s)}_k|}{n} -\mathfrak{a} \Big| < \mathfrak{a} \Delta_s$ for  $1\leq k\leq K_s$ and $1 \leq s \leq t$;
    \item[(v.)] $\Big | \frac{|\Pi^{(s)}_k \cap \Pi^{(s)}_{l}|}{n} -\phi(  \frac{12}{K_{s-1}}  \langle {\beta}^{(s-1)}_k , {\beta}^{(s-1)}_l \rangle) \Big|,  \Big| \frac{|\Gamma^{(s)}_k \cap \Gamma^{(s)}_{l}|}{n} -\phi( \frac{12}{K_{s-1}}  \langle {\beta}^{(s-1)}_k ,{\beta}^{(s-1)}_l \rangle) \Big| < \mathfrak{a} \Delta_s$ for $1\leq k, l\leq K_s$ and $1 \leq s \leq t$;
    \item[(vi.)] $\Big | \frac{| \pi(\Gamma^{(s)}_k) \cap \Pi^{(s)}_{l}|}{n} -\phi( \frac{ \Hat{\rho} }{2}  \frac{12}{K_{s-1}}   \langle \Hat{\beta}^{(s-1)}_k , \Hat{\beta}^{(s-1)}_l \rangle) \Big| < \mathfrak{a} \Delta_s$ for $1\leq k, l\leq K_s$ and $1 \leq s \leq t$;
    \item[(vii.)] $\Big| \frac{|\Gamma^{(s)}_k \cap \Gamma^{(r)}_{l}|}{n} - \mathfrak{a}^2 \Big|, \Big| \frac{|\Pi^{(s)}_k \cap \Pi^{(r)}_{l}|}{n} - \mathfrak{a}^2 \Big| < \mathfrak{a} \Delta_s$ for $1\leq k\leq K_s$, $1\leq l\leq K_r$ and $1 \leq r < s \leq t$;
    \item[(viii.)] $\Big | \frac{| \pi(\Gamma^{(s)}_k) \cap \Pi^{(r)}_{l}|}{n} - \mathfrak{a}^2 \Big| < \mathfrak{a} \Delta_{\max(s,r)}$ for $1\leq k\leq K_s$, $1\leq l\leq K_r$ and $1 \leq r \not = s \leq t$.
    \item[(ix.)] $\Big| \frac{|\Gamma^{(s)}_k \cap \Gamma^{(0)}_{l}|}{n} - \mathfrak{a} \vartheta \Big|, \Big| \frac{|\Pi^{(s)}_k \cap \Pi^{(0)}_{l}|}{n} - \mathfrak{a} \vartheta \Big| < \sqrt{\vartheta \mathfrak{a}} \Delta_s$ for $1\leq k\leq K_s$, $1\leq l\leq K_0$ and $0 < s \leq t$;
    \item[(x.)] $\Big | \frac{| \pi(\Gamma^{(s)}_k) \cap \Pi^{(0)}_{l}|}{n} - \mathfrak{a} \vartheta \Big|, \Big | \frac{| \pi(\Gamma^{(0)}_l) \cap \Pi^{(s)}_{k}|}{n} - \mathfrak{a} \vartheta \Big| < \sqrt{\vartheta \mathfrak{a}} \Delta_s$ for $1\leq k\leq K_s$, $1\leq l\leq K_0$ and $0 < s \leq t$.
\end{enumerate}
Here $\mathfrak{a},\vartheta,K_t,\phi,\beta^{(t)}$ and $\Hat{\beta}^{(t)}$ are defined previously in Section~\ref{sec:algorithm-description}.
\end{Definition}
Proofs for \cite[(3.3)-(3.8)]{DL22+} can be easily adapted (with no essential change) to show that under the assumption of admissibility, the matrices $\mathrm{M}^{(t,t)}_{\Gamma}, \mathrm{M}^{(t,t)}_{\Pi}$ and $ \mathrm{P}^{(t,t)}_{\Gamma,\Pi}$ concentrate around $\Phi^{(t)},\Phi^{(t)}$ and $\Psi^{(t)}$ respectively with error $\Delta_t$, and $\mathrm{M}^{(t,s)}_{\Gamma}, \mathrm{M}^{(t,s)}_{\Pi}, \mathrm{P}^{(t,s)}_{\Gamma,\Pi}$ have entries bounded by $\Delta_t$. For notational convenience, define
\begin{align}
    \mathcal{E}_t = \{  (\Gamma^{(s)}_k,\Pi^{(s)}_k)_{ 0 \leq s \leq t, 1 \leq k \leq K_s } \mbox{ is $t$-admissible} \}\,. \label{equ-def-admissible}
\end{align}
As hinted earlier, to verify $\mathcal{E}_t$ inductively, the main difficulty is the complicated dependency among the iteration steps. In \cite{DL22+}, much effort was dedicated to this issue even with the help of Gaussian property. In the present work, this is even much harder than in \cite{DL22+} due to the lack of Gaussian property (for instance, a crucial Gaussian property is that conditioned on linear statistics a Gaussian process remains Gaussian). To this end, our rough intuition is to try to compare our process with a Gaussian process whenever possible, and (as we will see) major challenge arises in case when such comparison is out of control.

We first consider the initialization. Define 
\begin{equation}\label{eq-def-REV}
    \mathrm{REV}^{(a)} = \cup_{0 \leq j \leq a} \cup_{1 \leq k \leq K_0} \big( \aleph^{(j)}_k \cup \pi^{-1}( \Upsilon^{(j)}_k ) \big)
\end{equation}
to be the set of vertices that we have explored for initialization either directly or indirectly (e.g., through correlation of the latent matching). We further denote $\mathrm{REV}=\mathrm{REV}^{(\chi)}$. Define
\begin{align*}
    \mathfrak{S}_{\mathrm{init}} = \sigma \big\{ \mathrm{REV}, \{ \overrightarrow{G}_{u,w}, \overrightarrow{\mathsf{G}}_{\pi(u), \pi(w)}: u \in \mathrm{REV} \mbox{ or } w \in \mathrm{REV} \} \big\}
\end{align*}
(note that $\mathrm{REV} = A$ is measurable with respect to $ \{ \overrightarrow{G}_{u,w}, \overrightarrow{\mathsf{G}}_{\pi(u), \pi(w)}: u \in A \mbox{ or } w \in A \} $).
We have $\{ \Gamma^{(0)}_k, \Pi^{(0)}_k \}$ is measurable with respect to $\mathfrak{S}_{\mathrm{init}}$. We will show that conditioning on a realization of $(\aleph^{(a)}_j,\Upsilon^{(a)}_j)$ will not affect the degree of ``most'' vertices, thus verifying the concentration of $\aleph^{(a+1)}_j, \Upsilon^{(a+1)}_j$ inductively. This will eventually yield that $\mathcal{E}_0$ holds with probability $1-o(1)$, as incorporated in Section~\ref{sec:priliminary-events}.

We now consider the iterations. When comparing our process to the case of Wigner matrices, a key challenge is that in each time $t$ we can only control the behavior (i.e., show they are close to ``Gaussian'') of all but a vanishing fraction of $\langle \eta^{(t)}_k, D^{(t)}_v \rangle$. This set of uncontrollable vertices will be inductively defined as $\mathrm{BAD}_{t}$ in \eqref{equ-def-set-BAD}. However, the algorithm forces us to deal with the behavior of $\langle \eta^{(t+1)}_k, D^{(t+1)}_v \rangle$ conditioned on \emph{all} $\{ \langle \eta^{(s)}_k, D^{(s)}_v \rangle : 0 \leq s \leq t \}$ (since the algorithm has explored all these variables), and thus we also need to control the influence from these ``bad'' vertices. To address this problem, we will separate $\langle \eta^{(t+1)}_k, D^{(t+1)}_v \rangle$ into two parts, one involved with bad vertices and one not. To this end,
for $0\leq s \leq t+1$ and  for $v \not \in \mathrm{BAD}_t$,  we decompose $D^{(s)}_v (k)$ into a sum of two terms with $D^{(s)}_v (k) = b_t{D}^{(s)}_v (k) + g_t{D}^{(s)}_v (k)$ (and similarly for the mathsf version), where 
\begin{align}
     b_t{D}^{(s)}_v(k) \overset{\triangle}{=} \frac{1}{\sqrt{(\mathfrak{a}_s-\mathfrak{a}_s^2)n \Hat{q}(1-\Hat{q})}} \sum_{u \in \mathrm{BAD}_t} (\mathbf{1}_{u \in \Gamma^{(s)}_k} - \mathfrak{a}_s) (\overrightarrow{G}_{v,u}-\Hat{q})  \,.
     \label{eq-def-b-t-D-t}
\end{align}
Further, we define $\mathfrak{S}_{t}$ to be the $\sigma$-field generated by
\begin{equation}\label{eq-def-mathfrak-S-t}
\begin{aligned}
    & \{ W^{(s)}_v(k) + \langle \eta^{(s)}_k,D^{(s)}_v \rangle, \mathsf{W}^{(s)}_{ \pi(v) }(k) + \langle \eta^{(s)}_k,\mathsf{D}^{(s)}_{\pi(v)} \rangle : 1\leq k \leq \tfrac{K_s}{12}, 0 \leq s \leq t, v\in V \}\,,  \\
    &\mathrm{BAD}_t \mbox{ and } \{  \overrightarrow{G}_{u,w} , \overrightarrow{\mathsf{G}}_{\pi(u), \pi(w)} : u \mbox{ or } w \in \mathrm{BAD}_t \}   \,.
 \end{aligned}
\end{equation}
Then we see that $D_v^{(s)}$ is fixed under $\mathfrak S_t$ for $v\in \mathrm{BAD}_t$ and $1\leq s\leq t$, and for $v \not\in \mathrm{BAD}_t$ in a sense conditioned on $\mathfrak S_t$ we may view $b_t{D}^{(s)}_v(k) $ and $g_t{D}^{(s)}_v(k)$ as the ``biased'' part and the ``free'' part, respectively. 
We set $\mathrm{BAD}_{-1} = \mathrm{REV}$, and inductively define 
\begin{equation}
    \begin{aligned}
        & \mathrm{BIAS}_{D,t,s,k} = \Big\{ v \in V \setminus \mathrm{BAD}_{t-1} : | b_{t-1}{D}^{(s)}_v (k) | > e^{ - 10 (\log \log n)^{10} } \Big\}   \,, \\
        & \mathrm{BIAS}_{\mathsf{D},t,s,k} = \Big\{ v \in V \setminus \mathrm{BAD}_{t-1} : | b_{t-1}\mathsf{D}^{(s)}_{\pi(v)} (k) | > e^{ - 10 (\log \log n)^{10} } \Big\} \,,
    \label{equ-def-set-BIAS}
    \end{aligned}
\end{equation}
and then define $\mathrm{BIAS}_t = \cup_{0 \leq s \leq t} \cup_{1 \leq k \leq K_s} \big( \mathrm{BIAS}_{D,t,s,k} \cup \mathrm{BIAS}_{\mathsf{D},t,s,k} \big)$ to be the collection of vertices that are overly biased by the set $\mathrm{BAD}_{t-1}$ (see \eqref{equ-def-set-BAD}). Accordingly, we will give up on controlling the behavior for vertices in $\mathrm{BIAS}_t$.

Now we turn to the term $g_t D^{(t+1)}_v$. We will try to argue that this ``free'' part behaves like a suitably chosen Gaussian process via the technique of density comparison, as in Section~\ref{sec:density-compare}. To this end, we sample a pair of Gaussian matrices $(\overrightarrow{Z}, \overrightarrow{\mathsf{Z}})$ with zero diagonal terms  such that their off-diagonals $\{ \overrightarrow{Z}_{u,v}, \overrightarrow{\mathsf{Z}}_{\mathsf{u,v}} \}$ is a family of Gaussian process with mean 0 and variance $\Hat{q}(1-\Hat{q})$, and the only non-zero covariance occurs on pairs of the form $(\overrightarrow{Z}_{u,v},  \overrightarrow{\mathsf{Z}}_{\pi(u),\pi(v)})$ or $(\overrightarrow{Z}_{u,v},  \overrightarrow{\mathsf{Z}}_{\pi(v),\pi(u)})$ (for $u\neq v\in V$) where  $\mathbb{E}[ \overrightarrow{Z}_{u,v} \overrightarrow{\mathsf{Z}}_{\pi(u),\pi(v)} ] = \mathbb{E}[ \overrightarrow{Z}_{u,v} \overrightarrow{\mathsf{Z}}_{\pi(v),\pi(u)} ] = \Hat{\rho}$ (this is analogous to the process defined in \cite[Section 2.1]{DL22+}). In addition, we sample i.i.d\ standard normal variables $\Tilde{W}^{(s)}_v(k), \Tilde{\mathsf{W}}^{(s)}_{\pi(v)}(k)$ for $0 \leq s \leq t^*, v \in V, 1 \leq k \leq K_s/12$. (We emphasize that we will sample $(\overrightarrow{Z}, \overrightarrow{\mathsf{Z}}, \Tilde{W},\Tilde{\mathsf{W}})$ only once and then will stick to it throughout the analysis.) We can then define the following ``Gaussian substitution'', where we replace each $\overrightarrow{G}_{v,u}$ with $\overrightarrow{Z}_{v,u}$ for each $u,v \not \in \mathrm{BAD}_t$: for $v \not \in \mathrm{BAD}_{t}$, define $g_t \Tilde{D}^{(s)}_v$ to be a $\sum_{s \leq t+1} K_s$-dimensional vector whose $k$-th entry is given by  
\begin{align}
    g_{t} \Tilde{D}^{(s)}_v (k) = \frac{ \sum_{u \in V \setminus \mathrm{BAD}_{t}} (\mathbf{1}_{u \in \Gamma^{(s)}_k} - \mathfrak{a}_s) \overrightarrow{Z}_{v,u} } {\sqrt{(\mathfrak{a}_s-\mathfrak{a}_s^2) n \Hat{q}(1-\Hat{q})}} \mbox{ for } 0 \leq s \leq t+1 \,,
    \label{equ-def-g-Tilde-D}
\end{align}
and we define $g_{t} \Tilde{\mathsf{D}}^{(s)}_{\mathsf{v}}$ similarly. We will use a delicate Lindeberg's interpolation argument to bound the ratio of the densities before and after the substitution. To be more precise, we will process each pair $(u,w)$ sequentially (in an arbitrarily prefixed order) where we replace $\{ \overrightarrow{G}_{u,w}- \Hat{q}, \overrightarrow{G}_{w,u}-\Hat{q}, \overrightarrow{\mathsf{G}}_{\pi(u),\pi(w)}-\Hat{q}, \overrightarrow{\mathsf{G}}_{\pi(w),\pi(u)}-\Hat{q}  \}$ by $\{ \overrightarrow{Z}_{u,w}, \overrightarrow{Z}_{w,u}, \overrightarrow{\mathsf{Z}}_{\pi(u),\pi(w)}, \overrightarrow{\mathsf{Z}}_{\pi(w),\pi(u)}  \}$. Define this operation as $\mathbf{O}_{ \{ u,w \} }$, and the key is to bound the change of density ratio for each operation. To this end, list $\{ \mathbf{O}_{\{u,w\}} : (u,w) \in E_0, u \neq w \}$ as $\{ \mathbf{O}_{\{ u_1,w_1 \}}, \ldots, \mathbf{O}_{\{ u_N,w_N \}} \}$ in the aforementioned prefixed order, and define a corresponding operation $\mathbf{U}_{\{ u,w \}}$ which replaces $\{ \overrightarrow{G}_{u,w}- \Hat{q}, \overrightarrow{G}_{w,u}-\Hat{q}, \overrightarrow{\mathsf{G}}_{\pi(u),\pi(w)}-\Hat{q}, \overrightarrow{\mathsf{G}}_{\pi(w),\pi(u)}-\Hat{q} \}$ by $\{0, 0, 0, 0\}$. For $0 \leq j \leq N$ and for any random variable $\mathbf{X}$, define 
\begin{align}
    \mathbf{X}_{(j)} = \circ_{i \leq j} \mathbf{O}_{ \{ u_i,w_i \} } \big( \mathbf{X} \big), \quad \mathbf{X}_{ \langle j \rangle} = \circ_{i \leq j} \mathbf{U}_{ \{ u_i,w_i \} } \big( \mathbf{X} \big), \quad \mathbf{X}_{[j]} = \mathbf{X}_{(j)} - \mathbf{X}_{ \langle j \rangle }\,, \label{equ-def-degree-after-j-substitution}
\end{align}
where $\circ_{i \leq j} \mathbf{O}_{ \{ u_i,w_i \} }$ is the composition of the first $j$ operations $\{\mathbf{O}_{ \{ u_1,w_1 \} }, \ldots, \mathbf{O}_{ \{ u_j,w_j \} }\}$ and $\circ_{i \leq j} \mathbf{U}_{ \{ u_i, w_i \} }$ is the composition of the first $j$ operations  $\{\mathbf{U}_{ \{ u_1,w_1 \} }, \ldots, \mathbf{U}_{ \{ u_j,w_j \} }\}$. For notational convenience, we simply define $\circ_{i \leq 0} \mathbf{O}_{\{ u_i, w_i \}} = \circ_{i \leq 0} \mathbf{U}_{\{ u_i, w_i \}} = \mathbf{Id}$ to be the identity map. We will see that a crucial point of our argument is to employ suitable truncations; such truncations will be useful when establishing Lemma~\ref{lem-good-set-good-variable-realization}.
To this end, we define $\mathrm{LARGE}^{(0)}_{t,s,k} $ to be the collection of vertices $ v \in V \setminus \mathrm{BAD}_{t-1} $ such that  for some $0\leq j\leq N$
\begin{equation}\label{equ-def-set-LARGE-(0)}
    \big| W^{(s)}_v(k) \big| \mbox{ or } \big| \langle \eta^{(s)}_k, g_{t-1} {D}^{(s)}_v \rangle_{ \langle j \rangle } \big|  \mbox{ or }  \big| \mathsf{W}^{(s)}_{\pi(v)}(k) \big|  \mbox{ or } \big| \langle \eta^{(s)}_k, g_{t-1} {\mathsf D}^{(s)}_{\pi(v)} \rangle_{ \langle j \rangle } \big|> n^{ \frac{1} {\log \log \log n}}\,.
\end{equation}
Let $\mathrm{LARGE}^{(0)}_t = \cup_{s=0}^{t} \cup_{k=1}^{\frac{K_s}{12}}\mathrm{LARGE}^{(0)}_{ t,s,k }$.
We then define
    \begin{align*}
    b_{t,0} {D}^{(s)}_v(k) &= \frac{1}{\sqrt{ ( \mathfrak{a}_s - \mathfrak{a}_s^2) n \Hat{q}(1-\Hat{q}) }} \sum_{ u \in \mathrm{LARGE}_{t}^{(0)} \cup \mathrm{BIAS}_t \cup \mathrm{PRB}_t } ( \mathbf{1}_{ u \in \Gamma^{(s)}_k } - \mathfrak{a}_s ) (\overrightarrow{G}_{v,u} - \Hat{q})\,,\\
    b_{t,0} {\mathsf D}^{(s)}_{\pi(v)}(k) &= \frac{1}{\sqrt{ ( \mathfrak{a}_s - \mathfrak{a}_s^2) n \Hat{q}(1-\Hat{q}) }} \sum_{ u \in \mathrm{LARGE}_{t}^{(0)} \cup \mathrm{BIAS}_t \cup \mathrm{PRB}_t } ( \mathbf{1}_{ \pi(u) \in \Pi^{(s)}_k } - \mathfrak{a}_s ) (\overrightarrow{\mathsf{G}}_{\pi(v),\pi(u)} - \Hat{q}) \,.
\end{align*}
Here $\mathrm{PRB}_t$ is defined in \eqref{equ-def-set-PRB} below and later we will explain that this definition is valid (although we have not defined $\mathrm{PRB}_t$ yet). 
For $a\geq 0$, we inductively define $\mathrm{LARGE}^{(a+1)}_{t,s,k} $ to be the collection of vertices $ v \in V \setminus \mathrm{BAD}_{t-1} $ such that  for some $0\leq j\leq N$
\begin{equation}
\big| \langle \eta^{(s)}_k , b_{t,a} {D}^{(s)}_v \rangle_{ \langle j \rangle } \big| \mbox{ or } \big| \langle \eta^{(s)}_k , b_{t,a} {\mathsf D}^{(s)}_{\pi(v)} \rangle_{ \langle j \rangle } \big|  > n^{ \frac{1} {\log \log \log n}}\,,  \label{equ-def-set-LARGE}
\end{equation}
define $\mathrm{LARGE}^{(a+1)}_t = \bigcup_{ s=0}^{ t } \bigcup_{k=1}^{\frac{K_s}{12}} \mathrm{LARGE}^{(a+1)}_{t,s,k }$, and define
\begin{align*}
    b_{t,a+1} {D}^{(s)}_v(k) = \frac{1}{\sqrt{ ( \mathfrak{a}_s - \mathfrak{a}_s^2) n \Hat{q}(1-\Hat{q}) }} \sum_{ u \in \mathrm{LARGE}^{(a+1)}_t } ( \mathbf{1}_{ u \in \Gamma^{(s)}_k } - \mathfrak{a}_s ) (\overrightarrow{G}_{v,u} - \Hat{q}) \,.
\end{align*}
Also we similarly define $b_{t,a+1} {\mathsf D}^{(s)}_{\pi(v)}(k)$ (using similar analogy for $a=0$). 
Having completed this inductive definition, we finally write $\mathrm{LARGE}_{t} =  \cup_{a=0}^{\infty} \mathrm{LARGE}^{(a)}_t$. We will argue that after removing $\mathrm{LARGE}_t$ the remaining random variables have smoothed density whose change in each substitution can be bounded, thereby verifying that their original joint density is not too far away from that of a Gaussian process by a delicate Lindeberg's argument. The details of such Lindeberg's argument are incorporated in Section~\ref{sec:density-compare}. 

Thanks to the above discussions, we have \emph{essentially} reduced the problem to analyzing the corresponding Gaussian process (not exactly yet since we still need to consider yet another type of bad vertices arising from Gaussian approximation as in \eqref{equ-def-set-PRB} below). To this end, we will employ the techniques of Gaussian projection. Define $\mathcal{F}_t = \sigma(\mathfrak F_t)$ where
\begin{equation}
    \mathfrak{F}_t = \Bigg\{ \begin{aligned}
    \Tilde{W}^{(s)}_v(k) + \langle \eta^{(s)}_k, g_{t} \Tilde{D}^{(s)}_v \rangle \\ \Tilde{\mathsf{W}}^{(s)}_{\mathsf{v}}(k) + \langle \eta^{(s)}_k, g_{t} \Tilde{\mathsf{D}}^{(s)}_{\mathsf{v}} \rangle \end{aligned} : 0 \leq s \leq t, 1 \leq k \leq \frac{K_s}{12} , v, \pi^{-1}(\mathsf{v}) \not \in \mathrm{BAD}_{t} \Bigg\}\,. 
    \label{eq-def-mathcal-F-t}
\end{equation} 
We will condition on (see \eqref{eq-convention-announcement} below) 
\begin{equation}\label{eq-to-be-conditioned-on}
\{\Gamma^{(s)}_k, \Pi^{(s)}_k, \mathrm{BAD}_s: 0\leq s\leq t, 1\leq k\leq K_s\}
\end{equation}
and thus $\mathfrak F_t$ is viewed as a Gaussian process. We can then obtain the conditional distribution of $\Tilde{W}^{(t+1)}_v(k) + \langle \eta^{(t+1)}_k, g_t \Tilde{D}^{(t+1)}_v \rangle $ given $\mathcal F_t$ (see Remark~\ref{remark-conditional-Gaussian}). In particular, we will show that the projection of $\Tilde{W}^{(t+1)}_v(k) + \langle \eta^{(t+1)}_k, g_t \Tilde{D}^{(t+1)}_v \rangle $ onto $\mathcal{F}_t$ has the form
\begin{align*}
    \begin{pmatrix}
        g_t[\Tilde{Y}]_{t} &
        g_t[\Tilde{\mathsf{Y}}]_{t}
    \end{pmatrix}
    \mathbf{Q}_t     
    \begin{pmatrix}
        H_{t+1,k,v} & \mathsf{H}_{t+1,k,v}  
    \end{pmatrix}^{*}  \,.
\end{align*}
Here $g_{t}[\Tilde{Y}]_{t} (s,k,u) = \Tilde{W}^{(s)}_u(k) + \langle \eta^{(s)}_k, g_{t} \Tilde{D}^{(s)}_u \rangle, g_{t}[\Tilde{\mathsf{Y}}]_{t} (s,k,\mathsf u) = \Tilde{\mathsf W}^{(s)}_{\mathsf u}(k) + \langle \eta^{(s)}_k, g_{t} \Tilde{\mathsf{D}}^{(s)}_{\mathsf u} \rangle$, and $\mathbf{Q}^{-1}_t$ is the  conditional covariance matrix  of $\mathfrak{F}_t$ given \eqref{eq-to-be-conditioned-on}. In addition, ${H}_{t+1,k,v}(s,l,u)$ is the conditional covariance between $\Tilde{W}^{(t+1)}_v(k) + \langle \eta^{(t+1)}_k, g_t \Tilde{D}^{(t+1)}_v \rangle$ and $ \Tilde{W}^{(s)}_u(l) + \langle \eta^{(s)}_l, g_t \Tilde{D}^{(s)}_u \rangle$, and $\mathsf{H}_{t+1,k,v}(s,l,\mathsf{u})$ is that between $\Tilde{W}^{(t+1)}_v(k) + \langle \eta^{(t+1)}_k, g_t \Tilde{D}^{(t+1)}_v \rangle$ and $\Tilde{\mathsf{W}}^{(s)}_{\mathsf{u}}(l) + \langle \eta^{(s)}_l, g_t \Tilde{\mathsf{D}}^{(s)}_{\mathsf{u}} \rangle$. See \eqref{equ-formulate-projection} and Remark~\ref{remark-conditional-Gaussian} for precise definitions of $\mathbf{Q}$ and $H,\mathsf{H}$. Furthermore, we define
\begin{align*}
g_{t-1}[Y]_{t-1} (s,k,u)& =  W^{(s)}_u(k) + \langle \eta^{(s)}_k, g_{t-1} D^{(s)}_u \rangle\,,\\ 
[g Y]_{t-1} (s,k,u) &= W^{(s)}_u(k) + \langle \eta^{(s)}_k, g_{s-1} D^{(s)}_u \rangle\,,
\end{align*}  
and define the mathsf version similarly.
We let $\mathrm{PRB}_{t,k} $ be the collection of $v\in V$ such that
\begin{align} 
    \big| \begin{pmatrix}
    [gY]_{t-1} - g_{t-1}[Y]_{t-1} & [g\mathsf{Y}]_{t-1} - g_{t-1}[\mathsf{Y}]_{t-1}
    \end{pmatrix}  \mathbf{Q}_{t-1}
    \begin{pmatrix}
        H_{t,k,v} & \mathsf{H}_{t,k,v}
    \end{pmatrix}^{*} \big| > \Delta_{t}
    \label{equ-def-set-PRB}
\end{align}
and let $\mathrm{PRB}_{t} = \cup_{1 \leq k \leq K_t} \mathrm{PRB}_{t,k}$; we will see that removing the vertices in $\mathrm{PRB}_{t}$ is crucial for establishing Lemma~\ref{lemma-projection-replace}. Since $[g Y]_{t-1}, g_{t-1}[Y]_{t-1}$, $\mathbf{Q}_{t-1}$ and $H_{t,k,v}$ are all measurable with respect to $\mathfrak{S}_{t-1}$,  we could have defined $\mathrm{PRB}_t$ before defining $\mathrm{LARGE}^{(a)}_t$ as in \eqref{equ-def-set-LARGE}; we chose to postpone the definition of $\mathrm{PRB}_t$ until now since it is only natural to introduce it after writing down the form of the Gaussian projection. 
Finally, we are ready to complete the inductive definition for the ``bad'' set as follows:
\begin{align}
    \mathrm{BAD}_{t} = \mathrm{BAD}_{t-1} \cup \mathrm{LARGE}_t \cup \mathrm{BIAS}_{t} \cup \mathrm{PRB}_t \,.  \label{equ-def-set-BAD} 
\end{align}
We summarize in Figure \ref{fig:logic} the logic flow for defining $\mathrm{BAD}_t$ from $\mathrm{BAD}_{t-1}$ and variables in $\mathfrak{S}_{t-1}$, which should illustrate clearly that our definitions are not ``cyclic''.
\begin{figure}[!ht]
    \centering
    \vspace{0cm}
    \includegraphics[height=7cm,width=15.8cm]{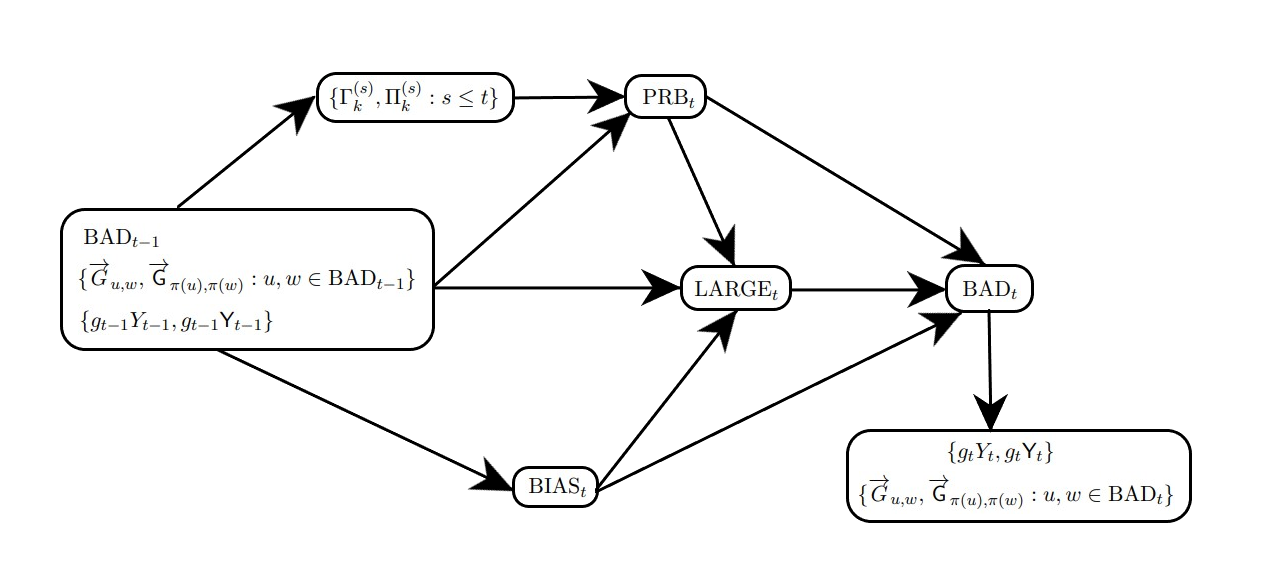}
    \caption{Logic of the definition}
    \label{fig:logic}
\end{figure}

Recall \eqref{eq-def-REV} and recall $\mathrm{BAD}_{-1} = \mathrm{REV}$.
Let $\mathcal{T}_{-1}= \{ |\mathrm{BAD}_{-1}| \leq  ( 4 n K_0^2 \vartheta_{\chi} + \sqrt{n/\Hat{q}(1-\Hat{q})} + n^{1-\frac{2}{\log \log \log n}} ) \} $, and for $t \geq 0$ let $\mathcal{T}_t$ be the event such that for $s \leq t$
\begin{equation}\label{eq-def-mathcal-T-t}
\begin{aligned}
    & |\mathrm{BAD}_s| \leq e^{ 30(s+1)(\log \log n)^{20}} \vartheta^{-3(s+1)} ( 4 n K_0^2 \vartheta_{\chi} + \sqrt{n/\Hat{q}(1-\Hat{q})} + n^{1-\frac{2}{\log \log \log n}} ) \,, \\
    & \mbox{and } \mathrm{LARGE}_s^{(\log n)} = \emptyset \,.
\end{aligned}
\end{equation}
By Lemma~\ref{lemma-property-vartheta-varsigma} and the fact that $t \leq 2 \log \log \log n$, we have  $|\mathrm{BAD}_t| \ll n \vartheta^{10} \Delta_t^{10}$ on the event $\mathcal T_t$. Our hope is that, on the one hand $\mathcal T_t$ occurs typically (as stated in Proposition~\ref{prop-cardinality-BAD} below) so that the number of ``bad'' vertices is under control, and on the other hand on $\mathcal T_t$ most vertices can be dealt with by techniques of Gaussian projection which then allows to control the conditional distribution.
\begin{Proposition} {\label{prop-cardinality-BAD}}
    We have $\mathbb P(\cap_{0 \leq t \leq t^*} \mathcal{T}_t \cap \mathcal E_t) = 1-o(1)$.
\end{Proposition}
In Section~\ref{sec:actual-proof}, we will prove Proposition~\ref{prop-cardinality-BAD} via induction.

\subsection{Preliminaries on probability and linear algebra}
In this subsection we collect some standard lemmas on probability and linear algebra that will be useful for our further analysis.

The following version of Bernstein's inequality appeared in \cite[Theorem 1.4]{DP09}.
\begin{Lemma} {\label{lemma-Bernstein-inequality}}
\textup{(Bernstein's inequality)}.
Let $X=\sum_{i=1}^{m} X_i$, where $X_i$'s are independent random variables such that $|X_i| \leq K$ almost surely. Then, for $s>0$ we have
\begin{align*}
    \mathbb{P}( |X-\mathbb{E}[X]| > s ) \leq 2 \exp \Big\{ - \frac{ s^2 }{ 2(\sigma^2 + Ks/3) } \Big\} \,,
\end{align*}
where $\sigma^2 = \sum_{i=1}^{m} \mathrm{Var}(X_i)$ is the variance of $X$.
\end{Lemma}

We will also use Hanson-Wright inequality (see \cite{HW71, Wright73, FR13} and see \cite[Theorem 1.1]{RV14}), which is useful in controlling the quadratic form of sub-Gaussian random variables.
\begin{Lemma}{\label{lemma-Hanson-Wright}}
    \textup{(Hanson-Wright Inequality)}. Let $X=(X_1,\ldots,X_m)$ be a random vector with independent components which satisfy  $\mathbb{E}[X_i]=0$ and  $\| X_i \|_{\psi_2} \leq K$ for all $1\leq i\leq m$. If $\mathrm{A}$ is an $m\!*\!m$ symmetric matrix, then for an absolute constant $c>0$ we have
    \begin{equation}
        \mathbb{P} \big( |X \mathrm{A} X^{*} - \mathbb{E}[X \mathrm{A} X^{*} ] | > s \big) \leq 2 \exp \Big\{  -c \min \Big( \frac{s^2}{ K^4 \| \mathrm{A} \|^2_{\mathrm{HS}}}, \frac{s}{ K^2 \| \mathrm{A} \|_{\mathrm{op}}} \Big) \Big\}\,.
        \label{equ_Hanson_Wright_tail}
    \end{equation}
\end{Lemma}

The following is a corollary of Hanson-Wright inequality (see
 \cite[Lemma 3.8]{DL22+}).
 \begin{Lemma} {\label{lemma-modify-Hanson-Wright}}
    Let $X_1,\ldots,X_m,Y_1,\ldots,Y_m$ be mean-zero variables with $\| X_i \|_{\psi_2}, \| Y_i \|_{\psi_2} \leq K$ for all $1\leq i\leq m$. In addition, assume for all $i$ that $(X_i, Y_i)$ is independent with $(X_{\setminus i}, Y_{\setminus i})$ where $X_{\setminus i}$ is obtained from $X$ by dropping its $i$-th component (and similarly for $Y_{\setminus i}$). Let $\mathrm{A}$ be an $m\!*\!m$ matrix with diagonal entries being 0. Then for an absolute constant $c>0$ and for every $s > 0$ 
    \begin{align*}
        \mathbb{P} ( | X \mathrm{A} Y^{*} | > s ) \leq 2 \exp \Big \{  -c \min \Big(  \frac{s^2}{K^4 \| \mathrm{A} \|^2_{\mathrm{HS}}}, \frac{ s }{ K^2 \| \mathrm{A} \|_{\mathrm{op}} } \Big)  \Big \}\,.
    \end{align*}
\end{Lemma}
The following inequality is standard for posterior estimation.
\begin{Lemma} {\label{lemma-posterior-estimator}}
For a random variable $X$ and an event $A$ in the same probability space, we have $\mathbb{P}_{x \sim X|A} \big( \mathbb{P}(A|X=x) \geq \epsilon \mathbb{P}(A) \big) \geq 1-\epsilon$ for $\epsilon \in [0,1]$.
\end{Lemma}
\begin{proof}
We have that
\begin{align*}
    \mathbb{P}_{x \sim X|A} \big( \mathbb{P}(A|X=x) \geq \epsilon \mathbb{P}(A) \big) = & \mathbb{E}_{x \sim X} \Big[ \frac{1}{\mathbb{P}(A)} \mathbb{P}(A|X=x) \mathbf{1}_{ \{ \mathbb{P}(A|X=x) \geq \epsilon \mathbb{P}(A) \} } \Big]  \\
    \geq & \frac{1}{\mathbb{P}(A)} \mathbb{E}_{x \sim X} [ \mathbb{P}(A|X=x) - \epsilon \mathbb{P}(A) ] = 1 - \epsilon \,,  
\end{align*} 
completing the proof of this lemma.
\end{proof}

We also need some results in linear algebra.
\begin{Lemma} {\label{lemma-bound-on-infty-norm}}
For two $m\!*\!m$ matrices $\mathrm{A,B}$, if $\| \mathrm{(A+B)}^{-1} \|_{\mathrm{op}} \leq C, \| \mathrm{A}^{-1} \|_{\infty} \leq L$ and the entries of $\mathrm{B}$ are bounded by $\frac{K}{m}$, then $\| \mathrm{(A+B)}^{-1} \|_{\infty} \leq \max \{ 2KCL, 2L \}$.
\end{Lemma}
\begin{proof}
It suffices to show that $\| \mathrm{(A+B)} x^* \|_{\infty} \geq \min \{ \frac{1}{2KCL}, \frac{1}{2L} \}$ for any $\| x \|_{\infty} =1$. First we consider the case when $\| x \|_2 \geq \frac{\sqrt{m}}{2KL}$. Since $\| \mathrm{(A+B)} x^* \|_2 \geq \| \mathrm{(A+B)}^{-1} \|_{\mathrm{op}}^{-1} \| x \|_2 \geq C^{-1} \| x \|_2$, we get $\| \mathrm{(A+B)} x^* \|^2_{\infty} \geq \frac{ \| \mathrm{(A+B)} x^* \|^2_2 }{ m } \geq \frac{1}{4 K^2 C^2 L^2}$. Next we consider the case when $\| x \|_2 \leq \frac{\sqrt{m}}{2KL}$. In this case $\| x \|_{1} \leq \frac{m}{2KL}$. Thus, $\| \mathrm{A} x^* \|_{\infty} \geq \| \mathrm{A}^{-1} \|_{\infty}^{-1} \| x \|_{\infty} \geq L^{-1} $ and $\| \mathrm{B} x^* \|_{\infty} \leq \frac{K}{m} \| x \|_{1} \leq (2L)^{-1}$. Therefore, $\| \mathrm{(A+B)} x^* \|_{\infty} \geq \| \mathrm{A} x^* \|_{\infty} - \| \mathrm{B} x^* \|_{\infty}  \geq \frac{1}{2L}$, as desired.
\end{proof}

\begin{Lemma} {\label{lemma-bound-on-op-norm}}
For an $m*\ell$ matrix $\mathrm{A}$, suppose that there exist two partitions $\{ 1,\ldots,m \} = \sqcup_{k=1}^{K} \mathcal{I}_k$ and $\{ 1,\ldots,\ell \} = \sqcup_{k=1}^{K} \mathcal{J}_k$ with $|\mathcal{I}_k|, |\mathcal{J}_k| \leq D$ (for $1\leq k\leq K$) such that $| \mathrm{A}_{a,b}| \leq \delta$ for $1 \leq k \leq K$ and $(a,b) \in \mathcal{I}_k \times \mathcal{J}_k$,  and that $\sum_{k \neq l} \sum_{ (a,b) \in \mathcal{I}_k \times \mathcal{J}_l } \mathrm{A}_{a,b}^2 \leq C^2$. Then we have $\| \mathrm{A} \|_{\mathrm{op}} \leq D \delta + C$.
\end{Lemma}
\begin{proof}
Denote $|\mathcal I_k|= m_k$ and $|\mathcal J_k|=\ell_k$.
Define $\mathrm{A}^{\mathrm{diag}}$ such that $\mathrm{A}^{\mathrm{diag}}_{a,b} = \mathrm{A}_{a,b}$ for $(a,b) \in \mathcal{I}_k \times \mathcal{J}_k$ and that $\mathrm{A}^{\mathrm{diag}}_{a,b} = 0$ otherwise. Then, there exist two permutation matrices $\mathrm{Q_1,Q_2}$ such that 
\[
\mathrm{Q}_1  \mathrm{A}^{\mathrm{diag}} \mathrm{Q}_2 = \begin{pmatrix}
    \mathrm{A}_{1} & \mathrm{O}_{m_1*\ell_2} & \cdots & \mathrm{O}_{m_1*\ell_K} \\
    \mathrm{O}_{m_2*\ell_1} & \mathrm{A}_{2} & \cdots & \mathrm{O}_{m_2*\ell_K} \\
    \vdots & \vdots & \vdots & \vdots \\
    \mathrm{O}_{m_K*\ell_1} & \mathrm{O}_{m_K*\ell_2} & \cdots & \mathrm{A}_{K} 
\end{pmatrix} \,.
\]
where $\mathrm{A}_{k}$ is a matrix with size $m_k\!*\!\ell_k$ and with entries bounded by $\delta$. Thus $\| \mathrm{A}^{\mathrm{diag}} \|_{\mathrm{op}} = \max_{1 \leq k \leq K} \| \mathrm{A}_k \|_{\mathrm{op}} \leq D \delta$. Also we have $\| \mathrm{A}- \mathrm{A}^{\mathrm{diag}} \|_{\mathrm{op}}^2 \leq \| \mathrm{A}- \mathrm{A}^{\mathrm{diag}} \|_{\mathrm{HS}}^2 \leq C^2$, and thus the result follows from the triangle inequality.
\end{proof}

\subsection{Analysis of initialization} \label{sec:priliminary-events}
In this subsection we analyze the initialization. We will prove a concentration result for $\aleph^{(a)}_k,\Upsilon^{(a)}_k$ for $1 \leq a \leq \chi+1, 1 \leq k \leq K_0$ in Lemma~\ref{lemma-concentration-iterative-aleph-Upsilon}, which then implies that $\P(\mathcal{E}_{0}) = 1- o(1)$ as in Lemma~\ref{lemma-E-0-holds}. As preparations, we first collect a few technical estimates on binomial variables. 
\begin{Lemma} \label{lemma-difference-Bernoulli-tail-I}
For $m \ll N \ll p^{-1}$ and $l=\Theta(1)$, we have
\begin{align*}
    \mathbb{P}( \mathrm{Bin}(N+m,p) \geq l ) - \mathbb{P}( \mathrm{Bin}(N,p) \geq l ) \leq \frac{2ml}{N} \mathbb{P}( \mathrm{Bin}(N,p) \geq l ) \,.
\end{align*}
\end{Lemma}
\begin{proof}
The left hand side equals to $\sum_{k=1}^{l} \mathbb{P}( \mathrm{Bin}(m,p) = k ) * \mathbb{P}( l > \mathrm{Bin}(N,p) \geq l-k )$. Since $mp \ll 1$, a straightforward computation yields that $\mathbb{P}(\mathrm{Bin}(m,p) \geq k) \sim \frac{ (mp)^{k} }{k!} $ for any fixed $k$. Since $Np \ll 1$, we also have that $\frac{ \mathbb{P}( \mathrm{Bin}(N,p) \geq l-k ) }{ \mathbb{P}( \mathrm{Bin}(N,p) \geq l) } \sim \frac{l!}{(l-k)!} (Np)^{-k} $. Thus,
\begin{align*}
    \frac{ \mathbb{P}( \mathrm{Bin}(N+m,p) \geq l ) - \mathbb{P}( \mathrm{Bin}(N,p) \geq l )}{ \mathbb{P}( \mathrm{Bin}(N,p) \geq l )} \leq 1.5 \sum_{k=1}^{l} \frac{l!}{(l-k)!} (Np)^{-k} \cdot\frac{1}{k!}(mp)^{k} \leq \frac{2ml}{N},
\end{align*}
which yields the desired bound.
\end{proof}

\begin{Corollary} \label{corollary-difference-CorBernoulli-tail-I}
For $m \ll N,M \ll p^{-1}$ and $l = \Theta(1)$, we have
\begin{align*}
    & \mathbb{P}( \mathrm{CorBin} (N+m,N+m,p;M+m,\rho) \geq (l,l)) - \mathbb{P}(\mathrm{CorBin}(N,N,p;M,\rho) \geq (l,l) ) \\
    & \leq \frac{4 l m}{N} \mathbb{P} ( \mathrm{Bin}(N,p) \geq l ) \,.
\end{align*}
\end{Corollary}
\begin{proof}
Let $(X,Y) \overset{d}{=} \mathrm{CorBin}(N,N,p;M,\rho)$ and $(U,U^{\prime}) \overset{d}{=} \mathrm{CorBin}(m,m,p;m,\rho)$ be such that $(X, Y)$ is independent of $(U,U^{\prime})$. Then the left hand side (in the corollary-statement) equals to
\begin{align*}
    & \mathbb{P}( X+U, Y+U^{\prime} \geq l ) - \mathbb{P}(X,Y \geq l) \leq \mathbb{P}( X+U \geq l > X ) + \mathbb{P}( Y+U^{\prime} \geq l > Y ) \\
    = \ & 2 ( \mathbb{P} ( \mathrm{Bin}(N+m,p) \geq l ) - \mathbb{P}( \mathrm{Bin}(N,p) \geq l ) ) \leq \frac{4lm}{N} \mathbb{P}(\mathrm{Bin}(N,p) \geq l)\,,
\end{align*}
where the last inequality follows from Lemma~\ref{lemma-difference-Bernoulli-tail-I}. This gives the desired bound.
\end{proof}

\begin{Lemma} \label{lemma-difference-Bernoulli-tail-II}
For all $N \gg 1,M,m,l \in \mathbb{N}$ and $p,\epsilon>0$ we have
\begin{align*}
     \mathbb{P}(\mathrm{Bin}(N+m,p) \geq l ) - \mathbb{P}( \mathrm{Bin}(N,p) \geq l ) 
    \lesssim \epsilon^{2} + \big( mp + \epsilon^{-1} \sqrt{ mp(1-p) } \big)  \frac{\log N}{\sqrt{Np(1-p)}}   \,.
\end{align*}
\end{Lemma}
\begin{proof}
Writing $Q=mp + \epsilon^{-1}  \sqrt{ mp(1-p) }$, we have that the left hand side is bounded by $\mathbb{P}( \mathrm{Bin}(m,p) \geq Q ) + \mathbb{P}( l-Q < \mathrm{Bin}(N,p) < l )$.
Applying Chebyshev's inequality, we have $\mathbb{P}( \mathrm{Bin}(m,p) \geq Q) \leq \epsilon^2 $. In addition, we have
\begin{align}\label{eq-bound-binomial-difference}
    & \mathbb{P}( l-Q < \mathrm{Bin}(N,p) < l ) \leq Q \max_{k>0} \{ \mathbb{P}( \mathrm{Bin}(N,p) = k ) \} \nonumber \\
    \leq \ & Q \cdot \max_{k \in \{ [Np], [Np]+1 \}, k \neq 0} \Big\{  \binom{ N }{ k } p^{k} (1-p)^{N-k}  \Big\} \lesssim \frac{Q \log N}{\sqrt{Np(1-p)}}\,,   
\end{align}
where $[Np] = \max\{ r : r \leq Np\}$. Here the last inequality above can be verified as follows: if $Np(1-p)= O(1)$, then $\frac{\log N}{\sqrt{Np(1-p)}} \gg 1$ (and thus the bound holds); if $Np(1-p) \gg 1$, then by Sterling's formula we have $\binom{N}{k} p^{k} (1-p)^{N-k} = O( \frac{\log N}{\sqrt{Np(1-p)}} )$, as desired.
\end{proof}

\begin{Corollary} \label{corollary-difference-CorBernoulli-tail-II}
For $N \gg 1,M,m_1,m_2 = o(N) ,l \in \mathbb{N}$ and $p,\epsilon > 0$, we have
\begin{align*}
    & \mathbb{P}( \mathrm{CorBin} (N+m_1,N+m_1,p;M+m_2,\rho) \geq (l,l) ) - \mathbb{P}( \mathrm{CorBin}(N,N,p;M,\rho) \geq (l,l) ) \\
    \lesssim \ & \epsilon^2 + ( 2m_1 p + 4 \epsilon^{-1} \sqrt{ m_1 p(1-p) } + 4 \epsilon^{-1} \sqrt{m_2 p(1-p)} ) \frac{\log N}{\sqrt{Np(1-p)}} \,.
\end{align*}
\end{Corollary}
\begin{proof}
Let $(X,Y) \overset{d}{=} \mathrm{CorBin} (N-m_2,N-m_2,p;M,\rho)$, $(Z,Z^{\prime}) \overset{d}{=} \mathrm{CorBin}(m_2,m_2,p;0,\rho)$, $(W,W^{\prime}) \overset{d}{=} \mathrm{CorBin}(m_2,m_2,p;m_2,\rho)$ and $(U,U^{\prime}) \overset{d}{=} \mathrm{CorBin}(m_1,m_1,p;0,\rho)$ be independent pairs of variables. Let $Q_1 = m_1 p + 2\epsilon^{-1} \sqrt{ m_1 p(1-p) } $ and $Q_2 = 2\epsilon^{-1} \sqrt{m_2 p(1-p)}$. Then the difference of probabilities in the statement is equal to
\begin{align*}
    & \mathbb{P}( X+U+W, Y+U^{\prime}+W^{\prime} \geq l) - \mathbb{P}( X+Z, Y+Z^{\prime} \geq l)  \\
    \leq \ & \mathbb{P}( X+U+W \geq l > X+Z ) + \mathbb{P}( Y+U^{\prime}+W^{\prime} \geq l > Y+Z^{\prime} ) \\
    \leq \ & 2 \big( \mathbb{P}( |W-Z| > Q_2 ) + \mathbb{P}(U > Q_1 ) + \mathbb{P}(  l-Q_1-Q_2 \leq X + Z < l ) \big) \\
    \lesssim \ & 4 \epsilon^2 + 2(Q_1+Q_2) (\log N) / \sqrt{Np(1-p)} \,,
\end{align*}
where the last transition follows from Chebyshev's inequality and \eqref{eq-bound-binomial-difference}.
\end{proof}
Recall \eqref{equ-def-iter-vartheta-varsigma}. Define the targeted approximation error in the $a$-th iteration of the initialization by
    \begin{align}
        \Lambda_{a} = 100^{a} (n \Hat{q})^{-\frac{1}{2}} (\log n)  \vartheta_{a}  \mbox{ for } 0 \leq a \leq \chi   \,.
        \label{equ-def-Lambda}
    \end{align}
\begin{Lemma}
\label{lemma-concentration-iterative-aleph-Upsilon}
    The following hold with probability $1 - o(1)$ for all $0\leq a\leq \chi$: \\
        \noindent (1) $\Big| \frac{|\aleph^{(a)}_k|}{n} - \vartheta_{a} \Big|, \Big| \frac{|\Upsilon^{(a)}_k|}{n} - \vartheta_a \Big| \leq \Lambda_{a}$ for $1 \leq k \leq K_0$; \\
        \noindent (2) $\Big| \frac{|\aleph^{(a)}_k \cap \aleph^{(a)}_l|}{n} - \vartheta_a^2 \Big| , \Big| \frac{|\Upsilon^{(a)}_k \cap \Upsilon^{(a)}_l|}{n} - \vartheta_a^2 \Big|, \Big| \frac{|\pi(\aleph^{(a)}_k) \cap \Upsilon^{(a)}_l|}{n} - \vartheta_a^2 \Big| \leq \Lambda_{a} $ for $1 \leq k \neq l \leq K_0$; \\
        \noindent (3) $\Big| \frac{|\pi(\aleph^{(a)}_k) \cap \Upsilon^{(a)}_k|}{n} - \varsigma_{a} \Big| \leq \Lambda_{a} $ for $1 \leq k \leq K_0$.
\end{Lemma}
\begin{proof}
The proof is by induction on $a$. The base case for $a=0$ is trivial. Now suppose that Items (1), (2) and (3) hold up to some $a \leq \chi-1$ and we wish to prove that (1), (2) and (3) hold with probability $1-o(1)$ for $a+1$. To this end, applying \eqref{eq-def-chi} and Lemma~\ref{lemma-property-vartheta-varsigma} we have $\vartheta_a = \Theta( n^{-1} (n \Hat{q})^{\chi-1} ) \ll \frac{1}{n\Hat{q}}$. Recall the definition of $\mathrm{REV}^{(a)}$ as in \eqref{eq-def-REV}, which records the collection of vertices explored by our algorithm. By the induction hypothesis we know 
\begin{align}\label{eq-REV-approximation}
   |\mathrm{REV}^{(a)}| \leq 4K_0 \vartheta_a n \ll \Lambda_{a+1} n \,,
\end{align}
where the last transition follows from Lemma~\ref{lemma-property-vartheta-varsigma}.
Thus, it suffices to control the concentration of $|\aleph^{(a+1)}_k \setminus \mathrm{REV}^{(a)}|$ in order to control that for $|\aleph^{(a+1)}_k|$. Note that
\begin{align*}
    \frac{|\aleph^{(a+1)}_k \setminus \mathrm{REV}^{(a)}|}{n} - \vartheta_{a+1} = \frac{1}{n} \sum_{ u \in V \setminus \mathrm{REV}^{(a)} }  \Big( \mathbf{1}_{ \{ u \in \aleph^{(a+1)}_k \} } - \vartheta_{a+1} \Big) + O(\vartheta_{a+1} \vartheta_a) \,,
\end{align*}
where $\vartheta_{a+1} \vartheta_a \ll \vartheta_{a+1}/ n\Hat{q} \ll \Lambda_{a+1}$ by Lemma~\ref{lemma-property-vartheta-varsigma}. Since the indicators in the above sum are measurable with respect to $\{ \overrightarrow{G}_{v,u} : v \in \aleph^{(a)}_k\}$, we see that conditioned on a realization of $\{ \aleph^{(a)}_k, \Upsilon^{(a)}_k \}$ we have that $\{ \mathbf{1}_{ \{ u \in \aleph^{(a+1)}_k \} } : u \in V \setminus \mathrm{REV}^{(a)}\}$ is a collection of i.i.d.\ Bernoulli random variables with parameter given by
\begin{align*}
    p^{(a+1)}_{k} = \mathbb{P}\big( u \in \aleph^{(a+1)}_k \big) = \mathbb{P}\big( \mathrm{Bin} (|\aleph^{(a)}_k|, \Hat{q}) \geq 1 \big) \,.
\end{align*}
By the induction hypothesis, we have $\big| |\aleph^{(a)}_k| - n \vartheta_a \big| \leq n \Lambda_{a}$. Combined with Lemma~\ref{lemma-difference-Bernoulli-tail-I}, it yields that
\begin{align}
    & \big| \vartheta_{a+1} - p^{(a+1)}_{k} \big| \leq \big| \mathbb{P} \big(\mathrm{Bin} (n\vartheta_a, \Hat{q}) \geq 1 \big) - \mathbb{P} \big( \mathrm{Bin} (n\vartheta_a + n \Lambda_{a}, \Hat{q}) \geq 1 \big) \big| \nonumber \\
    \leq \ & \frac{ 2 n \Lambda_{a} }{ n \vartheta_a }  \mathbb{P}( \mathrm{Bin} (n\vartheta_a, \Hat{q}) \geq 1 ) \overset{\eqref{equ-def-iter-vartheta-varsigma}, \eqref{equ-def-Lambda}}{\leq} \frac{1}{10} \Lambda_{a+1} \,.
    \label{equ-bound-p-a+1-minus-vartheta-a+1}
\end{align}
Thus, we may apply  Lemma~\ref{lemma-Bernstein-inequality} and get that
\begin{align}
    & \mathbb{P}\Big( \Big| \frac{|\aleph^{(a+1)}_k|}{n} - \vartheta_{a+1} \Big| > \Lambda_{a+1} \Big) \overset{\eqref{eq-REV-approximation}}{\leq} \mathbb{P} \Big( \Big| \frac{|\aleph^{(a+1)}_k \setminus \mathrm{REV}^{(a)}|}{n} - \vartheta_{a+1} \Big| > \frac{9}{10} \Lambda_{a+1} \Big) \nonumber \\
    \overset{\eqref{equ-bound-p-a+1-minus-vartheta-a+1}} {\leq} & \mathbb{P} \Big( \frac{1}{n} \Big|  \mathrm{Bin}(n-|\mathrm{REV}^{(a)}|, p^{(a+1)}_{k}) - ( n-|\mathrm{REV}^{(a)}| ) p^{(a+1)}_{k} \Big| > \frac{1}{2} \Lambda_{a+1} \Big) \nonumber \\
    \leq \ & 2 \exp \Big\{ - \frac{ (\frac{1}{2} n \Lambda_{a+1})^2 }{ 2 (n p^{(a+1)}_{k} +  n\Lambda_{a+1}/3 ) } \Big\} \leq 2 \exp \{ - \Hat{q}^{-1} \vartheta_{a+1} (\log n)^2 \} \,.
    \label{equ-concentration-result-1-a>0}
\end{align}
Similar results hold for $\Upsilon^{(a+1)}_k$. We now move to Item (2). Similarly, conditioned on a realization of $\{ \aleph^{(a)}_k, \Upsilon^{(a)}_k \}$, we have that
\begin{align*}
    \frac{ |\aleph^{(a+1)}_k \cap \aleph^{(a+1)}_l \setminus \mathrm{REV}^{(a)}| }{n} = \frac{1}{n} \sum_{u \in V \setminus \mathrm{REV}^{(a)}} \mathbf{1}_{ \{ u \in \aleph^{(a+1)}_k \cap \aleph^{(a+1)}_l \} }
\end{align*}
is a (normalized) sum of i.i.d.\ Bernoulli random variables with parameter give by
\begin{align*}
    p^{(a+1)}_{k,l} = & \mathbb{P} \Big(u \in \aleph^{(a+1)}_k \cap \aleph^{(a+1)}_l \Big) = \mathbb{P} \Big( \sum_{ w \in \aleph^{(a)}_k } \overrightarrow{G}_{w,u} , \sum_{ w \in \aleph^{(a)}_l } \overrightarrow{G}_{w,u} \geq 1 \Big) \,, 
\end{align*}
where $\big( \sum_{w \in \aleph^{(a)}_k} \overrightarrow{G}_{w,u} , \sum_{w \in \aleph^{(a)}_l} \overrightarrow{G}_{w,u} \big) \overset{d}{=} \mathrm{CorBin}(|\aleph^{(a)}_k|, |\aleph^{(a)}_l|, \Hat{q} ; |\aleph^{(a)}_k \cap \aleph^{(a)}_l|, \Hat{\rho} )$. 
By the induction hypothesis we have $\big| |\aleph^{(a)}_k|- n\vartheta_a \big|, \big| |\aleph^{(a)}_k|- n\vartheta_a \big|, \big| |\aleph^{(a)}_k \cap \aleph^{(a)}_l| - n \vartheta_a^2 \big| \leq n \Lambda_{a}$. In addition,  by Lemma~\ref{lemma-property-vartheta-varsigma}  we have  $\vartheta_a^2 \leq \vartheta_a / n \Hat{q} \ll \Lambda_a$ and thus  $|\aleph^{(a)}_k \cap \aleph^{(a)}_l| \leq 1.1 n \Lambda_a$. Combined with Corollary~\ref{corollary-difference-CorBernoulli-tail-I}, these yield that
\begin{align*}
    \Big| \mathbb{P}\Big(u \in \aleph^{(a+1)}_k \cap \aleph^{(a+1)}_l \Big) - \vartheta_{a+1}^2 \Big| \leq \frac{5 \Lambda_{a}}{\vartheta_{a}} \vartheta_{a+1} \leq \frac{1}{10} \Lambda_{a+1} \,.
\end{align*}
Applying Lemma~\ref{lemma-Bernstein-inequality} again, we get that
\begin{align}
    & \mathbb{P} \Big(\Big| \frac{|\aleph^{(a+1)}_k \cap \aleph^{(a+1)}_l|}{n} - \vartheta_{a+1}^2 \Big| > \Lambda_{a+1} \Big)  \nonumber \\
    \leq \ & \mathbb{P}\Big( \frac{1}{n} \Big| \mathrm{Bin}( n-|\mathrm{REV}^{(a)}|, p^{(a+1)}_{k,l} ) - (n-|\mathrm{REV}^{(a)}|) p^{(a+1)}_{k,l} \Big| > \frac{1}{2} \Lambda_{a+1} \Big)  \nonumber \\
    \leq \ & 2 \exp \{ - \Hat{q}^{-1} \vartheta_{a+1} (\log n)^2 \} \,.
    \label{equ-concentration-result-2-a>0}
\end{align}
The terms $ \frac{|\aleph^{(a+1)}_k \cap \aleph^{(a+1)}_l|}{n}, \frac{|\Upsilon^{(a+1)}_k \cap \Upsilon^{(a+1)}_l|}{n} $ can be bounded in the same way. We next turn to Item (3). In order to bound $\big| \frac{|\aleph^{(a+1)}_k \cap \Upsilon^{(a+1)}_k|}{n} - \varsigma_{a+1} \big|$, note that 
\begin{align*}
    \frac{|\aleph^{(a+1)}_k \cap \Upsilon^{(a+1)}_k \setminus \mathrm{REV}^{(a)}|}{n} - \varsigma_{a+1} 
    = \frac{1}{n} \sum_{ u \in V \setminus \mathrm{REV}^{(a)} }  \Big( \mathbf{1}_{ \{ u \in \aleph^{(a+1)}_k \cap \Upsilon^{(a+1)}_k \} } - \varsigma_{a+1} \Big) + O(\varsigma_{a+1} \vartheta_{a}) \,.
\end{align*}
Given a realization of $\{ \aleph^{(a)}_k, \Upsilon^{(a)}_k \}$, we have $\{\mathbf{1}_{ \{ u \in \aleph^{(a+1)}_k \cap \Upsilon^{(a+1)}_k \}}\}$ is a collection of i.i.d.\ Bernoulli variables, with parameter (by Lemma~\ref{corollary-difference-CorBernoulli-tail-I} and the induction hypothesis again)
\begin{align*}
     \Big| \mathbb{P}\Big( u \in \aleph^{(a+1)}_k \cap \Upsilon^{(a+1)}_k \Big) - \varsigma_{a+1} \Big| =  \Big| \mathbb{P} ( X,Y \geq 1) - \varsigma_{a+1} \Big| \leq \frac{ 5 \Lambda_{a} }{ \vartheta_{a} } \vartheta_{a+1} \leq \frac{1}{10} \Lambda_{a+1}  \,,
\end{align*}
where $(X, Y) \sim \mathrm{CorBin}(|\aleph^{(a)}_k|, |\Upsilon^{(a)}_k|, q; |\aleph^{(a)}_k \cap \Upsilon^{(a)}_k|; \rho)$.
By Lemma~\ref{lemma-Bernstein-inequality} again, we get that
\begin{align}
    & \mathbb{P} \Big(\Big| \frac{|\aleph^{(a+1)}_k \cap \Upsilon^{(a+1)}_k|}{n} - \varsigma_{a+1} \Big| > \Lambda_{a+1} \Big) \nonumber \\
    \leq \ & \mathbb{P} \Big( \Big| \frac{|\aleph^{(a+1)}_k \cap \Upsilon^{(a+1)}_k \setminus \mathrm{REV}^{(a)}|}{n} - \varsigma_{a+1} \Big| > \frac{9}{10} \Lambda_{a+1} \Big) \nonumber \\
    \leq \ & 2 \exp \Big\{ - \frac{ (\frac{1}{2} \Lambda_{a+1} n)^2 }{ 2(n \varsigma_{a+1} + n \Lambda_{a+1}/3 )}  \Big\} \leq 2 \exp \{ - \Hat{q}^{-1} \vartheta_{a+1} (\log n)^2 \} \,.
    \label{equ-concentration-result-3-a>0}
\end{align}
Combining \eqref{equ-concentration-result-1-a>0}, \eqref{equ-concentration-result-2-a>0}, and \eqref{equ-concentration-result-3-a>0} and applying a union bound, we see that (assuming (1), (2) and (3) hold for $a$) the conditional probability for (1), (2) or (3) to fail for $a+1$ is at most $2K_0^2 \exp\{ - (\log n)^2 \}$ since $\vartheta_{a+1} \geq \vartheta_{1} = \Hat{q}$. Therefore, we complete the proof of Lemma \ref{lemma-concentration-iterative-aleph-Upsilon} by induction (note that $\chi = O(1)$).
\end{proof}

\begin{Lemma}
\label{lemma-E-0-holds}
We have $\P(\mathcal{E}_{0}) = 1-o(1)$.
\end{Lemma}
\begin{proof}
By Lemma~\ref{lemma-concentration-iterative-aleph-Upsilon}, with probability $1-o(1)$ we have 
$$ |\mathrm{REV}| \leq 4K_0 n \vartheta_{\chi} \ll n \vartheta \Delta_{0} \,. $$
Here the last inequality can be derived as follows: from  Lemma~\ref{lemma-property-vartheta-varsigma}, we have either $\vartheta=\Theta(1)$ or $\vartheta_{\chi} \leq n^{-\alpha+o(1)}$. In addition, when $\vartheta=\Theta(1)$ we have (recall \eqref{equ-def-vartheta-varsigma})
$$ n \vartheta \Delta_{0} = \Theta( n \Delta_0 ) \overset{\eqref{equ-def-delta}}{\gg} n e^{-(\log \log n)^{100}} \overset{\eqref{eq-def-chi}}{\gg} 4K_0 n \vartheta_{\chi} \,; $$
And when $\vartheta_{\chi} \leq n^{-\alpha+o(1)}$ we have
$$ 4K_0 n \vartheta_{\chi} = n^{1-\alpha+o(1)} \ll n e^{-(\log \log n)^{200}} \overset{\eqref{eq-def-chi}, \eqref{equ-def-delta}}{\ll} n \vartheta \Delta_{0} \,. $$
Provided with the preceding bound on $|\mathrm{REV}|$, it suffices to analyze $\frac{|\Gamma^{(0)}_k \setminus \mathrm{REV}|}{n}$, whose conditional distribution given $\{ \aleph^{(a)}_k, \Upsilon^{(a)}_k: 1\leq k\leq K_0, 0\leq a\leq \chi\}$ is the same as the sum of $n-|\mathrm{REV}|$ i.i.d.\ Bernoulli variables with parameter given by
\begin{align*}
    \mathbb{P}( u \in \Gamma^{(0)}_k ) = \mathbb{P}( \mathrm{Bin}(|\aleph^{(\chi)}_k|,\Hat{q}) \geq \mathtt{d}_{\chi} ) \,.
\end{align*}
By Lemma \ref{lemma-concentration-iterative-aleph-Upsilon} again, we have $\Big| |\aleph^{(\chi)}_k|- \vartheta_{\chi} n \Big| \leq n \Lambda_{\chi}$ with probability $1-o(1)$. Provided with this and applying Lemma~\ref{lemma-difference-Bernoulli-tail-II} (with $\epsilon=\vartheta \Delta_0$), we get that
\begin{align*}
    | \mathbb{P}( u \in \Gamma^{(0)}_k ) - \vartheta|
    \leq & | \mathbb{P} ( \mathrm{Bin}(\vartheta_{\chi} n + n \Lambda_{\chi} ,\Hat{q}) \geq \mathtt{d}_{\chi} ) - \mathbb{P} (\mathrm{Bin}(\vartheta_{\chi} n, \Hat{q}) \geq \mathtt{d}_{\chi} )  | \\
    \leq & (\vartheta \Delta_0)^2 + (\log n) \frac{ n \Lambda_{\chi} \Hat{q} + 2 (\vartheta \Delta_0)^{-1} \sqrt{ n \Lambda_{\chi} \Hat{q} }  }{ \sqrt{\vartheta_{\chi} n \Hat{q}} } \ll \vartheta \Delta_0\,,
\end{align*}
where we used \eqref{equ-def-delta} and the inequalities $\frac{ n \Lambda_{\chi} \Hat{q}  }{ \sqrt{\vartheta_{\chi} n \Hat{q}} } = 100^\chi\sqrt{\vartheta_{\chi}} \log n \ll \vartheta \Delta_{0}/\log n$ (by Lemma~\ref{lemma-property-vartheta-varsigma}) as well as $\frac{ (\vartheta \Delta_0)^{-1} \sqrt{ n \Lambda_{\chi} \Hat{q} }  }{ \sqrt{\vartheta_{\chi} n \Hat{q}} } = (n \Hat{q})^{- \frac{1}{4}} (\vartheta \Delta_0)^{-1} \ll \vartheta \Delta_0/\log n$. By Lemma~\ref{lemma-Bernstein-inequality}, 
\begin{align}
    & \mathbb{P} \Big(\Big| \frac{|\Gamma^{(0)}_k|}{n} - \vartheta \Big| > \vartheta \Delta_{0} \Big) = \mathbb{P} \Big( \Big| \frac{|\Gamma^{(0)}_k \setminus \mathrm{REV}|}{n} - \vartheta \Big| > \frac{9}{10} \vartheta \Delta_{0}  \Big)  \nonumber\\
    \leq \ & \mathbb{P} \Big(  \Big| \frac{1}{n} \mathrm{Bin}( n-|\mathrm{REV}|, \vartheta+o( \vartheta \Delta_0 ) ) - \vartheta \Big| > \frac{9}{10} \vartheta \Delta_{0}  \Big) \nonumber \\
    \leq \ & 2 \exp \Big\{ - \frac{ ( \frac{1}{2} \vartheta \Delta_0 n )^2 }{ 2 (n \vartheta +  \vartheta \Delta_0 n /3 ) }  \Big\} \leq 2 \exp \{ -2 \vartheta \Delta_0^2 n \} \label{eq-Gamma-0-concentration}\,.
\end{align}
We can obtain the concentration for $\frac{|\Pi^{(0)}_k|}{n}$, $\frac{|\Gamma^{(0)}_k \cap \Gamma^{(0)}_l|}{n}$, $\frac{|\Pi^{(0)}_k \cap \Pi^{(0)}_l|}{n}$ and $\frac{| \pi(\Gamma^{(0)}_k) \cap \Pi^{(0)}_l|}{n}$ similarly. For instance, for $\frac{|\pi(\Gamma^{(0)}_k) \cap \Pi^{(0)}_l|}{n}$, we note that given $\{ \aleph^{(a)}_k, \Upsilon^{(a)}_k: 1 \leq k \leq K_0, 0 \leq a \leq \chi\}$,
\begin{align*}
    \frac{|\Gamma^{(0)}_k \cap \pi^{-1}(\Pi^{(0)}_l) \setminus \mathrm{REV}|}{n} = \sum_{ u \in V \setminus \mathrm{REV} } \mathbf{1}_{ \{ u \in \Gamma^{(0)}_k \cap \pi^{-1}(\Pi^{(0)}_l) \} } 
\end{align*}
is a sum of i.i.d.\ Bernoulli  variables with parameter given by
\begin{align*}
    \mathbb{P}(u \in \Gamma^{(0)}_k \cap \pi^{-1}(\Pi^{(0)}_l) ) = \mathbb{P}(  \mathrm{CorBin}( |\aleph^{(\chi)}_k|, |\Upsilon^{(\chi)}_l|, \Hat{q}; |\aleph^{(0)}_k \cap \pi^{-1}(\Upsilon^{(0)}_l)| , \Hat{\rho} ) \geq (\mathtt{d}_{\chi},\mathtt{d}_{\chi})  )\,.
\end{align*}
By Lemma~\ref{lemma-concentration-iterative-aleph-Upsilon}, we have $\big| |\aleph^{(\chi)}_k| - n \vartheta_{\chi} \big|, \big| |\Upsilon^{(\chi)}_k| - n \vartheta_{\chi} \big|, \big| |\aleph^{(0)}_k \cap \pi^{-1}(\Upsilon^{(0)}_l)| - n \varsigma_{\chi} \big| \leq n \Lambda_{\chi}$ with probability $1-o(1)$. Provided with this and applying Corollary~\ref{corollary-difference-CorBernoulli-tail-II} again (with $\epsilon = \vartheta \Delta_0$), we get that $\mathbb{P}( u \in \Gamma^{(0)}_k \cap \pi^{-1}(\Pi^{(0)}_l) )= \varsigma + o(\vartheta \Delta_0)$. Thus we can obtain a similar concentration bound using Lemma~\ref{lemma-Bernstein-inequality}. We omit further details due to similarity. 

By \eqref{eq-Gamma-0-concentration} (and its analogues) and a union bound, we deduce that
\begin{equation}\label{eq-mathcal-E-0-bound}
    \mathbb{P} (\mathcal{E}_0^{c}) \leq 20 K_0^2 \exp \{ -2 \vartheta \Delta_0^2 n \} = o(1) \,,
\end{equation}
where for the last step we recalled \eqref{equ-def-delta} and Lemma~\ref{lemma-property-vartheta-varsigma}. This completes the proof.
\end{proof}

\subsection{Density comparison}
\label{sec:density-compare}
Our proof of the admissibility along the iteration relies on a direct comparison of the smoothed Bernoulli density and the Gaussian density, which then allows us to  use the techniques developed in \cite{DL22+} for correlated Gaussian Wigner matrices. Recall \eqref{eq-def-mathfrak-S-t}, \eqref{equ-def-g-Tilde-D} and \eqref{eq-def-mathcal-F-t}. Our main result in this subsection is Lemma~\ref{lemma-bound-conditional-unbiased-density} below, and we need to introduce more notations before its statement.

For a random variable $X$ and a $\sigma$-field $\mathcal F$, we denote by $p_{\{X\mid \mathcal F\}}$ the conditional density of $X$ given $\mathcal F$. For a realization $\Xi_t = \{\xi^{(s)}_k, \zeta^{(s)}_k : s \leq t, 1\leq k\leq K_s \}$ of $\{ \Gamma^{(s)}_k, \Pi^{(s)}_k : s \leq t, 1\leq k\leq K_s \}$ and a realization $\mathrm{B}_{t-1}$ of $\mathrm{BAD}_{t-1}$, we define vector-valued functions $\varphi^{(s)}_{v} (\Xi_t, \mathrm{B}_{t-1}), \psi^{(s)}_{v} (\Xi_t, \mathrm{B}_{t-1})$ for $v \not \in \mathrm{B}_{t-1}$ and $0\leq s\leq t$, where for $1\leq k\leq K_s$ the $k$-th component is given by  
\begin{equation}\label{eq-def-vector-value-function-1}
\begin{aligned}
    \varphi^{(s)}_{v, k} (\Xi_t, \mathrm{B}_{t-1}) &= \frac{1}{\sqrt{ (\mathfrak{a}_s - \mathfrak{a}_s^2) n \Hat{q}(1-\Hat{q}) }} \sum_{u \not \in \mathrm{B}_{t-1}} (\mathbf{1}_{u \in \xi^{(s)}_k}-\mathfrak{a}_s) (\overrightarrow{G}_{v,u}-\Hat{q}) \,, \\
    \psi^{(s)}_{v, k} (\Xi_t, \mathrm{B}_{t-1}) &= \frac{1}{\sqrt{ (\mathfrak{a}_s - \mathfrak{a}_s^2) n \Hat{q}(1-\Hat{q}) }} \sum_{u \not \in \mathrm{B}_{t-1}} (\mathbf{1}_{u \in \xi^{(s)}_k}-\mathfrak{a}_s) \overrightarrow{Z}_{v,u}\,.
\end{aligned}
\end{equation}
Similarly, we define $\varphi^{(s)}_{\pi(v)} (\Xi_t, \mathrm{B}_{t-1}), \psi^{(s)}_{\pi(v)} (\Xi_t, \mathrm{B}_{t-1})$ where for $1 \leq k \leq K_s$ the $k$-th component is given by  
\begin{equation}\label{eq-def-vector-value-function-2}
\begin{aligned}
    \varphi^{(s)}_{\pi(v), k} (\Xi_t, \mathrm{B}_{t-1}) &= \frac{1}{\sqrt{ (\mathfrak{a}_s - \mathfrak{a}_s^2) n \Hat{q}(1-\Hat{q}) }} \sum_{u \not \in \mathrm{B}_{t-1}} (\mathbf{1}_{\pi(u) \in \zeta^{(s)}_k}-\mathfrak{a}_s) (\overrightarrow{\mathsf G}_{\pi(v), \pi(u)}-\Hat{q}) \,, \\
    \psi^{(s)}_{\pi(v), k} (\Xi_t, \mathrm{B}_{t-1}) &= \frac{1}{\sqrt{ (\mathfrak{a}_s - \mathfrak{a}_s^2) n \Hat{q}(1-\Hat{q}) }} \sum_{u \not \in \mathrm{B}_{t-1}} (\mathbf{1}_{\pi(u) \in \zeta^{(s)}_k}-\mathfrak{a}_s) \overrightarrow{\mathsf Z}_{\pi(v), \pi(u)}\,.
\end{aligned}
\end{equation}
In addition, define $\mathrm{B}_t= \mathrm{B}_t ( \Xi_{t}, \mathrm{B}_{t-1}, \overrightarrow{G}, \overrightarrow{\mathsf{G}}, W, \mathsf{W} )$ to be the corresponding realization of $\mathrm{BAD}_t$, i.e., $\mathrm{B}_t$ is the collection of vertices satisfying either of \eqref{equ-def-set-BIAS}, \eqref{equ-def-set-LARGE-(0)}, \eqref{equ-def-set-LARGE} and \eqref{equ-def-set-PRB} with $(\Gamma^{(s)}_k, \Pi^{(s)}_k)$ replaced by $(\xi^{(s)}_k, \zeta^{(s)}_k)$ and $\mathrm{BAD}_{t-1}$ replaced by $\mathrm{B}_{t-1}$. Define a random vector $\mathbf{X}^{\leq t} = \mathbf{X}^{ \leq t } (\mathrm{B}_t, \mathrm{B}_{t-1})$ by 
\begin{align*}
    \mathbf{X}^{\leq t} (s,k,v) = W^{(s)}_v (k) + \langle \eta^{(s)}_k, \varphi^{(s)}_v \rangle \mbox{ and } \mathbf{X}^{\leq t} (s,k,\pi(v)) = \mathsf{W}^{(s)}_{\pi(v)} (k) + \langle \eta^{(s)}_k, \varphi^{(s)}_{\pi(v)} \rangle
\end{align*}
where $0 \leq s \leq t, 1 \leq k \leq \frac{K_s}{12}$, and $v \not \in \mathrm{B}_{t-1}$ when $s < t$, and $v \not \in \mathrm{B}_{t}$ when $s=t$. Define $\mathbf{Y}^{\leq t}$ similarly by replacing $\varphi^{(s)}_v, \varphi^{(s)}_{\pi(v)}$ with $\psi^{(s)}_v, \psi^{(s)}_{\pi(v)}$ and replacing $W,\mathsf{W}$ with $\Tilde{W},\Tilde{\mathsf{W}}$. Let $\mathbf X^{=t}$ be the vector obtained from $\mathbf X^{\leq t}$ by keeping its coordinates with $s = t$, and let $\mathbf X^{<t}$ be the vector obtained from $\mathbf X^{\leq t}$ by keeping its coordinates with $s <t$. Define $\mathbf{Y}^{=t}$ and $\mathbf{Y}^{<t}$ with respect to $\mathbf{Y}^{\leq t}$ similarly. We also define $\Hat{\mathbf{X}}^{\leq t} (s,k,v) = \mathbf{X}^{\leq t} (s,k,v) - W^{(s)}_v (k)$ and $\Hat{\mathbf{Y}}^{\leq t} (s,k,v) = \mathbf{Y}^{\leq t} (s,k,v) - \Tilde{W}^{(s)}_v (k)$. Also, define $(\overrightarrow{G}_{\mathrm{B}}, \overrightarrow{\mathsf{G}}_{\mathrm{B}}) = (\overrightarrow{G}_{u,w}, \overrightarrow{\mathsf{G}}_{\pi(u),\pi(w)} : u \mbox{ or } w \in \mathrm{B}_{t-1})$, and define $(\overrightarrow{G}_{\setminus \mathrm{B}}, \overrightarrow{\mathsf{G}}_{\setminus \mathrm{B}}) = (\overrightarrow{G}_{u,w}, \overrightarrow{\mathsf{G}}_{\pi(u),\pi(w)} : u,w \not \in \mathrm{B}_{t-1})$. Denote by $(\overrightarrow{g}_{\mathrm{B}}, \overrightarrow{\mathsf{g}}_{\mathrm{B}})$  the realization of $(\overrightarrow{G}_{\mathrm{B}}, \overrightarrow{\mathsf{G}}_{\mathrm{B}})$. For any fixed realization $(\Xi_t, \mathrm{B}_{t}, \mathrm{B}_{t-1}, \overrightarrow{g}_{\mathrm{B}}, \overrightarrow{\mathsf{g}}_{\mathrm{B}})$, we further define $\mathtt{p}_{ \{ \mathbf{X}^{\leq t} \} } (x^{\leq t}) = p( {x}^{\leq t}, \Xi_t, \mathrm{B}_t, \mathrm{B}_{t-1}, \overrightarrow{g}_{\mathrm{B}}, \overrightarrow{\mathsf{g}}_{\mathrm{B}} )$ to be the conditional density as follows:
\begin{align*}
    \mathtt{p}_{ \{ \mathbf{X}^{\leq t} \} } (x^{\leq t}) = p_{ \{ \mathbf{X}^{\leq t}  \mid \mathrm{BAD}_t = \mathrm{B}_t; \mathrm{BAD}_{t-1} = \mathrm{B}_{t-1}; (\overrightarrow{G}_{\mathrm{B}}, \overrightarrow{\mathsf{G}}_{\mathrm{B}}) = (\overrightarrow{g}_{\mathrm{B}}, \overrightarrow{\mathsf{g}}_{\mathrm{B}}) \} } (x^{\leq t}) \,,
\end{align*}
where the support of $x^{\leq t}$ is consistent with the choice of $(\mathrm{B}_{t}, \mathrm{B}_{t-1})$ (i.e., $x^{\leq t}$ is a legitimate realization for $\mathbf{X}^{\leq t} = \mathbf{X}^{\leq t}(\mathrm{B}_{t}, \mathrm{B}_{t-1})$). 
Define $\mathtt{p}_{ \{ \mathbf{Y}^{\leq t} \} } (  {x}^{\leq t} )$ similarly but with respect to $\mathbf{Y}^{\leq t}$. For the purpose of truncation later, we say a realization $( \mathrm{B}_t, \mathrm{B}_{t-1} )$ for $( \mathrm{BAD}_t, \mathrm{BAD}_{t-1} )$ is an \emph{amenable} set-realization, if 
\begin{equation}
    \mathbb{P}( \mathrm{BAD}_t = \mathrm{B}_t, \mathrm{BAD}_{t-1} = \mathrm{B}_{t-1} ) \geq \exp \{ - n \Delta_t^{9} \}  \,. \label{equ-def-good-realization-LARGE}
\end{equation}
Also, we say $(\overrightarrow{g}_{\mathrm{B}}, \overrightarrow{\mathsf{g}}_{\mathrm{B}})$ is an \emph{amenable} bias-realization with respect to $(\mathrm{B}_t, \mathrm{B}_{t-1})$, if
\begin{equation}
    \mathbb{P}( \mathrm{BAD}_t = \mathrm{B}_t, \mathrm{BAD}_{t-1} = \mathrm{B}_{t-1} | (\overrightarrow{G}_{\mathrm{B}}, \overrightarrow{\mathsf{G}}_{\mathrm{B}}) = (\overrightarrow{g}_{\mathrm{B}}, \overrightarrow{\mathsf{g}}_{\mathrm{B}}) ) \geq \exp \{ - n \Delta_t^{8} \} \,.
    \label{equ-def-good-bias-realization}
\end{equation}
In addition, we say a realization $x^{\leq t} = x^{\leq t}( \mathrm{B}_t, \mathrm{B}_{t-1} )$ for $\mathbf{X}^{\leq t} (\mathrm{B}_t, \mathrm{B}_{t-1})$ is an \emph{amenable} variable-realization with respect to $(\mathrm{B}_t, \mathrm{B}_{t-1})$ and $(\overrightarrow{g}_{\mathrm{B}}, \overrightarrow{\mathsf{g}}_{\mathrm{B}})$, if it is consistent with the choice of $(\mathrm{B}_t,\mathrm{B}_{t-1})$ and $ (\overrightarrow{g}_{\mathrm{B}}, \overrightarrow{\mathsf{g}}_{\mathrm{B}})$, and (below the vector $x^{<t}$ is obtained by keeping the coordinates of $x^{\leq t}$ with $s<t$)
\begin{align}
    & \| x^{\leq t} \|_{\infty} \leq 2 (\log n) n^{\frac{1}{\log \log \log n}}, \label{equ-def-good-realization-mathbf-X-I} \\
    & \frac{ p_{ \{ \mathbf{Y}^{< t} | \mathrm{BAD}_t = \mathrm{B}_t , \mathrm{BAD}_{t-1} = \mathrm{B}_{t-1}, (\overrightarrow{G}_{\mathrm{B}}, \overrightarrow{\mathsf{G}}_{\mathrm{B}}) = (\overrightarrow{g}_{\mathrm{B}}, \overrightarrow{\mathsf{g}}_{\mathrm{B}}) \} } ( x^{< t} ) }{ p_{ \{ \mathbf{X}^{< t} | \mathrm{BAD}_t = \mathrm{B}_t, \mathrm{BAD}_{t-1} = \mathrm{B}_{t-1}, (\overrightarrow{G}_{\mathrm{B}}, \overrightarrow{\mathsf{G}}_{\mathrm{B}}) = (\overrightarrow{g}_{\mathrm{B}}, \overrightarrow{\mathsf{g}}_{\mathrm{B}}) \} } ( x^{< t} ) } \leq \exp \{ n \Delta_t^{10} \} \,,\label{equ-def-good-realization-mathbf-X-II}   \\
    \mbox{and } & \frac{ p_{ \{ \mathbf{Y}^{\leq t} | \mathrm{BAD}_t = \mathrm{B}_t, \mathrm{BAD}_{t-1} = \mathrm{B}_{t-1}, (\overrightarrow{G}_{\mathrm{B}}, \overrightarrow{\mathsf{G}}_{\mathrm{B}}) = (\overrightarrow{g}_{\mathrm{B}}, \overrightarrow{\mathsf{g}}_{\mathrm{B}}) \} } ( x^{\leq t} ) }{ p_{ \{ \mathbf{X}^{\leq t} | \mathrm{BAD}_t = \mathrm{B}_t, \mathrm{BAD}_{t-1} = \mathrm{B}_{t-1}, (\overrightarrow{G}_{\mathrm{B}}, \overrightarrow{\mathsf{G}}_{\mathrm{B}}) = (\overrightarrow{g}_{\mathrm{B}}, \overrightarrow{\mathsf{g}}_{\mathrm{B}}) \} } ( x^{\leq t} ) } \leq \exp \{ n \Delta_t^{10} \} \,.  \label{equ-def-good-realization-mathbf-X-III}
\end{align}

\begin{Lemma} \label{lem-good-set-good-variable-realization}
On the event $\mathcal{T}_t$, with probability $1 - O( \exp \{ - n \Delta_t^{10} \} )$ we have that 
\begin{itemize}
\item $( \mathrm{B}_t, \mathrm{B}_{t-1})$ sampled according to $(\mathrm{BAD}_t, \mathrm{BAD}_{t-1})$ is an amenable set-realization;
\item  $(\overrightarrow{g}_{\mathrm{B}}, \overrightarrow{\mathsf{g}}_{\mathrm{B}})$  sampled according to  $ \{ (\overrightarrow{G}_{\mathrm{B}}, \overrightarrow{\mathsf{G}}_{\mathrm{B}}) | \mathrm{BAD}_t = \mathrm{B}_t, \mathrm{B}_{t-1} = \mathrm{B}_{t-1} \}$  is an amenable bias-realization with respect to $(\mathrm{B}_t, \mathrm{B}_{t-1})$;
\item $x^{\leq t} (\mathrm{B}_t, \mathrm{B}_{t-1})$ sampled according to  $\{ \mathbf{X}^{\leq t} | \mathrm{BAD}_t = \mathrm{B}_t ; \mathrm{BAD}_{t-1} = \mathrm{B}_{t-1}; (\overrightarrow{G}_{\mathrm{B}}, \overrightarrow{\mathsf{G}}_{\mathrm{B}}) = (\overrightarrow{g}_{\mathrm{B}}, \overrightarrow{\mathsf{g}}_{\mathrm{B}})\}$ is an amenable variable-realization with respect to $(\mathrm{B}_t, \mathrm{B}_{t-1})$ and $(\overrightarrow{g}_{\mathrm{B}}, \overrightarrow{\mathsf{g}}_{\mathrm{B}})$.
\end{itemize}
\end{Lemma}
\begin{proof}
We first consider $(\mathrm{B}_t, \mathrm{B}_{t-1})$. On $\mathcal{T}_{t}$ we have $|\mathrm{BAD}_t|, |\mathrm{BAD}_{t-1}|  \leq n \Delta_t^{10} $. Note that $(\mathrm{BAD}_t, \mathrm{BAD}_{t-1})$ are two subsets of $V$ and each has at most $\binom{n}{n \Delta_t^{10} }$ possible values. Thus,
\begin{align*}
    \mathbb P(\mathrm{BAD}_t, \mathrm{BAD}_{t-1} \not\in \{\mbox{amenable set-realization}\}; \mathcal T_t ) \leq \Big(\binom{n}{n \Delta_t^{10} }\Big)^2 \exp \{ - n \Delta_t^9 \} \ll e^{ - \tfrac{1}{2} n \Delta_t^{9} } \,.
\end{align*}
Given an amenable set-realization $(\mathrm{B}_t, \mathrm{B}_{t-1})$, we now consider $(\overrightarrow{g}_{\mathrm{B}}, \overrightarrow{\mathsf{g}}_{\mathrm{B}})$. By Lemma~\ref{lemma-posterior-estimator},
\begin{align*}
    \mathbb{P}_{ (\overrightarrow{g}_{\mathrm{B}}, \overrightarrow{\mathsf{g}}_{\mathrm{B}}) \sim \{ (\overrightarrow{G}_{\mathrm{B}}, \overrightarrow{\mathsf{G}}_{\mathrm{B}}) | \mathrm{BAD}_t = \mathrm{B}_t, \mathrm{B}_{t-1} = \mathrm{B}_{t-1} \} } \big(   (\overrightarrow{g}_{\mathrm{B}}, \overrightarrow{\mathsf{g}}_{\mathrm{B}}) \mbox{ is bias-amenable} \big) \geq 1- e^{ -n\Delta_t^{-9} } \,. 
\end{align*}
Finally we consider $\mathbf{X}^{\leq t}$. Combining Markov's inequality and (below we write $\mathcal{A} = \{ \mathrm{BAD}_t = \mathrm{B}_t ; \mathrm{BAD}_{t-1} = \mathrm{B}_{t-1}; (\overrightarrow{G}_{\mathrm{B}}, \overrightarrow{\mathsf{G}}_{\mathrm{B}}) = (\overrightarrow{g}_{\mathrm{B}}, \overrightarrow{\mathsf{g}}_{\mathrm{B}}) \}$)
\begin{align*}
    \mathbb{E}_{ x^{\leq t} \sim \{ \mathbf{X}^{\leq t} \mid \mathcal{A} \} } \Big[ \frac{ p_{ \{ \mathbf{Y}^{\leq t} \mid \mathcal{A} \} } ( x^{\leq t} ) }{ p_{ \{ \mathbf{X}^{\leq t} \mid \mathcal{A} \} } ( x^{\leq t} ) } \Big] = 1 \,,
\end{align*}
we get that conditioned on $\mathcal A$, the (random) realization $x^{\leq t}$ for $\mathbf{X}^{\leq t}$ satisfies \eqref{equ-def-good-realization-mathbf-X-III} with probability at least $1 - \exp \{ - n \Delta_t^{10} \}$. Similarly, the realization $x^{\leq t}$ satisfies \eqref{equ-def-good-realization-mathbf-X-II} with probability at least $1 - \exp \{ - n \Delta_t^{10} \}$.
In addition, for any $s \leq t-1$ and $v \not \in \mathrm{B}_{t-1}$, recalling that $\mathrm{B}_{t-1}$ is the realization for $\mathrm{BAD}_{t-1}$ and using $\mathcal{T}_{t-1}$, \eqref{equ-def-set-LARGE-(0)} and \eqref{equ-def-set-LARGE}, we derive from the triangle inequality that 
\begin{align}
    | \mathbf{X}^{\leq t}(s,k,v) | & \leq |W^{(s)}_k(v)| + | \langle \eta^{(s)}_k, g_{t-2} \varphi^{(s)}_v \rangle | + \sum_{a=0}^{\log n} | \langle \eta^{(s)}_k, b_{t-1,a} \varphi^{(s)}_v \rangle | \nonumber \\ & 
    \leq (3+\log n) n^{\frac{1}{\log \log \log n}}  \,, \label{equ-bound-degree-time-s<t-not-BAD}
\end{align}
where $g_{t-2} \varphi^{(s)}_v (k)$ denotes the (first) expression \eqref{eq-def-vector-value-function-1} but the summation is taken for all $u \not \in\mathrm{BAD}_{t-2}$, $b_{t-1,0} \varphi^{(s)}_v$ denotes the expression \eqref{eq-def-vector-value-function-1} but the summation is taken for all $u\in \mathrm{LARGE}^{(0)}_{t-1} \cup \mathrm{BIAS}_{t-1} \cup \mathrm{PRB}_{t-1}$, and $b_{t-1,a} \varphi^{(s)}_v$ denotes the expression \eqref{eq-def-vector-value-function-1} but the summation is taken for all $u\in \mathrm{LARGE}^{(a)}_{t-1}$.
We have $| \mathbf{X}^{\leq t}(s,k,\pi(v)) | \leq  (3+\log n) n^{\frac{1}{\log \log \log n}}$ similarly. Since $\mathrm{LARGE}_t \subset \mathrm{B}_t$, for $v\not\in \mathrm{B}_t$ choosing $s=t$ in \eqref{equ-def-set-LARGE-(0)} gives that
\begin{align}
    |\mathbf{X}^{\leq t}(t,k,v)| \leq |W^{(t)}_k(v)| + | \langle \eta^{(t)}_k, \varphi^{(t)}_v \rangle | \leq 2 n^{\frac{1}{\log \log \log n}} \,. \label{equ-bound-degree-time-t-not-in-BAD}
\end{align}
Combining \eqref{equ-bound-degree-time-s<t-not-BAD} and \eqref{equ-bound-degree-time-t-not-in-BAD} we can show that $\{ \mathrm{BAD}_t = \mathrm{B}_t, \mathrm{BAD}_{t-1} = \mathrm{B}_{t-1} \}$ implies that $\| \mathbf{X}^{\leq t} \|_{\infty}, \| \mathbf{X}^{\leq t} (t,k,\pi(v)) \|_{\infty} \leq (3+\log n) n^{ \frac{1}{\log \log \log n} }$. Altogether, we complete the proof of the lemma.
\end{proof}

\begin{Lemma} {\label{lemma-bound-conditional-unbiased-density}}
For $t \leq t^*$, on the event $\mathcal{E}_{t} \cap \mathcal{T}_{t}$, fix an amenable set-realization $(\mathrm{B}_t, \mathrm{B}_{t-1})$, fix an amenable bias-realization $(\overrightarrow{g}_{\mathrm{B}}, \overrightarrow{\mathsf{g}}_{\mathrm{B}})$ with respect to $(\mathrm{B}_t,\mathrm{B}_{t-1})$, and fix an amenable variable-realization $x^{\leq t}$ with respect to $(\mathrm{B}_t,\mathrm{B}_{t-1})$ and $(\overrightarrow{g}_{\mathrm{B}}, \overrightarrow{\mathsf{g}}_{\mathrm{B}})$. Then we have (below the vector $x^{=t}$ is defined by keeping the coordinates of $x^{\leq t}$ such that $s=t$)
\begin{align} \label{eq-Lemma-3.9}
    \frac{  p_{ \{ \mathbf{X}^{=t} | \mathfrak{S}_{t-1}; \mathrm{BAD}_t = \mathrm{B}_t \} } ( x^{=t} ) }{ p_{ \{ \mathbf{Y}^{=t} |  \mathcal{F}_{t-1} \} } ( x^{=t} ) } = \exp \big\{ O \big( n \Delta_t^5 \big) \big\} \,.
\end{align}
\end{Lemma}
\begin{Remark}\label{rem-Y-conditioning-B-t}
Since $\{ \mathrm{BAD}_t = \mathrm{B}_t \}$ is measurable with respect to $( \overrightarrow{G}, \overrightarrow{\mathsf{G}}, W, \mathsf{W} )$ and thus is independent of $(\overrightarrow{Z}, \overrightarrow{\mathsf{Z}}, \Tilde{W}, \Tilde{\mathsf{W}})$, we have that ${p}_{ \{ \mathbf{Y}^{\leq t} \} } (x^{\leq t}) ={p}_{ \{ \mathbf{Y}^{\leq t} | \mathrm{BAD}_t = \mathrm{B}_t \} } (x^{\leq t})$. That is, we may add conditioning on $\mathrm{BAD}_{t} = \mathrm{B}_t$ in the denominator of \eqref{eq-Lemma-3.9}, but this will not change its conditional density.
\end{Remark}
The key to the proof of Lemma~\ref{lemma-bound-conditional-unbiased-density} is the following bound on the ``joint'' density. 

\begin{Lemma} {\label{lemma-joint-density-ratio-I}}
For $t \leq t^*$, on the event $\mathcal{E}_{t} \cap \mathcal{T}_{t}$, for an amenable set-realization $(\mathrm{B}_t, \mathrm{B}_{t-1})$, an amenable bias-realization $( \overrightarrow{g}_{\mathrm{B}}, \overrightarrow{\mathsf{g}}_{\mathrm{B}} )$ with respect to $(\mathrm{B}_t, \mathrm{B}_{t-1})$ and an amenable variable-realization $x^{\leq t}$ with respect to $(\mathrm{B}_t, \mathrm{B}_{t-1})$ and $( \overrightarrow{g}_{\mathrm{B}}, \overrightarrow{\mathsf{g}}_{\mathrm{B}} )$, we have 
\begin{equation} \label{equ-def-ratio-joint-density}
    \frac{\mathtt{p}_{ \{ \mathbf{X}^{\leq t} \} } ( x^{\leq t}  ) }{ \mathtt{p}_{ \{ \mathbf{Y}^{\leq t} \} } ( x^{\leq t} )}, \frac{ \mathtt{p}_{ \{ \mathbf{X}^{<t} \} } ( x^{<t} ) }{ \mathtt{p}_{ \{ \mathbf{Y}^{<t} \} } ( x^{<t} ) } = \exp \big\{ O ( n \Delta_t^5 ) \big\}  \,.
\end{equation}
\end{Lemma}
\begin{proof}[Proof of Lemma~\ref{lemma-bound-conditional-unbiased-density} assuming Lemma~\ref{lemma-joint-density-ratio-I}]
Applying Lemma~\ref{lemma-joint-density-ratio-I} for both $t$ and $t-1$, we get that 
\begin{align*}
    & \frac{  \mathtt{p}_{ \{ \mathbf{X}^{=t} | \mathbf{X}^{<t} \} } ( x^{=t} | x^{<t} ) }{ \mathtt{p}_{ \{ \mathbf{Y}^{=t} |  \mathbf{Y}^{<t} \} } ( x^{=t} | x^{<t} ) }   = \exp \{ O ( n \Delta_t^5 ) \}\,. 
\end{align*}
Since $p_{ \{ \mathbf{X}^{=t} | \mathfrak{S}_{t-1}, \mathrm{BAD}_t = \mathrm{B}_t \}} (x^{=t})  = \frac{  \mathtt{p}_{ \{ \mathbf{X}^{\leq t} \} (x^{\leq t}) }  }{ \mathtt{p}_{ \{ \mathbf{X}^{<t} \} (x^{<t}) } }$, we complete the proof of the lemma.
\end{proof}
 
The rest of this subsection is devoted to the proof of Lemma~\ref{lemma-joint-density-ratio-I}. Due to similarity, we only prove it for $\frac{ \mathtt{p}_{\{ \mathbf{X}^{\leq t} \}} (x^{\leq t}) }{ \mathtt{p}_{ \{ \mathbf{Y}^{\leq t} \} } (x^{\leq t}) }$. Also since $x^{\leq t}$ is an amenable variable-realization, the lower bound is obvious and as a result we only need to prove the upper bound. Note that $(\mathrm{BAD}_{t}, \mathrm{BAD}_{t-1}) = (\mathrm{BAD}_{t}( \overrightarrow{G}, \overrightarrow{\mathsf{G}}, W,\mathsf{W} ), \mathrm{BAD}_{t-1}( \overrightarrow{G}, \overrightarrow{\mathsf{G}}, W, \mathsf{W} ))$ is a function of $(\overrightarrow{G}, \overrightarrow{\mathsf{G}}, W, \mathsf{W})$, and that $( \overrightarrow{G}_{\mathrm{B}}, \overrightarrow{\mathsf{G}}_{\mathrm{B}} )$  is independent with $\mathbf{X}^{\leq t}(\mathrm{B}_t, \mathrm{B}_{t-1}, W,\mathsf{W})$ and with $( \overrightarrow{G}_{\setminus \mathrm{B}}, \overrightarrow{\mathsf{G}}_{\setminus \mathrm{B}}, W,\mathsf{W})$. Since for any independent random vectors  $X,Y$  and any function $f$
\begin{equation}\label{eq-fact-independence}
(f(X,Y)|X=x) \overset{d}{=} f(x,Y)\,,
\end{equation}
we can then apply \eqref{eq-fact-independence} and get that (note that the forms of $p$ are different in the equality below)
\begin{equation}\label{eq-density-mathcal-A-bar}
\mathtt{p}_{ \{ \mathbf{X}^{\leq t} \} } (x^{\leq t}) = p_{ \{ \mathbf{X}^{\leq t} | \Bar{\mathcal{A}} \} } (x^{\leq t})
\end{equation}
where $\Bar{\mathcal{A}} = \cap_{r=\in \{t-1, t\}}\big\{ \mathrm{BAD}_{r}( ( \overrightarrow{g}_{\mathrm{B}}, \overrightarrow{\mathsf{g}}_{\mathrm{B}} ), ( \overrightarrow{G}_{\setminus \mathrm{B}}, \overrightarrow{\mathsf{G}}_{\setminus \mathrm{B}} ), W,\mathsf{W} ) = \mathrm{B}_{r} \big\}$.
For an amenable set-realization $(\mathrm{B}_t, \mathrm{B}_{t-1})$ (for convenience we will drop $(\mathrm{B}_t, \mathrm{B}_{t-1})$ from notations in what follows), recall  \eqref{equ-def-degree-after-j-substitution} and define $\Hat{\mathbf{X}}^{\leq t}_{ \langle j \rangle }$ from  $\Hat{\mathbf{X}}^{\leq t}$ via the procedure in \eqref{equ-def-degree-after-j-substitution} (and similarly define $\mathbf{X}^{\leq t}_{ ( j ) }$ from $\mathbf X^{\leq t}$). For $1\leq j\leq N$, let $\mathcal M_j = \mathcal M_j(\mathrm{B}_t, \mathrm{B}_{t-1})$ be the event that 
\begin{align}
    \| \Hat{\mathbf{X}}^{\leq t}_{ \langle i \rangle}  \|_{\infty} \leq n^{ \frac{1}{\log \log \log n}} \mbox{ for } j \leq i \leq N \,.  \label{equ-def-truncate-event-mathcal-M}
\end{align}
Recalling Remark~\ref{rem-Y-conditioning-B-t}, we get from \eqref{eq-density-mathcal-A-bar} that
\begin{align*}
    \frac{ \mathtt{p}_{ \{ \mathbf{X}^{\leq t} \} } (x^{\leq t}) }{  \mathtt{p}_{ \{ \mathbf{Y}^{\leq t} \} } (x^{\leq t})  } = \frac{ p_{ \{ \mathbf{X}^{\leq t} | \Bar{\mathcal{A}} \} } (x^{\leq t}) }{ p_{ \{ \mathbf{Y}^{\leq t} \} } (x^{\leq t}) } =
    \frac{ {p}_{ \{ \mathbf{X}^{\leq t} | \Bar{\mathcal{A}} \} } (x^{\leq t}) }{  {p}_{ \{ \mathbf{X}^{\leq t} | \mathcal{M}_{1} \} } (x^{\leq t}) } \cdot \frac{ {p}_{ \{ \mathbf{X}^{\leq t} | \mathcal{M}_1 \} } (x^{\leq t}) }{ {p}_{ \{ \mathbf{Y}^{\leq t} \} } (x^{\leq t}) } \,.
\end{align*}
Note that $\Bar{\mathcal{A}} \subset \mathcal{M}_{1}$.  Applying \eqref{eq-fact-independence} with  $f=(\mathrm{BAD}_t,\mathrm{BAD}_{t-1})$, $X=( \overrightarrow{G}_{\mathrm{B}}, \overrightarrow{\mathsf{G}}_{\mathrm{B}} )$ and $Y=( \overrightarrow{G}_{\setminus \mathrm{B}}, \overrightarrow{\mathsf{G}}_{\setminus \mathrm{B}}, W,\mathsf{W} )$, we get that
\begin{align*}
    \mathbb{P}(\Bar{\mathcal{A}}) = \mathbb{P}( \mathrm{BAD}_t = \mathrm{B}_t, \mathrm{BAD}_{t-1} = \mathrm{B}_{t-1} \mid ( \overrightarrow{G}_{\mathrm{B}}, \overrightarrow{\mathsf{G}}_{\mathrm{B}} ) = ( \overrightarrow{g}_{\mathrm{B}}, \overrightarrow{\mathsf{g}}_{\mathrm{B}} ) ) \geq \exp \{ - n \Delta_t^{8} \}\,.
\end{align*}
Thus, for an amenable set-realization $(\mathrm{B}_t, \mathrm{B}_{t-1})$ and an amenable bias-realization $( \overrightarrow{g}_{\mathrm{B}}, \overrightarrow{\mathsf{g}}_{\mathrm{B}} )$, 
\begin{align*}
    \frac{ {p}_{ \{ \mathbf{X}^{\leq t} | \Bar{\mathcal{A}} \} } (x^{\leq t}) }{ \mathbb{P}(\mathcal{M}_1) \cdot {p}_{ \{ \mathbf{X}^{\leq t} | \mathcal{M}_1 \} } (x^{\leq t}) } \leq \frac{ 1 }{ \mathbb{P}( \Bar{\mathcal{A}} ) }
    \leq \exp \{ n \Delta_t^8 \} \,.
\end{align*}
Therefore, it remains to show that for an amenable variable-realization $x^{\leq t}$
\begin{align}
    \frac{ \mathbb{P}( \mathcal{M}_1 ) \cdot {p}_{ \{ \mathbf{X}^{\leq t} | \mathcal{M}_1 \} } (x^{\leq t}) }{  {p}_{ \{ \mathbf{Y}^{\leq t} \} } (x^{\leq t})  } \leq \exp \big\{ O \big( n \Delta_t^8 \big) \big\}\,.
    \label{equ-unconditioned-density-ratio}
\end{align}
For a random variable $X$, define $p_{X;\mathcal{M}_j} (x)$ to be the density of $X$ on the event $\mathcal{M}_j$, i.e., $\int_{ x \in A } p_{X;\mathcal{M}_j} (x) dx = \mathbb{P}( X \in A; \mathcal{M}_j )$ for any $A$. From the definition we see that $\mathcal M_j$ is increasing in $j$ (and thus $p_{X;\mathcal{M}_j} (x)$ is increasing in $j$). Combined with the facts that $p_{X;\mathcal{M}_j} (x) \leq p_{X}(x) $ and $p_{X;\mathcal{M}_j} (x) =\mathbb{P}(\mathcal{M}_j) p_{X|\mathcal{M}_j} (x)$, it yields that the left hand side of \eqref{equ-unconditioned-density-ratio} is equal to (also note that $\mathbf{X}^{\leq t}_{(N)} \overset{d}{=} \mathbf{Y}^{\leq t}$)
\begin{align}
     \frac{ p_{ \{ \mathbf{X}^{\leq t}_{(0)} ; \mathcal{M}_1 \} } (x^{\leq t}) }{ {p}_{ \{ \mathbf{X}^{\leq t}_{(N)} \} } (x^{\leq t}) } \leq \frac{ p_{ \{ \mathbf{X}^{\leq t}_{(0)} ; \mathcal{M}_1 \} } (x^{\leq t}) }{  {p}_{ \{ \mathbf{X}^{\leq t}_{(N)} ; \mathcal{M}_N \} } (x^{\leq t})  } \leq \prod_{j=1}^{N} \frac{ {p}_{ \{ \mathbf{X}^{\leq t}_{(j-1)} ; \mathcal{M}_j \} } (x^{\leq t}) }{  {p}_{ \{ \mathbf{X}^{\leq t}_{(j)} ; \mathcal{M}_j \} } (x^{\leq t}) } \,.
\end{align}
Since $N\leq n^2$, we can conclude the proof of Lemma~\ref{lemma-joint-density-ratio-I} by combining Lemma~\ref{lemma-bound-truncated-unconditional-density-ratio} below.
\begin{Lemma} {\label{lemma-bound-truncated-unconditional-density-ratio}}
For an amenable set-realization $(\mathrm{B}_t, \mathrm{B}_{t-1})$ and an amenable variable-realization $x^{\leq t}$, we have for all $1\leq j\leq N$
\begin{align*}
    \frac{ {p}_{ \{ \mathbf{X}^{\leq t}_{(j-1)} ; \mathcal{M}_j \} } (x^{\leq t}) }{  {p}_{ \{ \mathbf{X}^{\leq t}_{(j)} ; \mathcal{M}_j \} } (x^{\leq t}) } = 1 + O \big( K_t^{20} \vartheta^{-3} (\log n)^3 n^{ \frac{3}{ \log \log \log n } } / n \sqrt{n\Hat{q}} \big) \,.
\end{align*}
\end{Lemma}
The proof of Lemma~\ref{lemma-bound-truncated-unconditional-density-ratio} requires a couple of results on the Gaussian-smoothed density.
\begin{Lemma}
\label{lemma-bound-difference-exp-moments}
    For $1 \leq d \leq m$ and $C > 0$, let $U=(U_1,\ldots,U_m)$ be a random vector such that $|U_{k}| \leq C$ for $1 \leq k \leq d$, and let $X_1,\ldots,X_d,Y_1,\ldots,Y_d$ be sub-Gaussian random variables independent with $\{ U_1,\ldots,U_d \}$ such that $\mathbb{E}[X_k] = \mathbb{E}[Y_k]$ and $\mathbb{E}[X_k X_l] = \mathbb{E}[Y_k Y_l]$ for any $1 \leq k,l \leq d$. Define $\Tilde{X}$ such that $\Tilde{X}_k = X_k$ for $1 \leq k \leq d$ and $\Tilde{X}_k = 0$ for $d+1 \leq k \leq m$ (and define $\Tilde{Y}$ similarly). Then, for any positive definite $m\!*\!m$ matrix $\mathrm{A}$,
    \begin{align}
        & \Big| \mathbb{E} \Big[ e^{ \frac{1}{2} ( - \| U \|^2_{\mathrm{A}} + 2 \langle U , \Tilde{X} \rangle_{\mathrm{A}} - \| \Tilde{X} \|^2_{\mathrm{A}} ) } - e^{\frac{1}{2} ( - \| U \|^2_{\mathrm{A}} + 2 \langle U , \Tilde{Y} \rangle_{\mathrm{A}} - \| \Tilde{Y} \|^2_{\mathrm{A}} ) }  \Big] \Big| \label{eq-Lemma-3.11-LHS}  \\
        \leq & 100 d^3 \big( C \| \mathrm{A} \|_{\infty} + \| \mathrm{A} \|_{\mathrm{op}} \| \mathrm{A}^{-1} \|_{\mathrm{op}}^{1/2} \big)^3 \mathbb{E} \Big[ e^{ - \frac{1}{2} \| U \|^2_{ \mathrm{A}} } \Big] \mathbb{E} \Big[e^{ C \| \mathrm{A} \|_{\infty} \|{X}\|_{1} } \| X \|^3 + e^{  C \| \mathrm{A} \|_{\infty} \|{Y}\|_{1} } \| Y \|^3 \Big].\nonumber
    \end{align}
\end{Lemma}
\begin{proof}
For $x\in \mathbb R^d$ write $\psi(x) =  e^{-\frac{1}{2} \| \Tilde{x} \|^2_{\mathrm{A}} + \langle U, \Tilde{x} \rangle_{\mathrm{A}} }$, where $\Tilde{x}_i = x_i$ for $1 \leq i \leq d$ and $\Tilde{x}_i=0$ for $d+1 \leq i \leq m$. Then \eqref{eq-Lemma-3.11-LHS} is equal to $\big| \mathbb{E}  \big[ e^{- \frac{1}{2} \| U \|^2_{\mathrm{A}} } \big( \psi( X_1,\ldots,X_d) - \psi(Y_1,\ldots,Y_d) \big) \big] \big|$.
This motivates the consideration of high-dimensional Taylor's expansion. Define $\psi^{\prime}_a = \frac{\partial}{\partial x_a}\psi(0,\ldots,0)$, $\psi^{\prime \prime}_{ab} = \frac{\partial^2}{\partial x_a \partial x_b} \psi(0,\ldots,0)$ and $\psi^{\prime \prime \prime}_{abc}(t_1,\ldots,t_d) = \frac{\partial^3}{\partial x_a \partial x_b \partial x_c} \psi(t_1,\ldots,t_d)$. Then,
\begin{align*}
    \Big| \psi(x_1,\ldots,x_d) - \psi(0,\ldots,0) - \sum_{a=1}^{d} \psi^{\prime}_{a} x_a - \frac{1}{2} \sum_{a,b=1}^{d}  \psi^{\prime \prime}_{ab} x_a x_b \Big| \leq  \mathbf{R}_{\psi} d^3 \| x \|^3 \,,
\end{align*}
where the remainder $\mathbf{R}_{\psi}$ is bounded by
\begin{align}\label{eq-high-dimensional-taylor-expansion}
    |\mathbf{R}_{\psi}(x_1,\ldots,x_d)| \leq \sup_{ |t_j| \leq |x_j|, 1 \leq j \leq d } | \psi^{\prime \prime \prime}_{abc}(t_1,\ldots,t_d) | \,.
\end{align}
Since $\psi^{\prime}, \psi^{\prime \prime}$ are random variables measurable with respect to $\{ U_k: 1 \leq k \leq m \}$, we get
\begin{equation}\label{eq-first-two-derivaties-0}
\begin{aligned}
    & \mathbb{E} \Big[ e^{-\frac{1}{2} \| U \|^2_{ \mathrm{A} } } \psi^{\prime}_a (X_a-Y_a) \Big] = \mathbb{E} \Big[ e^{-\frac{1}{2} \| U \|^2_{ \mathrm{A} } } \psi^{\prime}_a \Big] \mathbb{E} \Big[ (X_a-Y_a) \Big] = 0 \,, \\
    & \mathbb{E} \Big[ e^{-\frac{1}{2} \| U \|^2_{ \mathrm{A} } } \psi^{\prime \prime}_{ab} (X_a X_b-Y_a Y_b) \Big] = \mathbb{E} \Big[ e^{-\frac{1}{2} \| U \|^2_{ \mathrm{A} } } \psi^{\prime \prime}_{ab} \Big] \mathbb{E} \Big[ (X_a X_b-Y_a Y_b) \Big] = 0 \,.
\end{aligned}
\end{equation}
Define $f_j$ to  be the $j$'th standard basis in $\mathbb{R}^{m}$, and define $\tilde t$ from $t$ similar to defining $\tilde x$ from $x$. Since 
\[
\langle \Tilde{t}, U \rangle_{\mathrm{A}} \leq \| \mathrm{A} U^{*} \|_{\infty} \| \Tilde{t} \|_{1} \leq \| U \|_{\infty} \| \mathrm{A} \|_{\infty} \| \Tilde{t} \|_{1} 
\]
which is bounded by $C \| \mathrm{A} \|_{\infty} \| x \|_1$ if $|t_j| \leq |x_j|$ for all $j$,  we have that for distinct $a, b, c$
\begin{align*}
    & \sup_{|t_j| \leq |x_j|} |\psi^{\prime \prime \prime}_{abc}(t_1,\ldots,t_d)| = \sup_{|t_j| \leq |x_j|} \Big| e^{-\frac{1}{2} \| \Tilde{t} \|_{\mathrm{A}}^2 + \langle \Tilde{t}, U \rangle_{\mathrm{A}} } \prod_{ \tau \in \{a,b,c \} } ( \langle f_{\tau}, U \rangle_{\mathrm{A}} - \langle f_{\tau}, \Tilde{t} \rangle_{\mathrm{A}} )  \Big| \\
    & \leq e^{ C \| \mathrm{A} \|_{\infty} \| x \|_1 }  \sup_{|t_j| \leq |x_j|} \Big| \prod_{ \tau \in \{a,b,c \} } \Big( ( \langle f_{\tau}, U \rangle_{\mathrm{A}} - \langle f_{\tau}, \Tilde{t} \rangle_{\mathrm{A}} ) e^{-\frac{1}{6} \| \Tilde{t} \|_{\mathrm{A}}^2 } \Big) \Big| \,,
\end{align*}
where $e^{-\frac{1}{6} \| \Tilde{t} \|_{\mathrm{A}}^2}$ arises since we split $e^{-\frac{1}{2} \| \Tilde{t} \|_{\mathrm{A}}^2}$ into the product of three copies of $e^{-\frac{1}{6} \| \Tilde{t} \|_{\mathrm{A}}^2}$.
Since $e^{- \frac{1}{6} x^2} x\leq 1$ for all $x\in \mathbb R$, we have that 
\begin{align*}
    e^{-\frac{1}{6} \| \Tilde{t} \|_{\mathrm{A}}^2  } | \langle f_{\tau}, \Tilde{t} \rangle_{\mathrm{A}} | \leq e^{-\frac{1}{6} \| \Tilde{t} \|_{\mathrm{A}}^2 } \| \Tilde{t} \mathrm{A} \| \leq e^{ -\frac{1}{6} \| \mathrm{A}^{-1} \|_{\mathrm{op}}^{-1} \| \Tilde{t} \|^2 } \| \mathrm{A} \|_{\mathrm{op}} \| \Tilde{t} \| \leq \| \mathrm{A} \|_{\mathrm{op}} \| \mathrm{A}^{-1} \|_{\mathrm{op}}^{\frac{1}{2}} \,.
\end{align*}
Combining the preceding two inequalities with 
\[
| \langle f_{\tau},U \rangle_{\mathrm{A}} | \leq \| \mathrm{A} U^{*} \|_{\infty} \| f_{\tau} \|_1 \leq \| U \|_{\infty} \| \mathrm{A} \|_{\infty} \| f_{\tau} \|_1 \leq C\| \mathrm{A} \|_{\infty} \,, 
\]
we get that $|\psi^{\prime \prime \prime}_{abc} (t_1,\ldots,t_d)| \leq e^{C \| \mathrm{A} \|_{\infty} \| x \|_1} ( C\| \mathrm{A} \|_{\infty} + \| \mathrm{A} \|_{\mathrm{op}} \| \mathrm{A}^{-1} \|_{\mathrm{op}}^{1/2} )^3$. Similarly, we have
\begin{align*}
    & \sup_{|t_j| \leq |x_j|} |\psi^{\prime \prime \prime}_{aab}(t_1,\ldots,t_d)|  \leq e^{ C \| \mathrm{A} \|_{\infty} \| x \|_{1} } ( 2 C \| \mathrm{A} \|_{\infty} + 4 \| \mathrm{A} \|_{\mathrm{op}} \| \mathrm{A}^{-1} \|_{\mathrm{op}}^{1/2} )^3 \,, \\
    & \sup_{|t_j| \leq |x_j|} |\psi^{\prime \prime \prime}_{aaa}(t_1,\ldots,t_d)| 
    \leq e^{ C \| \mathrm{A} \|_{\infty} \| x \|_{1} } \big( 4(C \| \mathrm{A} \|_{\infty})^3 + 100 (\| \mathrm{A} \|_{\mathrm{op}} \| \mathrm{A}^{-1} \|_{\mathrm{op}}^{1/2} )^3 \big) \,.
\end{align*}
Therefore, $\mathbf{R}_{\psi} (x_1,\ldots,x_d) \leq 100 \big( C \| \mathrm{A} \|_{\infty} + \| \mathrm{A} \|_{\mathrm{op}} \| \mathrm{A}^{-1} \|_{\mathrm{op}}^{1/2} \big)^3 \exp \{ C \| \mathrm{A} \|_{\infty} \| x \|_1 \}$. Combined with \eqref{eq-high-dimensional-taylor-expansion} and \eqref{eq-first-two-derivaties-0}, it yields that \eqref{eq-Lemma-3.11-LHS} is bounded by
\begin{align*}
    100 d^3 \big( C \| \mathrm{A} \|_{\infty} + \| \mathrm{A} \|_{\mathrm{op}} \| \mathrm{A}^{-1} \|_{\mathrm{op}}^{1/2} \big)^3 \mathbb{E} \Big[ e^{ - \frac{1}{2} \| U \|^2_{ \mathrm{A}} }  \Big(e^{ C \| \mathrm{A} \|_{\infty} \|{X}\|_{1} } \| X \|^3 + e^{  C \| \mathrm{A} \|_{\infty} \|{Y}\|_{1} } \| Y \|^3 \Big) \Big]\,,
\end{align*}
completing the proof since $\{X_1, \ldots, X_d, Y_1, \ldots, Y_d\}$ is independent of $\{U_1, \ldots, U_d\}$.
\end{proof}

The next lemma will be useful in bounding the density change when locally replacing Bernoulli variables by Gaussian variables.

\begin{Lemma} \label{lemma-substitute-one-group}
For $1 \leq d \leq m$, let $Z=(Z_1,\ldots,Z_m) \sim \mathcal{N}(0,\Sigma)$ be a normal vector, and let $U=(U_1,\ldots,U_m)$ be a sub-Gaussian vector independent with $Z$ such that $|U_{k}| \leq C$. Let $B,B^{\prime},\mathsf{B},\mathsf{B}^{\prime}$, $G,G^{\prime},\mathsf{G},\mathsf{G}^{\prime}$ be random variables independent with $Z,U$ such that $B,B^{\prime},\mathsf{B},\mathsf{B}^{\prime}$ all have the same law as $\mathrm{Ber}(q)-q$, and $G,G^{\prime},\mathsf{G},\mathsf{G}^{\prime} \sim \mathcal{N}(0,q(1-q))$. Also, suppose that $( B,B^{\prime},\mathsf{B},\mathsf{B}^{\prime})$ and $(G,G^{\prime},\mathsf{G},\mathsf{G}^{\prime})$ have the same covariance matrix. For any $\alpha,\beta,\theta,\gamma \in \mathbb{R}^d$ with $\ell_\infty$-norms at most $\epsilon$, define $\Tilde \alpha \in \mathbb R^m$ such that $\Tilde{\alpha}(i) = \alpha(i)$ for $1 \leq i \leq d$ and $\Tilde{\alpha} (i)=0$ for $d+1\leq i \leq m$; we similarly define $\Tilde{\beta},\Tilde{\theta},\Tilde{\gamma} \in \mathbb R^m$. Then for all $\lambda\in \mathbb R^m$, the joint densities of  $( Z+U+B \Tilde{\alpha} + B^{\prime} \Tilde{\beta} + \mathsf{B} \Tilde{\theta} + \mathsf{B}^{\prime} \Tilde{\gamma}  )$ and $( Z+U+G \Tilde{\alpha} + G^{\prime} \Tilde{\beta} + \mathsf{G} \Tilde{\theta} + \mathsf{G}^{\prime} \Tilde{\gamma} )$ satisfy that
\begin{align*}
    \Big| \frac{ p_{\{ Z+U+B \Tilde{\alpha} + B^{\prime} \Tilde{\beta} + \mathsf{B} \Tilde{\theta} + \mathsf{B}^{\prime} \Tilde{\gamma} \}} (\lambda) }{ p_{\{ Z+U+G \Tilde{\alpha} + G^{\prime} \Tilde{\beta} + \mathsf{G} \Tilde{\theta} + \mathsf{G}^{\prime} \Tilde{\gamma} \}} (\lambda) } - 1 \Big| \leq \ & 10^4 d^{5} \epsilon^3 q e^{8d \epsilon^2 q} \big( (\| \lambda \|_{\infty}+C) \| \Sigma^{-1} \|_{\infty} + \| \Sigma^{-1}\|_{\mathrm{op}} \| \Sigma \|_{ \mathrm{op}}^{1/2} \big)^3  \\
    & * \Big( e^{ 16d^2 \epsilon \| \Sigma^{-1} \|_{\infty} (C+\| \lambda \|_{\infty} ) } + e^{  4d^3 q \epsilon^2 \| \Sigma^{-1} \|^2_{\infty} (C+\| \lambda \|_{\infty})^2 } \Big)   \,.
\end{align*}
\end{Lemma}
\begin{proof}
We have that $p_{Z+U}(\lambda) = (\mathrm{det}(\Sigma))^{-\frac{1}{2}} (2\pi)^{-\frac{d}{2}}  \mathbb{E}_{U} \big[  \exp \{  -\frac{1}{2}\| \lambda - U \|_{\Sigma^{-1}}^2 \} \big]$, where $\mathbb{E}_{U}$ is the expectation by averaging over $U$.
Writing $\Tilde{X}=B \Tilde{\alpha} + B^{\prime} \Tilde{\beta} + \mathsf{B} \Tilde{\theta} + \mathsf{B}^{\prime} \Tilde{\gamma}$ and  writing $\Tilde{Y}= G \Tilde{\alpha} + G^{\prime} \Tilde{\beta} + \mathsf{G} \Tilde{\theta} + \mathsf{G}^{\prime} \Tilde{\gamma}$, we then have
\begin{align}\label{eq-mathfrak-I-1-2}
    \frac{ p_{Z+U+ \Tilde{X}}(\lambda) }{ p_{Z+U+ \Tilde{Y}}(\lambda) } -1 = \frac{ \mathbb{E}\big[ e^{ -\frac{1}{2} \| \lambda - U - \Tilde{X} \|^2_{\Sigma^{-1}}  }  \big]  }{ \mathbb{E}\big[ e^{ -\frac{1}{2} \| \lambda - U - \Tilde{Y} \|^2_{\Sigma^{-1}} } \big] } -1  = \mathfrak I_1 \times \mathfrak I_2
    \end{align}
where $\mathfrak I_1  = \frac{ \mathbb{E}\big[ e^{-\frac{1}{2} \| \lambda - U \|^2_{\Sigma^{-1}} } \big] }{ \mathbb{E}\big[ e^{-\frac{1}{2} \| \lambda - U - \Tilde{Y} \|^2_{\Sigma^{-1}} } \big] }$ and  \label{equ-density-ratio-part-II} 
    \begin{align}
    \mathfrak I_2 =   \frac{ \mathbb{E} \big[  e^{-\frac{1}{2} \| \lambda - U \|^2_{\Sigma^{-1}} } \big( e^{ \langle \lambda-U, \Tilde{X} \rangle_{\Sigma^{-1}} - \frac{1}{2} \| \Tilde{X} \|^2_{\Sigma^{-1}} } - e^{ \langle \lambda-U, \Tilde{Y} \rangle_{\Sigma^{-1}} - \frac{1}{2} \| \Tilde{Y} \|^2_{\Sigma^{-1}} } \big) \big] }{ \mathbb{E} \big[ e^{-\frac{1}{2} \| \lambda - U \|^2_{\Sigma^{-1}} } \big] } \label{equ-density-ratio-part-I}\,.
\end{align}
Here the equality in \eqref{eq-mathfrak-I-1-2} holds since one may simply cancel out the denominator in $\mathfrak I_2$ with the factor in $\mathfrak I_1$.
Thus, it suffices to bound $\mathfrak I_1$ and $\mathfrak I_2$ separately. 

We first bound $\mathfrak I_2$. Applying Lemma~\ref{lemma-bound-difference-exp-moments} with $\mathrm{A}=\Sigma^{-1}$ and using $|\lambda_{k}-U_{k}| \leq |\lambda_k|+C_k \leq \| \lambda \|_{\infty} + C$ and $\| X \|_1 \leq d \| X \|$ we get that   
\begin{align*}
    \mathfrak I_2 \leq \ & 100 d^3 ( (\| \lambda \|_{\infty}+C) \| \Sigma^{-1} \|_{\infty} + \| \Sigma^{-1}\|_{\mathrm{op}} \| \Sigma \|_{ \mathrm{op}}^{1/2} )^3  \\
    * & \Big( \mathbb{E} \Big[ e^{ d (C+\|\lambda \|_{\infty}) \| \Sigma^{-1} \|_{\infty} \| X \| } \|X\|^3 \Big]+ \mathbb{E} \Big[ e^{ d (C+\|\lambda \|_{\infty}) \| \Sigma^{-1} \|_{\infty} \| Y \| } \|Y\|^3 \Big] \Big)\,.
\end{align*}
For notational convenience, we denote $\mathtt E_X$ and $\mathtt E_Y$ the two expectations in the preceding inequality. Since $|X_k|=| B \alpha_k + B^{\prime} \beta_k + \mathsf{B} \theta_k + \mathsf{B}^{\prime} \gamma_k | \leq 4 \epsilon$, we have $\|X\| \leq 2 \sqrt{d} \epsilon$ and $\mathbb{P}[ \| X \| = 0 ] \geq \mathbb{P}[ B,B^{\prime},\mathsf{B},\mathsf{B}^{\prime} =0] \geq 1-4q$. Thus,
\begin{align}
    \mathtt E_X &\leq  \mathbb{E}_{X} \Big[ e^{ 4 d^2 \epsilon (C+\| \lambda \|_{\infty}) \| \Sigma^{-1} \|_{\infty} } (2\sqrt{d} \epsilon)^3 \mathbf{1}_{\| X \| \neq 0} \Big] \leq 100 d^2 \epsilon^3 q e^{ 4d^2 \epsilon (C+\| \lambda \|_{\infty}) \| \Sigma^{-1} \|_{\infty} }   \,.
    \label{equ-bound-density-ratio-part-I-1}
\end{align}
Since $Y_k= G \alpha_k + G^{\prime} \beta_k + \mathsf{G} \theta_k + \mathsf{G}^{\prime} \gamma_k$ is a sub-Gaussian variable with sub-Gaussian norm at most  $16 \epsilon^2 q$, we see that $\mathbb{P} \big[ \| Y \| \geq t \Big] \leq 2d e^{-\frac{t^2}{16 \epsilon^2 d q}}$. Consequently,
\begin{align*}
    \mathtt E_Y \leq 100 d^2 \epsilon^3 q  e^{4 d^3 q \epsilon^2 \| \Sigma^{-1} \|_{\infty}^2 (C + \| \lambda \|_{\infty})^2   }  \,.
\end{align*}
Combined with \eqref{equ-bound-density-ratio-part-I-1}, it yields that
\begin{align}
    \mathfrak I_2 \leq 10^4 d^{5} \epsilon^3 q \big( (\| \lambda \|_{\infty}+C) \| \Sigma^{-1} \|_{\infty} + \| \Sigma^{-1}\|_{\mathrm{op}} \| \Sigma \|_{ \mathrm{op}}^{1/2} \big)^3 \nonumber \\
    * \Big( e^{ 4d^2 \epsilon \| \Sigma^{-1} \|_{\infty} (C+\| \lambda \|_{\infty} ) } + e^{  4d^3 q \epsilon^2 \| \Sigma^{-1} \|^2_{\infty} (C+\| \lambda \|_{\infty})^2 } \Big) \,.
    \label{equ-bound-density-ratio-I}
\end{align}
We next bound $\mathfrak I_1$. Applying Jensen's inequality, we get that 
\begin{align*}
    & \mathbb{E}_{U,Y} \Big[ e^{  -\frac{1}{2} \| \lambda - U - \Tilde{Y} \|^2_{\Sigma^{-1}}  } \Big] \geq \mathbb{E}_{U} \Big[ e^{ \mathbb{E}_{Y} [ -\frac{1}{2} \| \lambda - U - \Tilde{Y} \|^2_{\Sigma^{-1}} ] }  \Big] \\
    \geq \ & e^{ - \| \Sigma^{-1} \|_{\mathrm{op}} \mathbb{E}[ \| \Tilde{Y} \|^2 ] } \mathbb{E}_{U} \Big[ e^{ -\frac{1}{2} \| \lambda - U \|^2_{\Sigma^{-1}} }  \Big] \geq e^{- 16 d q \epsilon^2 \| \Sigma^{-1} \|_{\mathrm{op}} }  \mathbb{E}_{U} \Big[ e^{ -\frac{1}{2} \| \lambda - U \|^2_{\Sigma^{-1}} }  \Big] \,,
\end{align*}
where the second inequality uses independence and the last inequality uses $\E Y_k^2 \leq 16 q \epsilon^2$. Thus,
$\mathfrak I_1 \leq e^{16d \epsilon^2 q \| \Sigma^{-1} \|_{\mathrm{op}}} $. 
Combined with \eqref{equ-bound-density-ratio-I}, this yields the desired bound.
\end{proof}

By Lemma~\ref{lemma-substitute-one-group}, we need to employ suitable truncations in order to control the density ratio; this is why we defined $\mathrm{LARGE}_t \subset \mathrm{BAD}_t$ as in \eqref{equ-def-set-LARGE}. We now prove Lemma~\ref{lemma-bound-truncated-unconditional-density-ratio}.

\begin{proof}[Proof of Lemma~\ref{lemma-bound-truncated-unconditional-density-ratio}.]
Fix $1\leq j\leq N$. We now set the framework for applying Lemma~\ref{lemma-substitute-one-group}. Recall (the third equality of) \eqref{equ-def-degree-after-j-substitution}.  Define 
\begin{align*}
    & Z(s,k,v) = W^{(s)}_v (k) + \mathbf{X}^{\leq t}_{[j-1]} (s,k,v), \quad \lambda(s,k,v) = x^{\leq t}(s,k,v) \,, \\
    & Z(s,k,\mathsf{v}) = \mathsf{W}^{(s)}_{\mathsf{v}} (k) + \mathbf{X}^{\leq t}_{[j-1]} (s,k,\mathsf{v}), \quad \lambda(s,k,\mathsf{v}) = {x}^{\leq t} (s,k,\mathsf{v}) \,.
\end{align*}
Let $( B, B^{\prime}, \mathsf{B}, \mathsf{B}^{\prime} ) = ( \overrightarrow{G}_{u_j,w_j}-\Hat{q}, \overrightarrow{G}_{w_j,u_j}-\Hat{q}, \overrightarrow{\mathsf{G}}_{\pi(u_j),\pi(w_j)}-\Hat{q}, \overrightarrow{\mathsf{G}}_{\pi(w_j),\pi(u_j)}-\Hat{q} )$, and let 
$$
\alpha(s,k,v) = \sum_{i=1}^{K_s} \eta^{(s)}_k (i) \frac{ \mathbf{1}_{ \{ w_j \in \xi^{(s)}_i \} } - \mathfrak{a}_s }{ \sqrt{ (\mathfrak{a}_s - \mathfrak{a}_s^2) n \Hat{q} (1-\Hat{q}) } } \mbox{ for } v = u_j \mbox{ and } \alpha(s,k,v) = 0 \mbox{ for } v \neq u_j \,.
$$ 
Thus we see that $\alpha(s,k,v)$ is the coefficient of $\overrightarrow{G}_{u_j,w_j}-\Hat{q}$ in $\langle \eta^{(s)}_k, \varphi^{(s)}_v \rangle$. Similarly denote by $\beta(s,k,v)$ the coefficient of $(\overrightarrow{G}_{w_j,u_j}-\Hat{q})$ in $\langle \eta^{(s)}_k, \varphi^{(s)}_v \rangle$, denote by $\theta(s,k,\mathsf{v})$ the coefficient of $(\overrightarrow{\mathsf{G}}_{\pi(u_j),\pi(w_j)}-\Hat{q})$ in $\langle \eta^{(s)}_k, \psi^{(s)}_{\mathsf{v}} \rangle$, and denote by $\gamma(s,k,\mathsf{v})$ the coefficient of $(\overrightarrow{\mathsf{G}}_{\pi(w_j),\pi(u_j)}-\Hat{q})$ in $\langle \eta^{(s)}_k, \psi^{(s)}_{\mathsf{v}} \rangle$. Further, define $U(s,k,v) = \langle \eta^{(s)}_k, \varphi^{(s)}_v \rangle_{\langle j \rangle}$ and $U(s,k,\pi(v)) = \langle \eta^{(s)}_k, \varphi^{(s)}_{\pi(v)} \rangle_{\langle j \rangle}$. Then we have $\mathbf{X}^{\leq t}_{(j-1)} = U + Z + (B \alpha + B^{\prime} \beta + \mathsf{B} \theta + \mathsf{B}^{\prime} \gamma  )$, and on $\mathcal{M}_j$ we have that $\| U \|_{\infty} \leq n^{\frac{1}{\log \log \log n}}$. Also, we have $| x^{\leq t} (s,k,v) | \leq 2 (\log n) n^{ \frac{1}{ \log \log \log n} }$ since $x^{\leq t}$ is an amenable variable-realization. 
Thus, by $\| \eta_k \|^2 \leq 2$ and Cauchy-Schwartz inequality,
\begin{align*}
    |\alpha(s,k,u_j)| , |\beta(s,k,w_j)|, |\theta(s,k,\pi(u_j))|, |\gamma(s,k,\pi(w_j))| \leq (K_s/ \mathfrak{a}_s n \Hat{q})^{1/2}\,.
\end{align*}
Finally, let $\Sigma_{[j-1]}$  be the covariance matrix of $Z$. Since $\Sigma_{[j-1]} - \mathrm{I}$ is the covariance matrix of $\{ \mathbf{X}^{\leq t}_{[j-1]} \}$, we have $\| \Sigma^{-1}_{[j-1]} \|_{\mathrm{op}} \leq 1$. Also, $\|\Sigma_{[j-1]}\|_{\infty} \leq 2$ and 
$$
\Sigma_{[j-1]} ((s,k,u);(r,l,v)) = O \big( (\mathbf{1}_{v \in \cup_{i} \Gamma^{(s)}_i} - \mathfrak{a}_s) (\mathbf{1}_{u \in \cup_{i} \Gamma^{(r)}_i} - \mathfrak{a}_r) / n \sqrt{\mathfrak{a}_s \mathfrak{a}_r}  \big) \mbox{ for } u \neq v\,.
$$ 
Thus, we may apply Lemma~\ref{lemma-bound-on-op-norm} with 
\begin{align*}
    \mathcal{I}_v = \mathcal{J}_v = \big\{ (s,k,v), (s,k,\pi(v)): 0\leq s\leq t, 1\leq k\leq K_s/12 \big\}
\end{align*}
and derive that $\| \Sigma_{[j-1]} \|_{\mathrm{op}} \leq 2 K_t^2$. Furthermore, we can bound $\| \Sigma^{-1}_{[j-1]} \|_\infty$ by Lemma~\ref{lemma-bound-on-infty-norm} as follows. Set $\mathrm{A}$ (in Lemma~\ref{lemma-bound-on-infty-norm}) to be a matrix with $\mathrm{A}((s,k,v),(r,l,v)) = \Sigma_{[j-1]} ((s,k,v),(r,l,v)) $ and all other entries being 0, and set $\mathrm{B} = \Sigma_{[j-1]} - \mathrm{A}$. Then $\| (\mathrm{A+B})^{-1} \|_{\mathrm{op}} = \| \Sigma_{[j-1]}^{-1} \|_{\mathrm{op}} \leq 1$, the entries of $\mathrm{B}$ are bounded by $\vartheta^{-1} K_t^2 / n$, and $\mathrm{A}= \mathrm{diag}( \mathrm{I} + \mathrm{A}_v )$ is a block-diagonal matrix where each block $\mathrm{A}_v$ is the covariance matrix of 
$$ \big\{ \mathbf{X}^{\leq t}_{[j-1]} (s,k,v), \mathbf{X}^{\leq t}_{[j-1]} (s,k,\pi(v)): 0\leq s\leq t, 1\leq k\leq K_s/12 \big\}\,.$$ 
So $\mathrm{A}_v$ is semi-positive definite and thus $\| (\mathrm{I} + \mathrm{A}_{v})^{-1} \|_{\mathrm{op}} \leq 1$. Since also the dimension of $\mathrm{A}_v$ is bounded by $K_t$,  we have $\| ( \mathrm{I} + \mathrm{A}_v )^{-1} \|_{\infty} \leq K_t^2 \| (\mathrm{I} + \mathrm{A}_{v})^{-1} \|_{\mathrm{op}} \leq K_t^2$. Thus, we have that $\| \mathrm{A}^{-1} \|_{\infty} = \max_v \{ \| (\mathrm{I} + \mathrm{A}_v)^{-1} \|_{\infty} \} \leq K_t^2$. Therefore, we can apply Lemma~\ref{lemma-bound-on-infty-norm} with $C=1,m=K_t n, K = \vartheta^{-1} K_t^3$ and $L=K_t^2$ and obtain that $\| \Sigma_{[j-1]}^{-1} \|_{\infty} \leq 2 K_t^5 \vartheta^{-1}$. Applying Lemma~\ref{lemma-substitute-one-group} and using independence between $\{ \mathcal{M}_j, U \}$ and $Z$, we get that
\begin{align*}
    & \Big| \frac{ {p}_{ \{ \mathbf{X}^{\leq t}_{(j-1)} ; \mathcal{M}_j \} } (x^{\leq t}) }{  {p}_{ \{ \mathbf{X}^{\leq t}_{(j)} ; \mathcal{M}_j \} } (x^{\leq t}) } - 1 \Big|  \\
    \lesssim \ & K_t^5 (\sqrt{n \Hat{q}})^{-3} \Hat{q} ( (\log n) n^{\frac{1}{\log \log \log n}} K_t^5 \vartheta^{-1} )^3 
    \exp \big\{ 4K_t \sqrt{ K_t / \vartheta n \Hat{q} } \cdot 2 (\log n) n^{ \frac{1}{\log \log \log n} } \big\}  \\
    \leq \ & K_t^{20} \vartheta^{-3} (\log n)^3 n^{ \frac{3}{ \log \log \log n } } / n\sqrt{n\Hat{q}}  \,. \qedhere
\end{align*}
\end{proof}

\subsection{Gaussian analysis} \label{sec:Gaussian-analysis}

Since 
\begin{equation}\label{eq-independence-observation}
    \eqref{eq-to-be-conditioned-on} \mbox{ is independent of Gaussian variables } \big\{W^{(s)}_v, \mathsf W^{(s)}_{\mathsf v},  \overrightarrow{Z}_{v,u}, \overrightarrow{\mathsf Z}_{\mathsf v,\mathsf u} \big\}\,,
\end{equation} 
when analyzing the process defined by \eqref{equ-def-g-Tilde-D} it would be convenient to and thus we will 
\begin{equation}\label{eq-convention-announcement}
\mbox{ condition on the realization of } \eqref{eq-to-be-conditioned-on}\,. 
\end{equation}
As such,  the following process defined by \eqref{equ-def-g-Tilde-D} can be viewed as a Gaussian process:
\begin{align*}
    \Bigg\{ \begin{split} & \Tilde{W}^{(s)}_v(k) + \langle \eta^{(s)}_k, g_{t}\Tilde{D}^{(s)}_v \rangle \\ & \Tilde{\mathsf{W}}^{(s)}_{\pi(v)} (k) + \langle \eta^{(s)}_k, g_{t}\Tilde{\mathsf{D}}^{(s)}_{\pi(v)} \rangle  \end{split} : 0 \leq s \leq t+1, 1 \leq k \leq \frac{K_s}{12}, v \in V \setminus \mathrm{BAD}_{t}  \Bigg\}.
\end{align*}
Note that our convention here is consistent with \eqref{eq-Lemma-3.9} (see also Remark~\ref{rem-Y-conditioning-B-t}).
Recall that $\mathcal{F}_t$ is the $\sigma$-field generated by $\mathfrak{F}_t$ (see \eqref{eq-def-mathcal-F-t}), which is slightly different from the above process since in $\mathfrak F_t$ we have $s\leq t$. We will study the conditional law of $\Tilde{W}^{(t+1)}_v(k) + \langle \eta^{(t+1)}_k, g_{t}\Tilde{D}^{(t+1)}_v \rangle$ given $\mathcal{F}_{t}$. A plausible approach is to apply  techniques of Gaussian projection developed in \cite{DL22+} (see also e.g. \cite{BM10} for important development on problems related to a single random matrix). To this end, we define the operation $\Hat{\mathbb{E}}$ as follows: for any function $h$ (of the form $h ( \Gamma, \Pi, \mathrm{BAD}_t, \overrightarrow{Z}, \overrightarrow{\mathsf{Z}}, \Tilde{W}, \Tilde{\mathsf{W}} )$), define 
\begin{equation}\label{eq-def-Hat-E}
    \Hat{\mathbb{E}}[ h ( \Gamma, \Pi, \mathrm{BAD}_t, \overrightarrow{Z}, \overrightarrow{\mathsf{Z}},  \Tilde{W}, \Tilde{\mathsf{W}} ) ] = \mathbb{E} [ h \mid \Gamma, \Pi, \mathrm{BAD}_t ]\,.
\end{equation} Our definition of $\Hat{\mathbb E}$ appears to be simpler than that in \cite{DL22+} thanks to \eqref{eq-independence-observation}. We emphasize that the two definitions are in fact identical, and the simplicity of the expression in \eqref{eq-def-Hat-E} is due to the reason that we have introduced an independent copy of Gaussian process already for the purpose of applying Lindeberg's argument.
Further, when calculating $\Hat{\mathbb{E}}$ with respect to $g_t \Tilde{D}^{(s)}_v$ and $g_t \Tilde{\mathsf{D}}^{(s)}_{\mathsf{v}}$, we regard $g_t \Tilde{D}^{(s)}_v = \psi_{v}^{(s)}$ and $g_t \Tilde{\mathsf{D}}^{(s)}_{\pi(v)} = \psi_{\pi(v)}^{(s)}$ as vector valued functions defined in \eqref{eq-def-vector-value-function-1} and \eqref{eq-def-vector-value-function-2}. For convenience of definiteness, we list variables in $\mathfrak{F}_t$ in the following order: first we list all $\Tilde{W}^{(s)}_v(k) + \langle \eta^{(s)}_k, g_t\Tilde{D}^{(s)}_v \rangle$ indexed by $(s,k,v)$ in the dictionary order and then we list all $\Tilde{\mathsf{W}}^{(s)}_{\mathsf{v}}(k) + \langle \eta^{(s)}_k, g_t \Tilde{\mathsf{D}}^{(s)}_{\mathsf{v}} \rangle$ indexed by $(s,k,\mathsf{v})$ in the dictionary order. Since $\{ \Tilde{W}^{(s)}_v(k), \Tilde{\mathsf{W}}^{(s)}_{\mathsf{v}}(k) \}$ are i.i.d.\ standard Gaussian variables, it suffices to calculate correlations between variables in the collection $\{ \langle \eta^{(s)}_k, g_t\Tilde{D}^{(s)}_v \rangle, \langle \eta^{(s)}_k, g_t \Tilde{\mathsf{D}}^{(s)}_{\mathsf{v}} \rangle \}$. Under this ordering, on the event $\mathcal{E}_{t}$ we have for all $r,s \leq t$
\begin{align*}
    \Hat{\mathbb{E}} [ g_t\Tilde{D}^{(r)}_v (k) g_t\Tilde{D}^{(s)}_u (l) ] = \Hat{\mathbb{E}} [ g_t\Tilde{\mathsf{D}}^{(r)}_{\mathsf{v}}(k) g_t\Tilde{\mathsf{D}}^{(s)}_{\mathsf{u}} (l) ] =0 \mbox{ for } u \neq v, \mathsf{u \neq v} \,,
\end{align*}
since $\{\overrightarrow{Z}_{u,w}: u\neq w\}$ is a collection of independent variables (and similarly for $\overrightarrow{\mathsf Z}$). Also, for all $0 \leq r,s \leq t$ (recall that $\mathfrak{a}_s = \mathfrak{a}$ for $s \geq 1$ and $\mathfrak{a}_0 = \vartheta$), on the event $\mathcal T_t$ we have
\begin{align*}
    & \Hat{\mathbb{E}} [ g_t\Tilde{D}^{(r)}_v (k) g_t\Tilde{D}^{(s)}_v (l) ] = \frac{1}{ n\sqrt{(\mathfrak{a}_r-\mathfrak{a}_r^2)(\mathfrak{a}_s-\mathfrak{a}_s^2)} } \sum_{w \in V \setminus \mathrm{BAD}_{t}} (\mathbf{1}_{w \in \Gamma^{(r)}_k} - \mathfrak{a}_r) (\mathbf{1}_{w \in \Gamma^{(s)}_l} - \mathfrak{a}_s) \\
    =& \frac{ \big| \Gamma^{(r)}_k \cap \Gamma^{(s)}_l \setminus \mathrm{BAD}_{t} \big| - \mathfrak{a}_s \big| \Gamma^{(r)}_k \setminus \mathrm{BAD}_{t} \big| - \mathfrak{a}_r \big| \Gamma^{(s)}_l \setminus \mathrm{BAD}_{t} \big| + \mathfrak{a}_r \mathfrak{a}_s (n-|\mathrm{BAD}_{t}|)  }{ n \sqrt{(\mathfrak{a}_r-\mathfrak{a}_r^2) (\mathfrak{a}_s-\mathfrak{a}_s^2)}}  \\
    =& \frac{\big|\Gamma^{(r)}_k \cap \Gamma^{(s)}_l \big| - \mathfrak{a}_s \big|\Gamma^{(r)}_k \big| - \mathfrak{a}_r \big|\Gamma^{(s)}_l \big| +  \mathfrak{a}_r \mathfrak{a}_s n}{n \sqrt{(\mathfrak{a}_r-\mathfrak{a}_r^2) (\mathfrak{a}_s-\mathfrak{a}_s^2)} } + O\big(\frac{ |\mathrm{BAD}_{t}| }{n \vartheta}\big) = \mathrm{M}^{(r,s)}_{\Gamma} (k,l) + o( \Delta^2_{0} ) \,,
\end{align*}
where the last transition follows from \eqref{eq-def-mathcal-T-t}. Similarly we have (again on $\mathcal T_t$)
\begin{align*}
    & \Hat{\mathbb{E}} [ g_t\Tilde{\mathsf{D}}^{(r)}_{\mathsf{v}} (k) g_t\Tilde{\mathsf{D}}^{(s)}_{\mathsf{v}} (l) ] = \mathrm{M}^{(r,s)}_{\Pi} (k,l) + o( \Delta^2_{0} ) \,, \\
    & \Hat{\mathbb{E}} [ g_t\Tilde{{D}}^{(r)}_v (k) g_t\Tilde{\mathsf{D}}^{(s)}_{\pi(v)} (l) ] = \mathrm{P}^{(r,s)}_{\Gamma,\Pi} (k,l) + o( \Delta^2_{0} ) \,.
\end{align*}
Thus, on $\mathcal E_t \cap \mathcal T_t$ we have
\begin{align}
    \Hat{\mathbb{E}}[\langle \eta^{(r)}_k, g_t\Tilde{D}^{(r)}_v \rangle \langle \eta^{(s)}_m, g_t\Tilde{D}^{(s)}_u \rangle] &= \Hat{\mathbb{E}}[ \langle \eta^{(r)}_k, g_t\Tilde{\mathsf{D}}^{(r)}_{\mathsf{v}} \rangle \langle \eta^{(s)}_m, g_t\Tilde{\mathsf{D}}^{(s)}_{\mathsf{u}} \rangle ]=0  \mbox{ for } u\neq v, \mathsf u\neq \mathsf v\,, \label{equ-degree-correlation-1} \\
    \Hat{\mathbb{E}}[\langle \eta^{(r)}_k, g_t\Tilde{D}^{(r)}_v \rangle \langle \eta^{(s)}_m, g_t\Tilde{D}^{(s)}_v \rangle] &= \eta^{(r)}_k \mathrm{M}_{\Gamma}^{(r,s)} \big( \eta^{(s)}_m \big)^{*} + o(K_t \Delta_0^2)  \nonumber  \\ 
    & \overset{ \eqref{equ-linear-space}, \eqref{equ-vector-orthogonal} }{=} o( K_t \Delta_0^2) \mbox{ for } (r,k) \neq (s,m) \,,  \label{equ-degree-correlation-2}\\
    \Hat{\mathbb{E}}[\langle \eta^{(r)}_k, g_t\Tilde{\mathsf{D}}^{(r)}_{\mathsf{v}} \rangle \langle \eta^{(s)}_m, g_s\Tilde{\mathsf{D}}^{(s)}_{\mathsf{v}} \rangle] &= \eta^{(r)}_k \mathrm{M}_{\Pi}^{(r,s)} \big( \eta^{(s)}_m \big)^{*} + o(K_t \Delta_0^2)  \nonumber \\
    &\overset{ \eqref{equ-linear-space}, \eqref{equ-vector-orthogonal} }{=} o(K_t \Delta_0^2 ) \mbox{ for } (r,k) \neq (s,m)  \,. \label{equ-degree-correlation-3}  
\end{align}
In addition, we have
\begin{align}
    \Hat{\mathbb{E}}[\langle \eta^{(r)}_k, g_t\Tilde{{D}}^{(r)}_{v} \rangle \langle \eta^{(s)}_m, g_s\Tilde{\mathsf{D}}^{(s)}_{\pi(v)} \rangle] &= \Hat{\rho} \cdot \eta^{(r)}_k \mathrm{P}_{\Gamma,\Pi}^{(r,s)} \big( \eta^{(s)}_m \big)^{*} + o(K_t \Delta_0^2) \nonumber \\
    &\overset{ \eqref{equ-linear-space}, \eqref{equ-vector-orthogonal} }{=}  o(K_t \Delta_t ) \mbox{ for } (r,k) \neq (s,m)  \,, \label{equ-degree-correlation-4}
\end{align}
where in the second equality above we also used that the entries of $\mathrm{P}^{(r,s)}_{\Gamma,\Pi}, \mathrm{M}^{(r,s)}_{\Gamma}, \mathrm{M}^{(r,s)}_{\Pi}$ are bounded by $\Delta_{t}$ when $r \neq s$, and $\mathrm{P}^{(s,s)}_{\Gamma,\Pi}, \mathrm{M}^{(s,s)}_{\Gamma}, \mathrm{M}^{(s,s)}_{\Pi}$ concentrate around $\Psi^{(s)},\Phi^{(s)},\Phi^{(s)}$ respectively with error $\Delta_t$. In addition, we have that
\begin{align}
    & \Big| \Hat{\mathbb{E}}[ \langle \eta^{(r)}_k, g_t\Tilde{D}^{(r)}_v \rangle^2 ] - 1 \Big|  = \Big| \eta^{(r)}_k \mathrm{M}_{\Gamma}^{(r,r)} \big( \eta^{(r)}_k \big)^{*} -1 \Big| + o( K_t \Delta_0^2) \nonumber \\
    = \ & \Big| \eta^{(r)}_k \Phi^{(r)} \big( \eta^{(r)}_k \big)^{*} + \eta^{(r)}_k \big( \mathrm{M}_{\Gamma}^{(r,r)}-\Phi^{(r)} \big) \big( \eta^{(r)}_k \big)^{*} -1 \Big| + o(K_t \Delta_0^2) \overset{ \eqref{equ-vector-unit} }{\leq} K_t \Delta_t  \,,  \label{equ-degree-variance-1} \\
    &\Big| \Hat{\mathbb{E}}[\langle \eta^{(r)}_k, g_t\Tilde{\mathsf{D}}^{(r)}_{\mathsf{v}} \rangle^2 ] - 1 \Big| = \Big| \eta^{(r)}_k \mathrm{M}_{\Pi}^{(r,r)} \big( \eta^{(r)}_k \big)^{*} -1 \Big| + o(K_t \Delta_0^2) \nonumber\\
    = \ & \Big| \eta^{(r)}_k \Phi^{(r)} \big( \eta^{(r)}_k \big)^{*} +  \eta^{(r)}_k  \big( \mathrm{M}_{\Pi}^{(r,r)} - \Phi^{(r)} \big) \big( \eta^{(r)}_k \big)^{*} - 1 \Big| + o(K_t \Delta_0^2) \overset{ \eqref{equ-vector-unit} }{\leq} K_t \Delta_t  \,,  \label{equ-degree-variance-2} \\
    &\Big| \Hat{\mathbb{E}}[ \langle \eta^{(r)}_k, g_t\Tilde{{D}}^{(r)}_{\mathsf{v}} \rangle \langle \eta^{(r)}_k, g_t \Tilde{\mathsf{D}}^{(r)}_{\pi(v)} \rangle ] - \Hat{\rho} \cdot \eta^{(r)}_k \Psi^{(r)} \big( \eta^{(r)}_k \big)^{*} \Big|  \nonumber\\
    = \ & \Hat{\rho} \cdot \Big|  \eta^{(r)}_k  \big( \mathrm{P}_{\Gamma, \Pi}^{(r,r)} - \Psi^{(r)} \big) \big( \eta^{(r)}_k \big)^{*}  \Big| + o(K_t \Delta_0^2) \leq K_t \Delta_t  \,.  \label{equ-degree-covariance}
\end{align}
Recall that $\mathfrak{F}_t$ consists of variables in \eqref{eq-def-mathcal-F-t} where $\{ W^{(t)}_v(k), \mathsf{W}^{(t)}_{\mathsf v}(k) \}$ is a collection of standard Gaussian variables independent with $\{ \langle \eta^{(t)}_k, g_t \Tilde{D}^{(t)}_v \rangle, \langle \eta^{(t)}_k, g_t \Tilde{\mathsf D}^{(t)}_{\mathsf v} \rangle \}$.
Therefore, we may write the $\Hat{\mathbb{E}}$-correlation matrix of $\mathfrak F_t$ as $\begin{pmatrix} \mathrm{I} + \mathbf{A}_t & \mathbf{B}_t \\ \mathbf{B}_t^{*} & \mathrm{I} + \mathbf{C}_t \end{pmatrix}$, such that the following hold:
\begin{itemize}
\item $\mathbf{A}_t, \mathbf{C}_t$ have diagonal entries in $(1-K_t \Delta_t,1+K_t \Delta_t)$;
\item for each fixed $(s,k,u)$, there are at most $2K_t$ non-zero non-diagonal $\mathbf{A}_t ((s,k,u);(r,l,u))$ (and also the same for $\mathbf{C}_t$) and these entries are all bounded by $K_t \Delta_t$ (this fact implies that $\| \mathbf{A}_t - \mathrm{I} \|_{\mathrm{op}}, \| \mathbf{C}_t - \mathrm{I} \|_{\mathrm{op}} = O(K_t^2 \Delta_t)$);
\item $\mathbf{B}_t$ is the matrix with row indexed by $(s,k,u)$ for $0 \leq s \leq t, 1 \leq k \leq \frac{K_s}{12}, u \in V \setminus \mathrm{BAD}_t$ and column indexed by $(r,l,\mathsf{w})$ for $0 \leq r \leq t, 1 \leq l \leq \frac{K_s}{12}, \mathsf{w} \in \mathsf{V} \setminus \pi(\mathrm{BAD}_t)$, and with entries $\mathbf{B}_t((s,k,u);(r,l,\mathsf{w}))$ given by $\Hat{\mathbb{E}}[\langle \eta^{(s)}_k, g_s\Tilde{D}^{(s)}_u \rangle  \langle \eta^{(r)}_l, g_r\Tilde{\mathsf{D}}^{(r)}_{\mathsf{w}} \rangle]$.
\end{itemize} 
Thus, by \cite[Lemma 3.10]{DL22+} we have
\begin{align}
    \mathbb{E}[ \langle \eta^{(t+1)}_k, g_{t} \Tilde{D}^{(t+1)}_v \rangle | \mathcal{F}_{t}] = 
    \begin{pmatrix}
        g_t[\Tilde{Y}]_{t} & g_t[\Tilde{\mathsf{Y}}]_{t}
    \end{pmatrix}
    \begin{pmatrix}
        \mathrm{I} + \mathbf{A}_t  & \mathbf{B}_t \\
        \mathbf{B}_t^{*} & \mathrm{I} + \mathbf{C}_t
    \end{pmatrix}^{-1}
    \begin{pmatrix}
        H_{t+1,k,v}^{*} \\ 
        \mathsf{H}_{t+1,k,v}^{*} 
    \end{pmatrix} \,.
    \label{equ-formulate-projection}
\end{align}
Here $H_{t+1,k,v},\mathsf{H}_{t+1,k,v}$ and $g_t[\Tilde{Y}]_t, g_t[\Tilde{\mathsf{Y}}]_t$ are all $\sum_{0 \leq s \leq t} \frac{K_s}{12}(n - |\mathrm{BAD}_t|)$ dimensional vectors; $H_{t+1,k,v}$ and $g_t[\Tilde{Y}]_t$ are indexed by triple $(s,l,u)$ with $0 \leq s \leq t, 1 \leq k \leq \frac{K_s}{12}, u \in V \setminus \mathrm{BAD}_t$ in the dictionary order; $\mathsf{H}_{t+1,k,v}$ and $g_t[\Tilde{\mathsf{Y}}]_t$ are indexed by triple $(s,l,\mathsf{u})$ with $0 \leq s \leq t, 1 \leq k \leq \frac{K_s}{12}, \mathsf{u} \in \mathsf{V} \setminus \pi(\mathrm{BAD}_t)$ in the dictionary order. Also $g_t[\Tilde{Y}]_t$ and $g_t[\Tilde{\mathsf{Y}}]_t$ can be divided into sub-vectors as follows:
\begin{align*}
    g_t[\Tilde{Y}]_t = [ g_t \Tilde{Y}_t \, | \,g_t \Tilde{Y}_{t-1} \,|\, \ldots \,|\, g_t \Tilde{Y}_0 ] \mbox{ and } g_t [\Tilde{\mathsf{Y}}]_t = [ g_t\Tilde{\mathsf{Y}}_t \,|\, g_t\Tilde{\mathsf{Y}}_{t-1} \,|\, \ldots \,|\, g_t\Tilde{\mathsf{Y}}_0 ] \,,
\end{align*}
where $g_t \Tilde{Y}_s$ and $g_t \Tilde{\mathsf{Y}}_s$ are $\frac{K_s}{12}(n-|\mathrm{BAD}_t|)$ dimensional vectors indexed by $(k,u)$ and $(k,\mathsf{u})$, respectively. In addition, their entries are given by
\begin{equation*}
\begin{aligned}
    & g_t \Tilde{Y}_s (l,u) = \Tilde{W}^{(s)}_u(l) +  \langle \eta^{(s)}_l, g_t \Tilde{D}^{(s)}_u \rangle, \quad  H_{t+1,k,v} ( s,l,u ) = \Hat{\mathbb{E}}[\langle \eta^{(t+1)}_k, g_t\Tilde{D}^{(t+1)}_v \rangle \langle \eta^{(s)}_l, g_t \Tilde{D}^{(s)}_{{u}} \rangle]\,; \\
    & g_t\Tilde{\mathsf{Y}}_s (l,\mathsf{u}) =  \Tilde{\mathsf{W}}^{(s)}_{\mathsf{u}}(l) + \langle \eta^{(s)}_l, g_t \Tilde{\mathsf{D}}^{(s)}_{\mathsf{u}} \rangle, \quad \mathsf{H}_{t+1,k,v} (s,l,\mathsf{u}) = \Hat{\mathbb{E}}[\langle \eta^{(t+1)}_k, g_t\Tilde{D}^{(t+1)}_v \rangle \langle \eta^{(s)}_l, g_t\Tilde{\mathsf{D}}^{(s)}_{\mathsf{u}} \rangle] \,.
\end{aligned}
\end{equation*}

\begin{Remark} \label{remark-conditional-Gaussian}
In conclusion, we have shown that  
    \begin{align}
    & \Big(\begin{pmatrix}
    g_t \Tilde{Y}_{t+1} & g_t \Tilde{\mathsf{Y}}_{t+1}  
    \end{pmatrix} \big| \mathcal{F}_{t} \Big)  \overset{d}{=} 
    \Big( 
    \begin{pmatrix}
        g_t[\Tilde{Y}]_{t} &
        g_t[\Tilde{\mathsf{Y}}]_{t}  
    \end{pmatrix}
    \begin{pmatrix}
        \mathrm{I} + \mathbf{A}_t  & \mathbf{B}_t \\
        \mathbf{B}_t^{*} & \mathrm{I} + \mathbf{C}_t
    \end{pmatrix}^{-1}
    \mathbf{H}_{t+1}^{*} \big| \mathcal{F}_{t}\Big)
    \label{equ-projection-part}  \\
    & +  \begin{pmatrix}
    g_t \Tilde{Y}_{t+1}^{\diamond} & g_t \Tilde{\mathsf{Y}}_{t+1}^{\diamond}    
    \end{pmatrix} -
    \begin{pmatrix}
        g_t [\Tilde{Y}]_{t}^{\diamond} &
        g_t [\Tilde{\mathsf{Y}}]_{t}^{\diamond}  
    \end{pmatrix}
    \begin{pmatrix}
        \mathrm{I} + \mathbf{A}_t  & \mathbf{B}_t \\
        \mathbf{B}_t^{*} & \mathrm{I} + \mathbf{C}_t
    \end{pmatrix}^{-1} 
    \mathbf{H}_{t+1}^{*}  \,.
    \label{equ-Gaussian-part}
\end{align}
In the above $\mathbf{H}_{t+1}$ is given by 
\begin{equation}\label{eq-def-H-entry}
    \mathbf{H}_{t+1}( (k,\tau_1);(s,l,\tau_2) ) = H_{t+1,k,\tau_1}(s,l,\tau_2) \mbox{ for } \tau_1, \tau_2 \in ( V \setminus \mathrm{BAD}_t ) \cup (\mathsf{V} \setminus \pi(\mathrm{BAD}_t) ) \,.
\end{equation}
In addition, $g_t \Tilde{Y}^{\diamond}_s$ is given by $g_t \Tilde{Y}^{\diamond}_s (l,v) = \Tilde{W}^{(t)}_v(l) + \langle \eta^{(s)}_l, (g_t\Tilde{D}^{(s)}_v)^{\diamond} \rangle$, where
\begin{align*}
    g_t\Tilde{D}^{(t)}_v(k)^{\diamond} =  \frac{1}{\sqrt{n(\mathfrak{a}_t - \mathfrak{a}_t^2)}} \sum_{u \in V \setminus \mathrm{BAD}_t }( \mathbf{1}_{u \in \Gamma^{(t)}_k} - \mathfrak{a}_t) \overrightarrow{Z}_{v,u}^{\diamond}
\end{align*}
is a linear combination of Gaussian variables $\{\overrightarrow{Z}_{v,u}^{\diamond}\}$ with coefficients fixed  (recall \eqref{eq-convention-announcement}), and $\{\overrightarrow{Z}_{v,u}^{\diamond}\}$ is an independent copy of $\{\overrightarrow{Z}_{v,u}\}$ (and similarly for $g_t \mathsf{D}^{(t)}_{\mathsf v}(k)^{ \diamond }$). For notational convenience, we denote \eqref{equ-projection-part} as $\textup{PROJ} \big( ( g_t \Tilde{Y}_{t+1} , g_t \Tilde{\mathsf{Y}}_{t+1} ) \big)$, which is a vector with entries given by  (the analogue of below for the mathsf version also holds)
\begin{align*}
    \textup{PROJ} \big( ( g_t \Tilde{Y}_{t+1} , g_t \Tilde{\mathsf{Y}}_{t+1} ) \big) (k,v) &= \mathrm{PROJ}(g_t \Tilde{Y}_{t+1}(k,v))\\
    & = \textup{PROJ}( \Tilde{W}^{(t+1)}_v(k) + \langle \eta^{(t+1)}_k, g_t\Tilde{D}^{(t+1)}_v \rangle) \mbox{ for } v\in V \setminus \mathrm{BAD}_t\,.
\end{align*}
We also denote \eqref{equ-Gaussian-part} as $\begin{pmatrix} g_t \Tilde{Y}_{t+1}^{\diamond} - \mathrm{GAUS}(g_t \Tilde{Y}_{t+1}) & g_t \Tilde{\mathsf{Y}}_{t+1}^{\diamond} - \mathrm{GAUS}(g_t \Tilde{\mathsf{Y}}_{t+1}) \end{pmatrix}$. We will further use the notation $\mathrm{PROJ}(g_t Y_{t+1}(k,v))$ (note that there is no tilde here) to denote the projection obtained from $\mathrm{PROJ}(g_t \Tilde{Y}_{t+1}(k,v))$ by replacing each $\overrightarrow{Z}_{v,u}, \Tilde{W}^{(s)}_v(k)$ with $\overrightarrow{G}_{v,u}, W^{(s)}_v(k)$ therein (and similarly for $\mathrm{PROJ}(g_t \mathsf Y_{t+1}(k, \mathsf v))$).
\end{Remark}

Denote $ \mathbf{Q}_t = \begin{pmatrix} \mathrm{I} + \mathbf{A}_t & \mathbf{B}_t \\ \mathbf{B}_t^{*} & \mathrm{I} + \mathbf{C}_t \end{pmatrix}^{-1}$. We next control norms of these matrices.

\begin{Lemma}{\label{lemma-op-norm-Q}}
    On the event $\mathcal{E}_{t} \cap \mathcal{T}_t$, we have $\| \mathbf{Q}_t \|_{\mathrm{op}} \leq 100$ if $\Hat{\rho}<0.1$.
\end{Lemma}
\begin{Lemma}{\label{lemma-1-norm-Q}}
    On the event $\mathcal{E}_{t} \cap \mathcal{T}_t$, we have $\| \mathbf{Q}_t \|_{\infty} = \| \mathbf{Q}_t \|_{1} \leq 100 K_t^{10} \vartheta^{-2}$.
\end{Lemma}
Similar versions of Lemmas~\ref{lemma-op-norm-Q} and \ref{lemma-1-norm-Q} were proved in \cite[Lemmas 3.13 and 3.15]{DL22+}, and the proofs can be adapted easily. Indeed, by proofs in \cite{DL22+}, in order to bound the operator norm it suffices to use the fact that $\langle \eta^{(s)}_k, g_t \Tilde{D}^{(s)}_v \rangle \in \mathrm{span}\{ \overrightarrow{Z}_{u,w} \}, \langle \eta^{(s)}_k, g_t \Tilde{\mathsf{D}}^{(s)}_{\mathsf{v}} \rangle \in \mathrm{span}\{ \overrightarrow{\mathsf{Z}}_{\mathsf{u,w}} \}$. Also, in order to bound the $\infty$-norm, it suffices to show that the operator norm is bounded by a constant and that $\mathbf{Q}^{-1}_t ( (s,k,v);(s^{\prime},k^{\prime},\mathsf{u}) ) = O(K_t \vartheta^{-1})$ when $\mathsf{u} \neq \pi(v)$. All of these can be easily checked and thus we omit further details.

\begin{Lemma} {\label{lemma-op-norm-H}}
On the event $\mathcal{E}_t \cap \mathcal{T}_{t-1}$, we have 
$$\| \mathbf{H}_{t} \|_{\mathrm{HS}}^2 \leq n K_t^4 \Delta_t^2\,, \quad \| \mathbf{H}_{t} \|_{\infty}, \| \mathbf{H}_{t} \|_{1} \leq \frac{K^3_t}{\sqrt{\vartheta}} \, \mbox{ and } \,\| \mathbf{H}_{t} \|_{\mathrm{op}} \leq 2 K_{t}^3\,.$$
\end{Lemma}
\begin{proof}
Recall \eqref{eq-def-H-entry}. By \eqref{equ-degree-correlation-1}, \eqref{equ-degree-correlation-2}, \eqref{equ-degree-correlation-3} and \eqref{equ-degree-correlation-4} we get $\mathbf{H}_{t} ((k,v);(s,l,u)) = 0$ for $u \neq v$ and that $\mathbf{H}_{t} ((k,v);(s,l,v)), \mathbf{H}_{t} ((k,v);(s,l,\pi(v))) = O(K_{t} \Delta_{t})$; The similar results also hold for $\mathbf{H}_{t} ((k,\mathsf v);(s,l,\mathsf u))$. In addition, for $\pi(v) \neq \mathsf{u} $ we have that $ \mathbf{H}_{t} ( (k,v);(s,l,\mathsf{u}) )$ is equal to
\begin{align*}
    &\Hat{\mathbb{E}} \Bigg[  \frac{ \sum_{i=1}^{K_{t}} \sum_{j=1}^{K_s} \sum_{w_1,w_2} \eta^{(t)}_k(i) \eta^{(s)}_l (j) (\mathbf{1}_{w_1 \in \Gamma^{(t)}_i} - \mathfrak{a}_t) (\mathbf{1}_{\pi(w_2) \in \Pi^{(s)}_j} - \mathfrak{a}_s) \overrightarrow{Z}_{v,w_1} \overrightarrow{\mathsf{Z}}_{\mathsf{u},\pi(w_2)} }{ n \Hat{q}(1-\Hat{q}) \sqrt{(\mathfrak{a}_t-\mathfrak{a}_t^2) (\mathfrak{a}_s-\mathfrak{a}_s^2) } } \Bigg]  \\
    &= O \Big( \frac{ K^2_{t} }{ n \sqrt{\mathfrak{a}_t \mathfrak{a}_s} }   (\mathbf{1}_{ \pi(v) \in \cup_{j} \Pi^{(s)}_j} - \mathfrak{a}_s) (\mathbf{1}_{u \in \cup_{i} \Gamma^{(t)}_i} - \mathfrak{a}_t) \Big)
\end{align*}
and a similar bound applies to $(\mathsf u, v)$ with $v\in V, \mathsf u \in \mathsf V$ and $ \mathsf u \neq \pi(v)$. Combined with Items (i) and (iv) in Definition~\ref{def-admissible} for $\mathcal E_t$, it yields that $\| \mathbf{H}_{t} \|_{\mathrm{HS}}^2 \leq n K_t^4 \Delta_t^2$ and $\| \mathbf{H}_{t} \|_{\infty}, \| \mathbf{H}_{t} \|_{1} \leq \frac{K^3_t}{\sqrt{\vartheta}}$. 

We next bound $\| \mathbf{H}_{t} \|_{\mathrm{op}}$. Applying Lemma~\ref{lemma-bound-on-op-norm} by setting  $\delta = K_t \Delta_t$, $C^2 = K_t^6$ and $\mathcal{I}_{v} = \mathcal{J}_{v} = \{ (s,k,v),(s,k,\pi(v)): 0\leq s\leq t, 1\leq k\leq K_s/12 \}$, we can then derive that $\| \mathbf{H}_t \|_{\mathrm{op}} \leq K_t^3 + 4K_t^2 \Delta_t \leq 2 K_t^3$.
\end{proof}

Remark~\ref{remark-conditional-Gaussian} provides an explicit expression for the conditional law of of $\Tilde{W}^{(t+1)}_v(k) + \langle \eta^{(t+1)}_k, g_t \Tilde{D}^{(t+1)}_v \rangle$. However, the projection $\mathrm{PROJ}(g_t \Tilde{Y}_{t+1}(k,v) )$ is not easy to deal with since the expression of every variable in $\mathfrak F_t$ relies on  $\mathrm{BAD}_t$ (even for those indexed with $s<t$). A more tractable ``projection'' is the projection of $\Tilde{W}^{(t+1)}_v(k) + \langle \eta^{(t+1)}_k, g_t \Tilde{D}^{(t+1)}_v \rangle$ onto $\mathcal F'_t$ for $\mathcal F'_t = \sigma(\mathfrak F'_t)$ where
\begin{align}
   \mathfrak F'_t = \Bigg\{ \begin{split} & \Tilde{W}^{(s)}_u(l) + \langle \eta^{(s)}_l, g_{s-1} \Tilde{D}^{(s)}_u \rangle \\ & \Tilde{\mathsf{W}}^{(s)}_{\pi(u)} (l) + \langle \eta^{(s)}_l , g_{s-1} \Tilde{\mathsf{D}}^{(s)}_{\pi(u)} \rangle   \end{split} : 0 \leq s \leq t, 1 \leq l \leq \frac{K_s}{12}, u \not \in \mathrm{BAD}_t \Bigg\}. \label{eq-def-mathfrak-F-t'}
\end{align}
We can similarly show that $\mathbb{E}[ ( g_t \Tilde{Y}_{t+1} , g_t \Tilde{\mathsf{Y}}_{t+1} ) | \mathcal{F}^{\prime}_t]$ (recall that the conditional expectation is the same as the projection) has the form 
\begin{align}
    \begin{pmatrix} {\mathrm{PROJ}}^{\prime} (g_t \Tilde{Y}_{t+1}) & {\mathrm{PROJ}}^{\prime}  (g_t \Tilde{\mathsf{Y}}_{t+1}) \end{pmatrix}  = \begin{pmatrix} [g \Tilde{Y}]_t & [g\Tilde{\mathsf{Y}}]_t \end{pmatrix} \mathbf{P}_t \mathbf{J}_{t+1}^{*}  \,,
    \label{equ-modified-projection-form}
\end{align}
where $[g \Tilde{Y}]_t (s,l,u) = \Tilde{W}^{(s)}_l(u) + \langle \eta^{(s)}_l, g_{s-1} \Tilde{D}^{(s)}_u \rangle$, ${\mathrm{PROJ}}^{\prime} (g_t \Tilde{Y}_{t+1})(k,v) = {\mathrm{PROJ}}^{\prime} (g_t \Tilde{Y}_{t+1}(k,v))$ is the projection of $\Tilde{W}^{(t+1)}_k(v) + \langle \eta^{(t+1)}_k, g_t \Tilde{D}^{(t+1)}_v \rangle$ onto $\mathcal{F}^{\prime}_t$, $\mathbf{J}_{t+1}$ is defined by 
\begin{equation}\label{eq-def-J}
    \mathbf{J}_{t+1} ((k,v),(s,l,u)) = \Hat{\mathbb{E}}[ \langle \eta^{(t+1)}_k, g_{t} \Tilde{D}^{(t+1)}_v \rangle \langle \eta^{(s)}_l, g_{s-1} \Tilde{D}^{(s)}_u \rangle ] \mbox{ for } u, v \not\in \mathrm{BAD}_t \,,
\end{equation}
and $\mathbf{P}_t^{-1}$ is defined to be the covariance matrix of $\mathfrak F'_t$.
Adapting proofs of Lemmas~\ref{lemma-op-norm-Q}, \ref{lemma-1-norm-Q} and \ref{lemma-op-norm-H}, we can show that under $\mathcal{E}_{t+1} \cap \mathcal{T}_t$, we have $\| \mathbf{J}_{t+1} \|_{\mathrm{op}} \leq 2K_t^3$ and $\| \mathbf{P}_t \|_{\mathrm{op}} \leq 100$.

Similarly, we denote by ${\mathrm{PROJ}}^{\prime} (g_t Y_{t+1})$ a vector obtained from replacing the substituted  Gaussian entries with the original Bernoulli variables in the projection ${\mathrm{PROJ}}^{\prime} (g_t \Tilde{Y}_{t+1})$ (i.e, on the right hand side of \eqref{equ-modified-projection-form}). The projection $\mathrm{PROJ}^{\prime}(g_t Y_{t+1})$ is more tractable than $\mathrm{PROJ}(g_t Y_{t+1})$ (recall its definition in Remark~\ref{remark-conditional-Gaussian}) since $W_v^{(s)}(k) + \langle \eta^{(s)}_k, g_{s-1} D^{(s)}_v \rangle$ is measurable with respect to $\mathfrak{S}_{s}$ (recall \eqref{eq-def-mathfrak-S-t}) and thus we can use induction to control $\mathrm{PROJ}^{\prime}$. Another advantage is that the matrix $\mathbf{P}_t$ is measurable with $\mathfrak{S}_{t-1}$.
\begin{Lemma} {\label{lemma-projection-replace}}
On the event $\mathcal{E}_{t+1} \cap \mathcal{T}_t$, we have
\begin{align*}
    \sum_{v \not \in \mathrm{BAD}_{t+1}}
    \big| \mathrm{PROJ}(g_{t}Y_{t+1}(k,v)) -  {\mathrm{PROJ}}^{\prime} (g_{t}Y_{t+1}(k,v)) \big|^2 \leq n \Delta_{t+1}^2 + \Delta_{t+1}^2 \big\| \begin{pmatrix} [g Y]_t & [g \mathsf{Y}]_t \end{pmatrix} \big\|^2 \,.
\end{align*}
And similar result holds for $\mathrm{PROJ} (g_t\mathsf{Y}_{t+1}) - {\mathrm{PROJ}}^{\prime} (g_t\mathsf{Y}_{t+1})$.
\end{Lemma}
\begin{proof}
By the triangle inequality, the left hand side in the lemma-statement is given by
\begin{align}
    & \sum_{v \not \in \mathrm{BAD}_{t+1}} \Big( \begin{pmatrix} g_t [Y]_t & g_t[\mathsf{Y}]_t \end{pmatrix} \mathbf{Q}_t \mathbf{H}_{t+1}^{*} (k,v) - \begin{pmatrix} [g Y]_t & [g \mathsf{Y}]_t \end{pmatrix} \mathbf{P}_t \mathbf{J}_{t+1}^{*} (k,v) \Big)^2 \nonumber\\
    \leq \ & 2 \sum_{v \not \in \mathrm{BAD}_{t+1}} \Big( \begin{pmatrix}
        [g Y]_{t} - g_t [Y]_t & [ g\mathsf{Y} ]_{t} - g_t[\mathsf{Y}]_t
    \end{pmatrix} \mathbf{Q}_t \mathbf{H}_{t+1}^{*} (k,v) \Big)^2 \label{equ-dif-proj-part-I} \\
    + \ & 2 \big\| \begin{pmatrix}
        [g Y]_t & [g \mathsf{Y}]_t
    \end{pmatrix} (\mathbf{P}_t \mathbf{J}_{t+1}^{*} - \mathbf{Q}_t \mathbf{H}_{t+1}^{*} ) \big\|^2 \,. \label{equ-dif-proj-part-II}
\end{align}
By \eqref{equ-def-set-PRB}, we have 
\begin{align*}
\eqref{equ-dif-proj-part-I} =     \sum_{v \not \in \mathrm{BAD}_{t+1}}  \Big( \begin{pmatrix}
        [g Y]_{t} - g_t [Y]_t & [ g\mathsf{Y} ]_{t} - g_t[\mathsf{Y}]_t
    \end{pmatrix} \mathbf{Q}_t 
    \begin{pmatrix}
        H_{t+1,k,v} & \mathsf{H}_{t+1,k,v}
    \end{pmatrix}^{*} \Big)^2 \leq n \Delta_{t+1}^2\,.
\end{align*}
It remains to bound \eqref{equ-dif-proj-part-II}. Note that for $u, v\in V$ and $\tau_u \in \{u, \pi(u)\}$ and $\tau_v \in \{v, \pi(v)\}$
\begin{align*}
    \big( \mathbf{J}_{t+1} - \mathbf{H}_{t+1} \big) ((k,\tau_v);(s,l, \tau_u)) =
    \begin{cases}
    0 , & v \neq u, u,v \not \in \mathrm{BAD}_t ; \\
    O(\frac{1}{n}), & v \neq u, v \in \mathrm{BAD}_t \mbox{ or } u \in \mathrm{BAD}_t ; \\
    O(|\mathrm{BAD}_t|/n), &v=u \,.
    \end{cases}
\end{align*}
As in the proof of Lemma~\ref{lemma-op-norm-H}, we can choose $\mathcal{I}_v = \mathcal{J}_v = \{ (s,k,v),(s,k,\pi(v)) : 0 \leq s \leq t, 1 \leq k \leq K_s/12 \}$, $\delta= \frac{ |\mathrm{BAD}_t| }{ n } $ and $C^2 = \frac{ |\mathrm{BAD}_t| }{ n } \ll \Delta_t^{10}$ (since we are on the event $\mathcal{T}_t$). We can then apply Lemma~\ref{lemma-bound-on-op-norm} and get that $\|\mathbf{J}_{t+1} - \mathbf{H}_{t+1}\|_{\mathrm{op}} \ll 2 \Delta_{t+1}^{5}$. Again by applying Lemma~\ref{lemma-bound-on-op-norm} for such $\mathcal{I}_v , \mathcal{J}_v, \delta$ and $C^2$, we get  $\|\mathbf{Q}_t^{-1} - \mathbf{P}_t^{-1}\|_{\mathrm{op}} \ll 2\Delta_{t+1}^{5}$. Thus,
\begin{align*}
    \| \mathbf{Q}_t - \mathbf{P}_t \|_{\mathrm{op}} \leq \| \mathbf{Q}_t \|_{\mathrm{op}} \| \mathbf{Q}_t^{-1} - \mathbf{P}_t^{-1} \|_{\mathrm{op}} \| \mathbf{P}_t \|_{\mathrm{op}} \leq 100^2  \Delta_{t+1}^{5}\,.
\end{align*}
Since in addition $\mathbf{J}_{t+1}, \mathbf{H}_{t+1}, \mathbf{P}_t, \mathbf{Q}_t$ have operator norms bounded by $O(K^3_{t+1})$, we get that
\begin{align*}
    \| \mathbf{P}_t \mathbf{J}_{t+1}^{*} - \mathbf{Q}_t \mathbf{H}_{t+1}^{*} \|_{\mathrm{op}} \leq \| \mathbf{P}_t - \mathbf{Q}_t \|_{\mathrm{op}} \| \mathbf{J}_{t+1} \|_{\mathrm{op}} + \| \mathbf{J}_{t+1} - \mathbf{H}_{t+1} \|_{\mathrm{op}} \| \mathbf{Q}_t \|_{\mathrm{op}} \leq \Delta_{t+1}^2 \,.
\end{align*}
This implies that $\eqref{equ-dif-proj-part-II} \leq \Delta_{t+1}^2 (\| [g Y]_t \|^2 + \| [g\mathsf{Y}]_t \|^2)$, as required.
\end{proof}

\subsection{Proof of Proposition~\ref{prop-cardinality-BAD}} {\label{sec:actual-proof}}

In this subsection, we prove Proposition~\ref{prop-cardinality-BAD} by induction on $t$. Recall the definition of $\mathcal{E}_t$ (see \eqref{equ-def-admissible}) and $\mathcal{T}_t$ (see \eqref{eq-def-mathcal-T-t}). For a given realization $(\mathrm{B}_t, \mathrm{B}_{t-1})$ for $(\mathrm{BAD}_t, \mathrm{BAD}_{t-1})$, define vectors $\Tilde{W}_t$ and $\Tilde{\mathsf{W}}_t$ where $\Tilde{W}_t(k,v) = \Tilde{W}^{(t)}_v (k)$ and $\Tilde{\mathsf{W}}_t(k,\pi(v)) = \Tilde{\mathsf{W}}^{(t)}_{\pi(v)} (k)$ for $v\not \in \mathrm{B}_{t}$. We recall $g_{t-1} \Tilde{Y}_s$ and define $g_{t-1} {Y}_s$ as follows:
\begin{equation*}
\begin{aligned}
    & g_{t-1} \Tilde{Y}_s (k,v) = \Tilde{W}^{(s)}_k(v)+ \langle \eta^{(s)}_k, g_{t-1} \Tilde{D}^{(s)}_v \rangle, \quad g_{t-1} {Y}_s (k,v) = W^{(s)}_k(v)+ \langle \eta^{(s)}_k, g_{t-1} {D}^{(s)}_v \rangle \,, \\
    & g_{t-1} \Tilde{\mathsf{Y}}_s (k,\mathsf{v}) = \Tilde{\mathsf{W}}^{(s)}_k(\mathsf{v})+ \langle \eta^{(s)}_k, g_{t-1} \Tilde{\mathsf{D}}^{(s)}_{\mathsf{v}} \rangle, \quad g_{t-1} {\mathsf{Y}}_s (k,\mathsf{v}) = \mathsf{W}^{(s)}_k(\mathsf{v})+ \langle \eta^{(s)}_k, g_{t-1} \mathsf{D}^{(s)}_{\mathsf{v}} \rangle \,,
\end{aligned}
\end{equation*}
where $0 \leq s \leq t, 1 \leq k \leq \frac{K_s}{12}$, and $v, \pi^{-1}(\mathsf v) \not \in \mathrm{B}_{t-1}$ when $s<t$, as well as $v, \pi^{-1}(\mathsf v) \not \in \mathrm{B}_{t}$ when $s=t$. In what follows, we will use $x_t = \{ x^{(t)}_{k,v} \}$ and $\mathsf x_t = \{\mathsf{x}^{(t)}_{k,\mathsf{v}}\}$ to denote realization for $g_{t-1}Y_t$ and $g_{t-1} \mathsf{Y}_t$ respectively. In addition, we define a mean-zero Gaussian process 
\begin{equation} \label{eq-def-Check-Gaussian-process}
    \big\{ \langle \eta^{(s)}_k, \Check{D}^{(s)}_v \rangle, \langle \eta^{(s)}_k, \Check{\mathsf{D}}^{(s)}_{\pi(v)} \rangle : 0 \leq s \leq t, 1 \leq k \leq \tfrac{K_s}{12}, v \in V \big\}
\end{equation}
where each variable has variance 1 and the only non-zero covariance is given by
\begin{align*}
    \mathbb{E}[ \langle \eta^{(s)}_k, \Check{D}^{(s)}_v \rangle \langle \eta^{(s)}_k, \Check{\mathsf{D}}^{(s)}_{\pi(v)} \rangle ] = \Hat{\rho} \eta^{(s)}_k \Psi^{(s)} (\eta^{(s)}_k)^{*} \mbox{ for } 0 \leq s \leq t, 1 \leq k \leq \tfrac{K_s}{12} \mbox{ and } v \in V \,.
\end{align*}
We defined \eqref{eq-def-Check-Gaussian-process} since eventually we will show that this is a good approximation for our actual process (see our definition of $\mathcal B_t$ below). For $v, \pi^{-1}(\mathsf{v}) \not \in \mathrm{B}_{t}$, we also define
\begin{align*}
    \langle \sigma^{(t)}_k, \Check{D}^{(t)}_v \rangle = \sqrt{\frac{12}{K_t}} \sum_{j=1}^{ \frac{K_t}{12} } \beta^{(t)}_k (j) \langle \eta^{(t)}_j, \Check{D}^{(t)}_v \rangle  \mbox{ and }
    \Check{Y}_t(k,v) = \Tilde{W}^{(t)}_v(k) + \langle \eta^{(t)}_k, \Check{D}^{(t)}_v \rangle\,,
\end{align*}
and we make analogous definitions for the mathsf version for $\mathsf v$. For $t\geq 0$, we define
\begin{align*}
    \mathcal{A}_t =& \Big\{  \sum_{v \in V \setminus \mathrm{BAD}_t}  |g_{t-1} Y_t (k,v)|^2 + |g_{t-1} \mathsf{Y}_t (k,\pi(v))|^2  \leq 100 n \mbox{ for all } k \Big\} \,,  \\
    \mathcal{B}_t =& \Big\{ ( g_{t-1} Y_t , g_{t-1} \mathsf{Y}_t ) \in \Big\{ (x_{t}, \mathsf{x}_t): \frac{ p_{ \{ g_{t-1} Y_t , g_{t-1} \mathsf{Y}_t | \mathfrak{S}_{t-1}; \mathrm{BAD}_t \}} (x_{t}, \mathsf{x}_t) }{ p_{\{  \Check{Y}_t , \Check{\mathsf{Y}}_t  \}} (x_{t}, \mathsf{x}_t) } \leq \exp \{ n K_t^{30} \Delta_t^2 \} \Big\} \Big\}\,, \\
    \mathcal{H}_t =& \Big\{ \sum_{v \in V \setminus \mathrm{BAD}_t } \big| \textup{PROJ}(g_{t-1} Y_t(k,v) ) \big|^2+ \big| \textup{PROJ}( g_{t-1} \mathsf{Y}_t (k,\pi(v)) ) \big|^2  \leq n K_t^{6} \Delta^2_t \mbox{ for all } k \Big\} \,.
\end{align*}
Note that $\mathcal{H}_0$ holds obviously. For notational convenience in the induction, we will also denote by $\mathcal A_{-1}$ and $\mathcal B_{-1}$ as the whole space. Also, by Lemma~\ref{lemma-concentration-iterative-aleph-Upsilon} $\mathcal{T}_{-1}$ holds with probability $1-o(1)$ since $\mathrm{BAD}_{-1} = \mathrm{REV}$. In addition, we have $\P(\mathcal{E}_0) = 1-o(1)$ by Lemma~\ref{lemma-E-0-holds}. With these clarified, our inductive proof consists of the following steps:

\noindent{\bf Step 1.} If $\mathcal{T}_{t-1}$ holds for $0 \leq t \leq t^{*}-1$, then $\mathcal{T}_{t}$ holds with probability $1-o(1)$;

\noindent{\bf Step 2.} If $\mathcal{A}_{t-1}, \mathcal{B}_{t-1}, \mathcal{T}_{t}, \mathcal{E}_{t}, \mathcal{H}_{t}$ hold for $0 \leq t \leq t^{*}-1$, then $\mathcal{B}_{t}$ holds with probability $1-o(1)$;

\noindent{\bf Step 3.} If $\mathcal{B}_{t}$ holds for $0 \leq t \leq t^{*}-1$, then $\mathcal{A}_{t}$ holds with probability $1-o(1)$;
   
\noindent {\bf Step 4.} If $\mathcal{A}_t, \mathcal{B}_t, \mathcal{E}_{t}, \mathcal{H}_t, \mathcal{T}_t $ hold for $0 \leq t \leq t^{*}-1$, then $\mathcal{E}_{t+1}$ holds with probability $1-o(1)$;
 
\noindent {\bf Step 5.} If $\mathcal{A}_{t}, \mathcal{B}_{t}, \mathcal{H}_{t}, \mathcal{T}_t, \mathcal{E}_{t+1}$ hold for $0 \leq t \leq t^{*}-1$, then $\mathcal{H}_{t+1}$ holds with probability $1-o(1)$.

\subsubsection{Step 1: \(\mathcal{T}_t\)}

In what follows, we assume $\mathcal{T}_{t-1}$ holds without further notice. As we will see, the philosophy of our proof throughout this subsection is to first consider an arbitrarily fixed realization for e.g. $\{ \Gamma^{(r)}_k, \Pi^{(r)}_k : 0 \leq r \leq t \}, \mathrm{BAD}_{t-1}$ and $\mathrm{BIAS}_{D, t,s,l}$, and then we prove a bound on the tail probability for some ``bad'' event---we shall emphasize that, we will not compute the conditional probability as it will be difficult to implement; instead we will compute the probability (which we denote as $\Hat{\mathbb P}$) by simply treating $\{ \Gamma^{(r)}_k, \Pi^{(r)}_k : 0 \leq r \leq t, 1 \leq k \leq K_r \}, \mathrm{BAD}_{t-1}$ and $\mathrm{BIAS}_{D,t,s,l}$ as deterministic objects. Formally, we define the operation $\Hat{\mathbb{E}}$ as follows: for any function $h$ (of the form $h ( \Gamma, \Pi, \mathrm{BAD}_{t-1}, \overrightarrow{G}, \overrightarrow{\mathsf{G}}, W, \mathsf{W})$) and any realization $\Xi,\mathrm{B}$ for $\{ \Gamma, \Pi \}$ and $\mathrm{BAD}_{t-1}$, define 
    $$ f( \Xi,\mathrm{B} ) = \mathbb{E}_{ \{ \overrightarrow{G}, \overrightarrow{\mathsf{G}}, W, \mathsf{W} \} } \big[ h ( \Xi,\mathrm{B}, \overrightarrow{G}, \overrightarrow{\mathsf{G}}, W, \mathsf{W} ) \big] \,. $$ 
Then the operator $\Hat {\mathbb E}$ is defined such that 
    $$ \Hat{\mathbb{E}} \big[ h ( \Gamma, \Pi, \mathrm{BAD}_{t-1}, \overrightarrow{G}, \overrightarrow{\mathsf{G}}, W, \mathsf{W} ) \big] = f( \Gamma, \Pi, \mathrm{BAD}_{t-1} ) \,. $$ 
Note that this definition of $\Hat{\mathbb E}$ is consistent with that in \cite{DL22+}. It is also consistent with \eqref{eq-def-Hat-E} except that thanks to \eqref{eq-independence-observation} for the special case considered in \eqref{eq-def-Hat-E} simplifications were applied. 

Provided with $\Hat{\mathbb E}$, we can now precisely define $\Hat{\mathbb P}(A) = \Hat{\mathbb E}[\mathbf 1_{A}]$ for any event $A$. 
After bounding the $\Hat{\mathbb P}$-probability, we apply a union bound over all possible realizations which then justifies that the bad event indeed typically will not occur; this union bound is necessary exactly since what we have computed earlier is not the conditional probability. The key to our success is that the tail probability is so small that it can afford a union bound.

\begin{Lemma} \label{lemma-bound-cardinality-BIAS}
We have 
\begin{align*}
    \mathbb P \Big( |\mathrm{BIAS}_{t}| \leq 8 \vartheta^{-1} K^4_t \big( |\mathrm{BAD}_{t-1}| + \sqrt{n / \Hat{q}(1-\Hat{q})} \big) e^{ 20 (\log \log n)^{10} } \Big) \geq  1-o(e^{-n K_t})\,.
\end{align*}
\end{Lemma}
\begin{proof}
Recall the definition of $\mathrm{BIAS}_{D,t,s,k}$ as in \eqref{equ-def-set-BIAS}. 
We first consider an arbitrarily fixed realization of $\big\{ \Gamma^{(r)}_k, \Pi^{(r)}_k : 0 \leq r \leq t, 1\leq k\leq K_r \big\}$ and $\mathrm{BAD}_{t-1}$. Since the events $v \in \mathrm{BIAS}_{D,t,s,k}$ over $v\not\in \mathrm{BAD}_{t-1}$ are independent of each other and by Lemma~\ref{lemma-Bernstein-inequality} each such event occurs with probability at most $2 \exp \Big\{ - \frac{ n \vartheta \Hat{q} (1-\Hat{q}) e^{ - 20 (\log \log n)^{10} } }{ 2( |\mathrm{BAD}_{t-1}| \Hat{q}(1-\Hat{q}) + \sqrt{n \Hat{q}(1-\Hat{q})} ) } \Big\}$, we can then  apply Lemma~\ref{lemma-Bernstein-inequality} (again) and derive that 
\begin{align*}
    &\Hat{\mathbb P}\Big( |\mathrm{BIAS}_{D,t,s,k}| > 4\vartheta^{-1} K_t^2 ( |\mathrm{BAD}_{t-1}| + \sqrt{n/\Hat{q}(1-\Hat{q})}) e^{ 20 (\log \log n)^{10} } \Big) \leq e^{-n K^2_t} \,.
\end{align*}
Clearly, a similar estimate holds for $\mathrm{BIAS}_{\mathsf D,t,s,k}$.
We next apply a union bound over all admissible realizations for $\{ \Gamma^{(r)}_k, \Pi^{(r)}_k : r \leq t, 1\leq k\leq K_r \}, \mathrm{BAD}_{t-1}$. Since we need to choose at most $4K_t$ subsets of $V$ (or $\mathsf{V}$), the enumeration is bounded by $2^{4K_t n}$. Therefore, applying a union bound over all these realizations and over $s, t$, we obtain the desired estimate by recalling that $\mathrm{BIAS}_t = \cup_{0 \leq s \leq t} \cup_{1 \leq k \leq K_s} \mathrm{BIAS}_{D,t,s,k} \cup \mathrm{BIAS}_{\mathsf{D},t,s,k}$.
\end{proof} 

\begin{Lemma} {\label{lemma-bound-cardinality-set-PRB}}
We have $\mathbb P(|\mathrm{PRB}_t| \leq \mathtt A ) \geq 1-o(e^{-n K_t})$ where
\begin{align*}
    \mathtt A =  K_t^{20} \vartheta^{-2} \Delta_{t}^{-2} \big( |\mathrm{BAD}_{t-1}| + \sqrt{n /\Hat{q}(1-\Hat{q})} \big) \,.
\end{align*}
\end{Lemma}
\begin{proof}
Recall \eqref{equ-def-set-PRB}.
For each fixed admissible realization of $\{ \Gamma^{(r)}_k, \Pi^{(r)}_k : 0 \leq r \leq t, 1 \leq k \leq K_r \}$ and a realization of $\mathrm{BAD}_{t-1}$, we have that the matrices $\mathbf{Q}_{t-1}$ and $H_{t,k,v}, \mathsf{H}_{t,k,v}$ (and thus the matrix $\mathbf{H}_t$) are fixed. Define a vector $\chi_{t, k}$ such that $\chi_{t,k}(k,v)= \mathbf{1}_{v \in \mathrm{PRB}_{t,k}}$ and $\chi_{t,k}(l,v)=0$ for $l \neq k$ for each $v \not \in \mathrm{BAD}_{t-1}$. Then, we have
\begin{align}\label{eq-PRB-eta-relation}
    \big| \begin{pmatrix}
        [g Y]_{t-1} - g_{t-1} [Y]_{t-1} & 0
    \end{pmatrix} \mathbf{Q}_{t-1} \mathbf{H}_{t}^{*} \chi^{*}_{t,k} \big| > \tfrac{1}{2} \Delta_{t} \| \chi_{t,k} \|^2  = \tfrac{1}{2} \Delta_{t}|\mathrm{PRB}_t| \,,
\end{align}
or we have a version of \eqref{eq-PRB-eta-relation} with $Y$ replaced by $\mathsf Y$. Without loss of generality, we in what follows assume that \eqref{eq-PRB-eta-relation} holds. 
For $v \neq u$, we have that
\begin{align*}
    & \mathbb{E}[ ([g Y]_{t-1}(k,v) - g_{t-1} [Y]_{t-1}(k,v)) ([g Y]_{t-1}(k,u) - g_{t-1} [Y]_{t-1}(k,u)) ] = 0 \,, \\
    & \mathbb{E}[ ([g Y]_{t-1}(k,v) - g_{t-1} [Y]_{t-1}(k,v)) ([g Y]_{t-1}(l,v) - g_{t-1} [Y]_{t-1}(l,v)) ] \leq K_{t-1}^2 |\mathrm{BAD}_{t-1}|/n  \,.
\end{align*}
So the covariance matrix of $\begin{pmatrix} [g Y]_{t-1} - g_{t-1} [Y]_{t-1} & 0 \end{pmatrix}$ (we denote as $\mathbf{R}_{t-1}$) is a block-diagonal matrix with each block of dimension at most $K_{t-1}$ and of entries bounded by $\frac{K_{t-1}^2|\mathrm{BAD}_{t-1}|}{n}$. As a result, $\|\mathbf{R}_{t-1}\|_{\mathrm{op}} \leq K_{t-1}^4 |\mathrm{BAD}_{t-1}|/n$. Thus, regarding $\chi_{t,k}$ as a deterministic vector, we have $\begin{pmatrix} [g Y]_{t-1} - g_{t-1} [Y]_{t-1} & 0 \end{pmatrix} \mathbf{Q}_{t-1} \mathbf{H}_{t}^{*} \chi^{*}_{t,k}$ is a linear combination of $\overrightarrow{G}_{u,w}-\Hat{q}$, with variance given by 
\begin{align*}
    \chi_{t,k} \mathbf{H}_{t} \mathbf{Q}_{t-1} \mathbf{R}_{t-1} \mathbf{Q}_{t-1} \mathbf{H}_{t}^{*} \chi_{t,k}^{*} \leq \| \chi_{t,k} \|^2 \| \mathbf{R}_{t-1} \|_{\mathrm{op}} \| \mathbf{Q}_{t-1} \|^2_{\mathrm{op}} \| \mathbf{H}_{t} \|^2_{\mathrm{op}} \leq \frac{1}{n} K_t^{10} |\mathrm{BAD}_{t-1}| \| \chi_{t,k} \|^2 \,.
\end{align*}
In addition, the coefficient of each $\overrightarrow{G}_{u,w}- \Hat{q}$ can be bounded as follows: for $0 \leq s \leq t-1, 1 \leq k \leq \frac{K_s}{12}, v \not \in \mathrm{BAD}_{t-1}$ denoting $\tau_{u,w}(s,k,v)$ the coefficient of $\overrightarrow{G}_{u,w}-\Hat{q}$ in $[gY]_{t-1} - g_{t-1}[Y]_{t-1}$, we have $\tau_{u,w}(s,k,v) = 0$ for $v \not\in \{u,w\}$ and we have $|\tau_{u,w}(s,k,u)|, |\tau_{u,w}(s,k,w)| = O( \frac{K_s}{\sqrt{ \mathfrak{a}_s n \Hat{q} }} )$. Combined with Lemmas~\ref{lemma-1-norm-Q} and \ref{lemma-op-norm-H}, it yields that  the coefficient of $\overrightarrow{G}_{u,w}$ in the linear combination satisfies that
\begin{align*}
    | \tau_{u,w} \mathbf{Q}_{t-1} \mathbf{H}_{t}^{*} \chi_{t,k}^{*} | \leq \| \tau_{u,w} \|_{1} \| \mathbf{Q}_{t-1} \mathbf{H}_{t}^{*} \chi_{t,k}^{*} \|_{\infty} \leq \| \tau_{u,w} \|_{1} \| \mathbf{Q}_{t-1} \mathbf{H}_{t}^{*} \|_{\infty}  \| \chi_{t,k}^{*} \|_{\infty} \leq \frac{ K_t^{13} }{ \vartheta^{2} \sqrt{ n \Hat{q}(1-\Hat{q}) } }\,.
\end{align*}
Thus, recalling \eqref{eq-PRB-eta-relation}, we can apply Lemma~\ref{lemma-Bernstein-inequality} to each realization of $\chi_{t, k}$ and derive that (noting that on $|\mathrm{PRB}_t| \geq \mathtt A $ we have $\|\chi_{t, k}\|^2 \geq \mathtt A$)
\begin{align*}
    \Hat{\mathbb P}(|\mathrm{PRB}_t| \geq \mathtt A ) 
    \leq 2^{K_t n}\sum_{\mathtt A' \geq \mathtt A} \exp \Big\{ - \frac{ ( 0.5 \Delta_t \mathtt A')^2 }{  K_t^{10} |\mathrm{BAD}_{t-1}| \mathtt A'/n + \Delta_t \mathtt A' K_t^{13} / \vartheta^{2} \sqrt{n\Hat{q}(1-\Hat{q})}  } \Big\} \,,
\end{align*}
which is bounded by $\exp \{ - n K_t^2 \}$.
Here in the above display the factor $2^{K_t n}$ counts the enumeration for possible realizations of $\chi_{t, k}$.
At this point, we apply a union bound over all possible realizations of $\{ \Gamma^{(r)}_k, \Pi^{(r)}_k : 0 \leq r \leq t, 1 \leq k \leq K_r \}$ and  $\mathrm{BAD}_{t-1}$ (whose enumeration is again bounded by $2^{4 K_t n}$), completing the proof of the lemma. 
\end{proof}

In the next few lemmas, we control $|\mathrm{LARGE}_t|$. 

\begin{Lemma} {\label{lemma-bound-cardinality-LARGE-(0)}}
With probability $1-o(e^{-n K_t})$ we have $|\mathrm{LARGE}^{(0)}_t| \leq 8 \vartheta^{-1} K^4_t n^{ 1 - \frac{2} {\log \log \log n}} $.
\end{Lemma}
\begin{proof}
Recall \eqref{equ-def-set-LARGE-(0)}. For each fixed realization of $\{ \Gamma^{(r)}_k, \Pi^{(r)}_k : r \leq t, 1\leq k\leq K_r \}$ and $\mathrm{BAD}_{t-1}$ and for each $j$ we can apply  Lemma~\ref{lemma-Bernstein-inequality} and obtain that 
\begin{align*}
    \Hat{\mathbb P} \Big(| \langle \eta^{(s)}_k, g_{t-1} D^{(s)}_v \rangle_{\langle j \rangle} | > n^{ \frac{1} {\log \log \log n}} \Big) \leq 2 \exp \{ - \tfrac{1}{2} \vartheta n^{ \frac{2} {\log \log \log n}} \}\,.
\end{align*}
Also $\mathbb{P}( |W^{(s)}_k(v)| > n^{\frac{1}{\log \log \log n}} ) \leq \exp \{ - \frac{1}{2} n^{\frac{2}{\log \log \log n }} \}$. Then applying a union bound over $j$, we get that
\begin{align*}
    \Hat{\mathbb P} \big( v \in \mathrm{LARGE}^{(0)}_{t,s,k} \big) \leq 2 n^2 \exp \big\{ - \tfrac{1}{2} \vartheta n^{ \frac{2} {\log \log \log n}} \big\} \leq \exp \big\{ - \tfrac{1}{4} \vartheta n^{ \frac{2} {\log \log \log n}} \big\}
\end{align*}
where in the last inequality we use the bound of $\vartheta=\vartheta_{\chi+1}$ in Lemma~\ref{lemma-property-vartheta-varsigma}.
Under the $\Hat{\mathbb P}$-measure, the events $v\in \mathrm{LARGE}^{(0)}_{t,s,k}$ over $v$ are independent of each other. Thus, another application of Lemma~\ref{lemma-Bernstein-inequality} yields that 
\begin{align*}
    \Hat{\mathbb P}\Big( |\mathrm{LARGE}_{t,s,k}| > \vartheta^{-1} K^2_t n^{ 1 - \frac{1} {\log \log \log n}} \Big) \leq \exp \big\{ - \tfrac{1}{4} \vartheta n^{ \frac{2}{\log \log \log n}} \cdot \vartheta^{-1} K^2_t n^{ 1- \frac{2} {\log \log \log n}} \big\}\,,
\end{align*}
which is bounded by $e^{- K_t^2 n /4 }$. Now, a union bound over all possible realizations (which is at most $2^{4 K_t n}$) and over $s, k$ (as well as for the mathsf version) completes the proof. 
\end{proof}

\begin{Lemma} {\label{lemma-bound-cardinality-LARGE-(a)}}
With probability $1-o(e^{-n K_t})$ we have 
\begin{align}
    |\mathrm{LARGE}^{(1)}_t|& \leq \vartheta^{-1} K_t^4 n^{ -\frac{2}{\log \log \log n}} (| \mathrm{BIAS}_t | + |\mathrm{PRB}_t| + | \mathrm{LARGE}^{(0)}_t |)\,,  \label{eq-LARGE-1}\\
    \mbox{ and } |\mathrm{LARGE}^{(a+1)}_t| &\leq \vartheta^{-1} K_t^4 n^{- \frac{2}{\log \log \log n}} | \mathrm{LARGE}^{(a)}_t |\mbox{ for } a \geq 1\,. \label{eq-LARGE-a}
\end{align}
\end{Lemma}
\begin{proof}
We will prove  \eqref{eq-LARGE-1} and the proof for \eqref{eq-LARGE-a} is similar. Recall \eqref{equ-def-set-LARGE}. For each fixed realization $\{ \Gamma^{(r)}_k, \Pi^{(r)}_k : r \leq t, 1\leq k\leq K_r \}$, $\mathrm{BAD}_{t-1}, \mathrm{LARGE}^{(0)}_t, \mathrm{BIAS}_t$, we apply Lemma~\ref{lemma-Bernstein-inequality} and obtain that
\begin{align*}
    \Hat{\mathbb P}(v \in \mathrm{LARGE}^{(1)}_{t,s,k}) \leq n^2 \exp \Big\{ - \frac{ \vartheta n^{ \frac{2}{\log \log \log n} } }{ (| \mathrm{BIAS}_t | + |\mathrm{PRB}_t| + | \mathrm{LARGE}^{(0)}_t |)/n } \Big\} \,.
\end{align*}
Since under the $\Hat{\mathbb P}$-measure we have  independence among $\{ v \in \mathrm{LARGE}^{(1)}_t \}$ for different $v$, we can then apply Lemma~\ref{lemma-Bernstein-inequality} again and get that 
\begin{align*}
   & \Hat{\mathbb P}(\mbox{the complement of } \eqref{eq-LARGE-1}) \\
   \leq \ & \exp \Big\{ - \frac{ \vartheta n^{ \frac{2}{\log \log \log n} } \cdot \vartheta^{-1} K_t^2 n^{ -\frac{2}{\log \log \log n}} (| \mathrm{BIAS}_t | + |\mathrm{PRB}_t| + | \mathrm{LARGE}^{(0)}_t |)  }{ (| \mathrm{BIAS}_t | + |\mathrm{PRB}_t| + | \mathrm{LARGE}^{(0)}_t |)/n }  \Big\} \leq  e^{-K_t^2 n }\,.
\end{align*}
Then a union bound over all possible realizations (whose enumeration is bounded by $2^{4 K_t n}$) completes the proof.
\end{proof}

We may assume that all the typical events as described in Lemmas~\ref{lemma-bound-cardinality-BIAS}, \ref{lemma-bound-cardinality-set-PRB}, \ref{lemma-bound-cardinality-LARGE-(0)} and \ref{lemma-bound-cardinality-LARGE-(a)} hold (note that this occurs with probability $1- K_t^2e^{-n K_t}$). Then, we see that $\mathrm{LARGE}_t^{(\log n)} = \emptyset$ (by \eqref{eq-LARGE-a}). In addition, we have that
\begin{align}\label{eq-bound-cardinality-LARGE}
    & |\mathrm{LARGE}_t| \leq |\mathrm{LARGE}^{(0)}_t| + \sum_{a=1}^{\log n} |\mathrm{LARGE}_{t}^{(a)}| \nonumber \\
    \leq \ & 8 \vartheta^{-1} K_t^4 \Big( n^{ 1- \frac{2}{\log \log \log n} } + ( |\mathrm{BIAS}_t| + |\mathrm{PRB}_t| + n^{1-\frac{2}{\log \log \log n}} ) \sum_{a=1}^{\log n} (\vartheta^{-1}K_t^4)^a n^{- \frac{a}{\log \log \log n} } \Big) \nonumber \\
    \leq \ & 20 \vartheta^{-2} K_t^8 n^{1-\frac{2}{\log \log \log n}} \,.
\end{align}
This (together with events in Lemmas~\ref{lemma-bound-cardinality-BIAS} and \ref{lemma-bound-cardinality-set-PRB}) implies that
\begin{align*}
    |\mathrm{BAD}_t| \leq 20 \vartheta^{-3} K_t^{30} e^{20 (\log \log n)^{10}} \big( |\mathrm{BAD}_{t-1}| + n^{1-\frac{2}{\log \log \log n}} \big) \,.
\end{align*}
Combined with the induction hypothesis $\mathcal{T}_{t-1}$, this yields that
\begin{align}
    \mathbb{P}( \mathcal{T}_t^{c}; \mathcal{T}_{t-1} ) \leq K_t^2 \exp\{ - n K_t \} \,.
    \label{equ-prob-mathcal-T-happens}
\end{align}

\subsubsection{Step 2: $\mathcal{B}_{t}$}

Before controlling $\mathcal B_t$, we prove a couple of lemmas as preparations. 
\begin{Lemma} {\label{lemma-HS-norm-product}}
For any two matrices $\mathrm{A,B}$, we have $\| \mathrm{AB} \|_{\mathrm{HS}}, \| \mathrm{BA} \|_{\mathrm{HS}} \leq \| \mathrm{A} \|_{\mathrm{op}} \| \mathrm{B} \|_{\mathrm{HS}}$.
\end{Lemma}
\begin{proof}
Since $\| \mathrm{A^{*}A} \|_{\mathrm{op}} = \| \mathrm{A} \|_{\mathrm{op}}^2$, we have $\| \mathrm{A} \|_{\mathrm{op}}^2 \mathrm{I} - \mathrm{A^{*}A}$ is semi-positive definite. Thus,
\begin{align*}
    & \| \mathrm{AB} \|_{\mathrm{HS}}^2 = \mathrm{tr} (\mathrm{ B^{*}A^{*} AB }) = \| \mathrm{A} \|_{\mathrm{op}}^2 \mathrm{tr} (\mathrm{ B^{*} B }) - \mathrm{tr} (\mathrm{B^{*} (\| \mathrm{A} \|_{\mathrm{op}}^2 \mathrm{I} - \mathrm{A^{*}A}) B }) \\
    \leq \ & \| \mathrm{A} \|_{\mathrm{op}}^2 \mathrm{tr} (\mathrm{ B^{*}B }) = \| \mathrm{A} \|_{\mathrm{op}}^2 \| \mathrm{B} \|_{\mathrm{HS}}^2 \,,
\end{align*}
which yields $\| \mathrm{AB} \|_{\mathrm{HS}} \leq \| \mathrm{A} \|_{\mathrm{op}} \| \mathrm{B} \|_{\mathrm{HS}}$. Similarly we can show $\| \mathrm{BA} \|_{\mathrm{HS}} \leq \| \mathrm{A} \|_{\mathrm{op}} \| \mathrm{B} \|_{\mathrm{HS}}$.
\end{proof}

\begin{Lemma}
\label{lemma-Gaussian-comparison}
    For any $m\geq 1$, let $\mu\in \mathbb R^m$ and let $\Sigma_X, \Sigma_Y$ be $m\!*\!m$ positive definite matrices. Suppose that $X \sim \mathcal{N}(0,\Sigma_{X})$ and $Y \sim \mathcal{N}(\mu,\Sigma_{Y})$. Then for all $u\in \mathbb R^m$
    \begin{align*}
        \frac{  p_Y (u)  }{  p_X (u)  } \leq \exp \Big\{ & \| \Sigma_{X} \|_{\mathrm{op}}^2 \| \Sigma_{Y}^{-1} \|_{\mathrm{op}}^2 \| \Sigma_{Y} \|_{\mathrm{op}}^2 \| \Sigma_{Y}^{-1} - \Sigma_{X}^{-1} \|_{\mathrm{HS}}^2 + ( \| \Sigma_{X}^{-1} \|_{\mathrm{op}} + \| \Sigma_{Y}^{-1} \|_{\mathrm{op}} ) \| \mu \|^2 \\
        &+ \langle \mu, u \rangle_{\Sigma_{Y}^{-1}}  + \frac{1}{2} \| u \|^2_{ (\Sigma_{X}^{-1} - \Sigma_{Y}^{-1}) } - \frac{1}{2} \mathbb{E} \big[ \| Y \|^2_{ (\Sigma_{X}^{-1} - \Sigma_{Y}^{-1}) } \big] \Big\} \,.
    \end{align*}
\end{Lemma}
\begin{proof}
Recalling the formula for Gaussian density, we have that
\begin{align}\label{eq-density-ratio-expression}
     \frac{  p_{Y} (u)  }{  p_{X} (u)  } 
    &= \sqrt{\frac{\textup{det}(\Sigma_{X})} {\textup{det}(\Sigma_{Y})}} \cdot \exp \Big\{  - \frac{1}{2} \| \mu \|_{\Sigma_{Y}^{-1}}^{2} + \langle u,\mu \rangle_{\Sigma^{-1}_{Y}}  + \frac{1}{2} \| u \|^2_{ (\Sigma_{X}^{-1} - \Sigma_{Y}^{-1}) } \Big\}\,.
\end{align}
Let $\Lambda_{Y}$ be a positive definite matrix such that  $\Sigma_{Y} = \Lambda_{Y} \Lambda_{Y}^{*}$. We have 
\begin{align}\label{eq-Y-Sigma-norm}
\mathbb{E} \big[\| Y \|^2_{(\Sigma_{X}^{-1} - \Sigma_{Y}^{-1})} \big] &= \| \mu \|^2_{(\Sigma_X^{-1} - \Sigma_Y^{-1})} + \mathrm{tr} \big( \Lambda_{Y}^{*} ( \Sigma_{X}^{-1} - \Sigma_{Y}^{-1} ) \Lambda_{Y} \big) \nonumber \\
& = \| \mu \|^2_{(\Sigma_X^{-1} - \Sigma_Y^{-1})} + \mathrm{tr} \big( \Lambda_{Y}^{*}  \Sigma_{X}^{-1} \Lambda_{Y} - \mathrm{I} \big)\,,
\end{align}
where $\mathrm{I}$ is an identity matrix.
We next control the determinants of $\Sigma_X, \Sigma_Y$. Let $\varrho_{1}, \ldots, \varrho_{m} \geq 0$ be the eigenvalues of $ \Lambda_{Y}^{*} \Sigma_{X}^{-1} \Lambda_{Y} $. Then $\sum_{k=1}^{m} (\varrho_{k}-1) = \mathrm{tr}( \Lambda_{Y}^{*} \Sigma_{X}^{-1} \Lambda_{Y} - \mathrm{I} ) $ and $\prod_{k=1}^{m} \varrho_k = \mathrm{det} (\Lambda_{Y}^{*} \Sigma_{X}^{-1} \Lambda_{Y}) = \frac{\mathrm{det}(\Sigma_Y)}{\mathrm{det}(\Sigma_{X})}$. Also, using $ \varrho_k^{-1} \leq \| (\Lambda_{Y}^{*} \Sigma_{X}^{-1} \Lambda_{Y})^{-1} \|_{\mathrm{op}} \leq \| \Sigma_{X} \|_{\mathrm{op}} \| \Sigma_{Y}^{-1} \|_{\mathrm{op}}$ and the fact that  $x-1-\log x \leq c^{-2} (x-1)^2$ for $x \geq c > 0$, we have that
\begin{align*}
    & \log \Big\{ \frac{\textup{det}(\Sigma_{X})} {\textup{det}(\Sigma_{Y})} \Big\} + \mathrm{tr}( \Lambda_{Y}^{*} \Sigma_{X}^{-1} \Lambda_{Y} - \mathrm{I} ) = \sum_{k=1}^{m} ( - \log \varrho_{k} + \varrho_{k}-1 ) \\
    \leq \ & \| \Sigma_{X} \|_{\mathrm{op}}^{2} \| \Sigma_{Y}^{-1} \|_{\mathrm{op}}^{2} \sum_{k=1}^{m} (\varrho_k-1)^2 = \| \Sigma_{X} \|_{\mathrm{op}}^{2} \| \Sigma_{Y}^{-1} \|_{\mathrm{op}}^{2} \| \mathrm{I} - \Lambda_{Y}^{*} \Sigma_{X}^{-1} \Lambda_{Y}  \|_{\mathrm{HS}}^2  \\
    = \ & \| \Sigma_{X} \|_{\mathrm{op}}^{2} \| \Sigma_{Y}^{-1} \|_{\mathrm{op}}^{2} \| \Lambda_{Y}^{*} (\Sigma_{Y}^{-1} - \Sigma_{X}^{-1}) \Lambda_{Y} \|_{\mathrm{HS}}^2 
    \leq \| \Sigma_{X} \|_{\mathrm{op}}^{2} \| \Sigma_{Y}^{-1} \|_{\mathrm{op}}^{2} \| \Lambda_Y \|_{\mathrm{op}}^4 \| \Sigma_{Y}^{-1} - \Sigma_{X}^{-1} \|_{\mathrm{HS}}^2  \\
    = \ & \| \Sigma_{X} \|_{\mathrm{op}}^{2} \| \Sigma_{Y}^{-1} \|_{\mathrm{op}}^{2} \| \Sigma_{Y} \|_{\mathrm{op}}^2 \| \Sigma_{Y}^{-1} - \Sigma_{X}^{-1} \|_{\mathrm{HS}}^2 \,,
\end{align*}
where the last inequality follows from Lemma \ref{lemma-HS-norm-product}. Combined with \eqref{eq-density-ratio-expression} and \eqref{eq-Y-Sigma-norm}, this completes the proof of the lemma.
\end{proof}
We now return to $\mathcal{B}_{t}$. Recall \eqref{eq-Lemma-3.9} as in Lemma \ref{lemma-bound-conditional-unbiased-density}.
It remains to bound the density ratio between $\big\{ g_{t-1} \Tilde{Y}_t , g_{t-1} \Tilde{\mathsf{Y}}_t | \mathcal{F}_{t-1} \big\}$ and $\big\{ \Check{Y}_t, \Check{\mathsf{Y}}_t \big\}$. Recall that in Remark \ref{remark-conditional-Gaussian} we have shown
\begin{align*}
    (g_{t-1} \Tilde{Y}_t(k,v) | \mathcal{F}_{t-1}) \overset{d}{=} & g_{t-1}\Tilde{Y}_{t}^{\diamond}(k,v) - \mathrm{GAUS}(g_{t-1}\Tilde{Y}_{t}(k,v)) + \mathrm{PROJ}(g_{t-1}\Tilde{Y}_{t}(k,v))  \,.
\end{align*}
Let $\Tilde{\Sigma}_t $ be the covariance matrix of the process
\begin{align*}
    \Bigg\{ \begin{split}
    g_{t-1}\Tilde{Y}_{t}^{\diamond}(k,v) - \mathrm{GAUS}(g_{t-1} \Tilde{Y}_{t}(k,v)) \\ g_{t-1} \Tilde{\mathsf{Y}}_{t}^{\diamond}(k,\mathsf{v}) - \mathrm{GAUS}(g_{t-1} \Tilde{\mathsf{Y}}_{t}(k,\mathsf{v})) 
    \end{split} : v, \pi^{-1}(\mathsf{v}) \not \in \mathrm{B}_{t}, 1 \leq k \leq \frac{K_t}{12} \Bigg\} \,,
\end{align*}
let $\Check{\Sigma}_t$ be the covariance matrix of
\begin{equation}\label{eq-def-Check-mathfrak-F-t}
    \Check{\mathfrak F}_t  = \big\{ \Check{Y}_t (k,v) , \Check{\mathsf{Y}}_t (k,\mathsf{v}) : v, \pi^{-1}(\mathsf{v}) \not \in \mathrm{B}_{t}, 1\leq k\leq \tfrac{K_t}{12} \big\} \,,
\end{equation}
and let $\Sigma^{\diamond}_{t}$ be the covariance matrix of 
\begin{align*}
    \big\{ g_{t-1} \Tilde{Y}^{\diamond}_{t}(k,v), g_{t-1} \Tilde{\mathsf{Y}}^{\diamond}_{t}(k,\mathsf{v}) : v, \pi^{-1}(\mathsf{v}) \not \in \mathrm{B}_{t}, 1 \leq k \leq \tfrac{K_t}{12} \big\} \,.
\end{align*}
Also define vectors $L^{(t)},\mathsf{L}^{(t)}$ such that for $1\leq k\leq \frac{K_t}{12}, v\not\in \mathrm{B}_t, \mathsf v\not\in \pi(\mathrm{B}_t)$ 
\begin{align*}
    L^{(t)} (k,v) = \mathrm{PROJ} ( g_{t-1} \Tilde{Y}_{t}(k,v) ) \mbox{ and }\mathsf{L}^{(t)} (k,\mathsf{v}) = \mathrm{PROJ}( g_{t-1} \Tilde{\mathsf{Y}}_{t}(k,\mathsf{v}) ) \,.
\end{align*}
Applying Lemma \ref{lemma-Gaussian-comparison} we have
\begin{align}
    & \frac{ p_{ \{ g_{t-1} \Tilde{Y}_t , g_{t-1} \Tilde{\mathsf{Y}}_t | \mathcal{F}_{t-1} \} } (x_{t}, \mathsf{x}_{t}) }{ p_{ \{ \Check{Y}_t ,   \Check{\mathsf{Y}}_t \} } (x_{t}, \mathsf{x}_{t}) } \nonumber \\
    \leq \ & \exp \Big\{ \| \Check{\Sigma}_t \|_{\mathrm{op}}^2 \| \Tilde{\Sigma}_{t}^{-1} \|_{\mathrm{op}}^2 \| \Tilde{\Sigma}_{t} \|_{\mathrm{op}}^2 \| \Tilde{\Sigma}_{t}^{-1} - \Check{\Sigma}_{t}^{-1} \|_{\mathrm{HS}}^2 + ( \| \Check{\Sigma}_{t}^{-1} \|_{\mathrm{op}} + \| \Tilde{\Sigma}_{t}^{-1} \|_{\mathrm{op}} ) \| ( L^{(t)} , \mathsf{L}^{(t)} ) \|^2 \nonumber \\
    & + ( L^{(t)} , \mathsf{L}^{(t)} ) \Tilde{\Sigma}_{t}^{-1} ( x_{t},\mathsf{x}_{t} )^{*} 
    + \frac{1}{2} \| ( x_{t} , \mathsf{x}_{t} ) \|^2_{ (\Tilde{\Sigma}_{t}^{-1} - \Check{\Sigma}_{t}^{-1}) } - \frac{1}{2} \mathbb{E} \big[ \| ( X_{t} , \mathsf{X}_{t} ) \|^2_{ (\Tilde{\Sigma}_{t}^{-1} - \Check{\Sigma}_{t}^{-1}) } \big]  \Big\} \,,
    \label{equ-density-ratio-II}
\end{align}
where the expectation is taken over $(X_{t}, \mathsf{X}_{t}) \sim ( \Tilde{W}_t + g_{t-1} \Tilde{Y}_t , \Tilde{\mathsf{W}}_t + g_{t-1} \Tilde{\mathsf{Y}}_t | \mathcal{F}_{t-1} )$. We need a few estimates on $\Tilde{\Sigma}_{t},\Check{\Sigma}_{t}$ and $L^{(t)}, \mathsf{L}^{(t)}$.

\begin{Claim} \label{claim-2-norm-L}
On the event $\mathcal H_t$, we have $\| L^{(t)} \|^2, \| \mathsf{L}^{(t)} \|^2 \leq n K_{t}^6 \Delta_{t}^2$.
\end{Claim}
\begin{proof}
On $\mathcal{H}_{t}$ we know $\sum_{k,v} \big( \mathrm{PROJ} ( g_{t-1}\Tilde{Y}_t(k,v) ) \big)^2 \leq n K_{t}^6 \Delta_{t}^2$. Thus, $\| L^{(t)} \|^2 \leq n K_{t}^6 \Delta_{t}^2$. We can bound $\| \mathsf{L}^{(t)} \|^2$ similarly.
\end{proof}

\begin{Claim} \label{claim-op-norm-W}
We have $\| \Tilde{\Sigma}_{t}^{-1} \|_{\mathrm{op}}, \| (\Sigma^{\diamond}_{t})^{-1} \|_{\mathrm{op}}, \| \Check{\Sigma}_{t}^{-1} \|_{\mathrm{op}} \leq 1$.
\end{Claim}
\begin{proof}
By definition, we see that $\Tilde{\Sigma}_t$ is the sum of the identity matrix and a positive definite matrix, and thus we have $\| \Tilde{\Sigma}_{t}^{-1} \|_{\mathrm{op}} \leq 1$. Similar results hold for $\Sigma^{\diamond}_{t}$ and $\Check{\Sigma}_{t}$.
\end{proof}

\begin{Claim} \label{claim-op-HS-norm-W-minus-V}
On  $\mathcal T_{t-1}\cap \mathcal E_t$, we have $\| \Tilde{\Sigma}_{t} - \Check{\Sigma}_{t} \|_{\mathrm{op}} \leq 200K_{t}^{6}$ and $\| \Tilde{\Sigma}_{t} - \Check{\Sigma}_{t} \|_{\mathrm{HS}}^2 \leq n K_{t}^{11} \Delta_{t}^2$.
\end{Claim}
\begin{proof}
By \cite[(3.12),(3.61)]{DL22+} which follow from standard properties for general Gaussian processes, we have
\begin{align*}
    & \mathbb{E} \Big[ ( g_{t-1} \Tilde{Y}_{t}^{\diamond}(k,v) - \mathrm{GAUS}( g_{t-1} \Tilde{Y}_{t}(k,v))) ( g_{t-1} \Tilde{Y}_{t}^{\diamond}(l,u) - \mathrm{GAUS}(g_{t-1} \Tilde{Y}_{t}(l,u))) \Big] \\
    = & \mathbb{E} \Big[ g_{t-1} \Tilde{Y}_{t}^{\diamond}(k,v) g_{t-1} \Tilde{Y}_{t}^{\diamond}(l,u) \Big] - \mathbb{E} \Big[ \mathrm{GAUS}( g_{t-1} \Tilde{Y}_{t}(k,v)) \mathrm{GAUS}(g_{t-1} \Tilde{Y}_{t}(l,u)) \Big] \,.
\end{align*}
Here the coefficient $g_{t-1}$ does not matter since this result follows from the fact that
\begin{align*}
    &\mathrm{Cov}\Big( Y_1 \mid \{ X_1,\ldots,X_n \} ,Y_2 \mid \{ X_1,\ldots,X_n \} \Big)  \\
    = \ & \mathbb{E}[ Y_1 Y_2 ] - \mathbb{E} \big[ \mathbb{E}[Y_1 \mid X_1,\ldots,X_n ] \mathbb{E}[ Y_2 \mid X_1,\ldots,X_n ] \big]
\end{align*}
for general Gaussian process $\{ X_1,\ldots,X_n,Y_1,Y_2 \}$.
Recalling \eqref{equ-Gaussian-part}, we see that the covariance matrix of $\{ \mathrm{GAUS}(g_{t-1} \Tilde{Y}_{t}), \mathrm{GAUS}( g_{t-1} \Tilde{\mathsf{Y}}_{t}) \}$ equals to
\begin{align*}
    \mathbb{E} \Big[ \mathbf{H}_{t} \mathbf{Q}_{t-1}
    \begin{pmatrix}
        g_{t-1} [\Tilde{Y}]_{t-1}^{\diamond} \\
        g_{t-1} [\Tilde{\mathsf{Y}}]_{t-1}^{\diamond}  
    \end{pmatrix}
    \begin{pmatrix}
        g_{t-1} [\Tilde{Y}]_{t-1}^{\diamond} &
        g_{t-1} [\Tilde{\mathsf{Y}}]_{t-1}^{\diamond}  
    \end{pmatrix}
    \mathbf{Q}_{t-1} \mathbf{H}_{t}^{*} \Big] 
    =  \mathbf{H}_{t} \mathbf{Q}_{t-1} \mathbf{H}_{t}^{*} \,.
\end{align*}
Thus, we have $\Tilde{\Sigma}_{t} = \Sigma^{\diamond}_{t} -  \mathbf{H}_{t} \mathbf{Q}_{t-1} (\mathbf{H}_{t})^{*}$. Combined with Lemmas~\ref{lemma-op-norm-Q} and \ref{lemma-op-norm-H}, it yields that $\| \Tilde{\Sigma}_{t} - \Sigma^{\diamond}_{t} \|_{\mathrm{op}} \leq 100K_{t}^6$.  In addition, applying Lemmas~\ref{lemma-op-norm-Q}, \ref{lemma-op-norm-H} and \ref{lemma-HS-norm-product}, we get
\begin{align}\label{eq-Tilde-diamond-HS-norm}
    & \| \Tilde{\Sigma}_{t} - \Sigma^{\diamond}_{t} \|_{\mathrm{HS}}^2 = \| \mathbf{H}_{t} \mathbf{Q}_{t-1} \mathbf{H}_{t}^{*} \|_{\mathrm{HS}}^2 \leq \| \mathbf{Q}_{t-1} \|_{\mathrm{op}}^2  \| \mathbf{H}_{t} \|_{\mathrm{op}}^2 \| \mathbf{H}_{t} \|_{\mathrm{HS}}^2 \leq 10^5 n K_{t}^{10} \Delta_{t}^2 \,.
\end{align}
Furthermore, by \eqref{equ-degree-correlation-2} and \eqref{equ-degree-variance-1} we get that $\Sigma^{\diamond}_{t}((k,v),(l,v)) - \Check{\Sigma}_{t}((k,v),(l,v)) = O(K_{t} \Delta_{t})$; by \eqref{equ-degree-correlation-3} and \eqref{equ-degree-variance-2} we get that  $\Sigma^{\diamond}_{t}((k,\mathsf{v}),(l,\mathsf{v})) - \Check{\Sigma}_{t}((k,\mathsf{v}),(l,\mathsf{v})) = O(K_{t} \Delta_{t})$; by \eqref{equ-degree-correlation-4} and \eqref{equ-degree-covariance} we get that $\Sigma^{\diamond}_{t}((k,v),(l,\pi(v))) - \Check{\Sigma}_{t}((k,v),(l,\pi(v))) = O(K_{t} \Delta_{t})$. Also, for $u \neq v$, by \eqref{equ-degree-correlation-1} we have for $\tau_u \in \{u, \pi(u)\}$ and $\tau_v \in \{v, \pi(v)\}$
\begin{align*}
     \Sigma^{\diamond}_{t}((k,\tau_v),(l, \tau_u)) - \Check{\Sigma}_{t}((k, \tau_v),(l, \tau_u))& = \Sigma^{\diamond}_{t}((k, \tau_v),(l, \tau_u))
    \\
    &= O \Big( \frac{K_t}{\mathfrak{a}_t n} (\mathbf{1}_{ v \in \cup_{i} \Gamma^{(t)}_i } - \mathfrak{a}_t ) (\mathbf{1}_{ u \in \cup_{i} \Gamma^{(t)}_i } - \mathfrak{a}_t ) \Big) \,.
\end{align*}
Combined with Items (i) and (iv) of $\mathcal{E}_t$ in Definition~\ref{def-admissible},
this implies that $\| \Check{\Sigma}_{t} - \Sigma^{\diamond}_{t} \|_{\mathrm{HS}}^2 \leq n K_{t}^{6} \Delta_{t}^2$. Applying Lemma~\ref{lemma-bound-on-op-norm} by setting $\mathcal{I}_v = \mathcal{J}_v = \{ (s,k,v),(s,k,\pi(v)) \}$, $\delta=K_t \Delta_t$ and $C = 10 K_t^3$, we can deduce that $\| \Check{\Sigma}_{t} - \Sigma^{\diamond}_{t} \|_{\mathrm{op}} \leq 100 K_{t}^3 $. Combined with \eqref{eq-Tilde-diamond-HS-norm} and the fact that $\| \Tilde{\Sigma}_{t} - \Sigma^{\diamond}_{t} \|_{\mathrm{op}} \leq 100K_{t}^6$, this completes the proof by the triangle inequality.
\end{proof}
\begin{Corollary} \label{corollary-2-HS-norm-W-inverse-minus-V-inverse}
On the event $\mathcal T_{t-1}\cap \mathcal E_t$, we have $\| \Tilde{\Sigma}_{t}^{-1} - \Check{\Sigma}_{t}^{-1} \|_{\mathrm{op}} \leq 200 K_{t}^{6}$ as well as $ \| \Tilde{\Sigma}_{t}^{-1} - \Check{\Sigma}_{t}^{-1} \|^2_{\mathrm{HS}} \leq n K_{t}^{11} \Delta_{t}^2$.
\end{Corollary}
\begin{proof}
Combining Claims \ref{claim-op-norm-W} and \ref{claim-op-HS-norm-W-minus-V}, we get
\begin{align*}
    \| \Tilde{\Sigma}_{t}^{-1} - \Check{\Sigma}_{t}^{-1} \|_{\mathrm{op}} & = \| \Tilde{\Sigma}_{t}^{-1} (  \Tilde{\Sigma}_{t} - \Check{\Sigma}_{t}  )  \Check{\Sigma}_{t}^{-1} \|_{\mathrm{op}} \\
    & \leq \| \Tilde{\Sigma}_{t}^{-1}  \|_{\mathrm{op}} \| \Tilde{\Sigma}_{t} - \Check{\Sigma}_{t} \|_{\mathrm{op}} \| \Check{\Sigma}_{t}^{-1} \|_{\mathrm{op}} \leq  200 K_{t}^{6} \,.
\end{align*}
It remains to control the HS-norm. By Lemma \ref{lemma-HS-norm-product} and  Claim \ref{claim-op-norm-W}, we have
\begin{align*}
    \| \Tilde{\Sigma}_{t}^{-1} - \Check{\Sigma}_{t}^{-1} \|^2_{\mathrm{HS}} & = \| \Tilde{\Sigma}_{t}^{-1} (  \Tilde{\Sigma}_{t} - \Check{\Sigma}_{t}  )  \Check{\Sigma}_{t}^{-1} \|^2_{\mathrm{HS}} \\
    & \leq \| \Tilde{\Sigma}_{t}^{-1} \|^2_{\mathrm{op}} \| \Check{\Sigma}_{t}^{-1} \|^2_{\mathrm{op}} \|  \Tilde{\Sigma}_{t} - \Check{\Sigma}_{t} \|^2_{\mathrm{HS}} \leq n K_{t}^{11} \Delta_{t}^2 \,. \qedhere
\end{align*}
\end{proof}
\begin{Corollary} {\label{corollary-op-norm-Sigma}}
On the event $\mathcal{T}_{t-1} \cap \mathcal{E}_t$ we have $\| \Check{\Sigma}_t \|_{\mathrm{op}} \leq 2$ and $\| \Tilde{\Sigma}_t \|_{\mathrm{op}} \leq 300 K_t^{6}$.
\end{Corollary}
\begin{proof}
By the definition of $\Check{\Sigma}_t$, we see that $\Check{\Sigma}_t = \mathrm{diag}(\Check{\Sigma}_{t,k,v})$ is a block-diagonal matrix where for $1 \leq k \leq K_t, v \not \in \mathrm{BAD}_t$ we have $\Check{\Sigma}_{t,k,v}$ is a $2\!*\!2$ matrix with diagonal entries 1 and non-diagonal entries 
$$ \Hat{\rho} \eta^{(t)}_k \Psi^{(t)} ( \eta^{(t)}_k )^{*} \overset{\eqref{equ-vector-unit}}{\leq} 2 \Hat{\rho} \varepsilon_t \overset{\eqref{eq-decrease-varepsilon}}{\leq} \Hat{\rho} \overset{\eqref{eq-assumetion-rho}}{\leq} 0.1 \,. $$
Thus, $\| \Check{\Sigma}_t \|_{\mathrm{op}} \leq 2$. By Claim \ref{claim-op-HS-norm-W-minus-V} and the triangle inequality, we get that $\| \Tilde{\Sigma}_t \|_{\mathrm{op}} \leq \| \Tilde{\Sigma}_t - \Check{\Sigma}_t \|_{\mathrm{op}} + \| \Check{\Sigma}_t \|_{\mathrm{op}} \leq 300 K_t^6$.
\end{proof}

We are now ready to show that $\mathcal{B}_{t}$ typically occurs. In what follows, we assume on the event $\mathcal{A}_{t-1}, \mathcal{B}_{t-1}, \mathcal{T}_{t}, \mathcal{E}_{t}, \mathcal{H}_{t}$ without further notice. Formally, we abuse the notation $\mathbb P(\cdot)$ by meaning $\mathbb P(\cdot \cap \mathcal{A}_{t-1} \cap \mathcal{B}_{t-1} \cap \mathcal{T}_{t} \cap \mathcal{E}_{t} \cap \mathcal{H}_{t})$; we abuse notation this way (and similarly in later subsections) since it shortens the notation and the meaning should be clear from the context. Define $\mathcal{C}$ be the event that the realizations of $\{ \mathrm{BAD}_t, \mathrm{BAD}_{t-1} \}, (g_{t-1} Y_t, g_{t-1} \mathsf{Y}_t )$ and $\{ \overrightarrow{G}_{u,w}, \overrightarrow{\mathsf{G}}_{\pi(u),\pi(w)} : u \mbox{ or } w \in \mathrm{BAD}_{t-1} \}$ are amenable. By Lemma~\ref{lem-good-set-good-variable-realization} we have 
\begin{equation}{\label{eq-bound-prob-mathcal-C}}
    \mathbb{P}(\mathcal{C}) \geq 1 - o( \exp \{ - n^{ \frac{1}{\log \log \log n}} \}) \,.
\end{equation}
By Lemma~\ref{lemma-bound-conditional-unbiased-density}, under $\mathcal{C}$ we have (recall that $\mathrm{BAD}_{t-1} = \mathrm B_{t-1}$ is implied in $\mathfrak S_{t-1}$)
\begin{align*}
    & \frac{ p_{ \{ g_{t-1} Y_t , g_{t-1} \mathsf{Y}_t | \mathfrak{S}_{t-1} , \mathrm{BAD}_t = \mathrm{B}_t\} } (x_{t}, \mathsf{x}_t) }{ p_{ \{ g_{t-1} \Tilde{Y}_t , g_{t-1} \Tilde{\mathsf{Y}}_t | \mathcal{F}_{t-1} \} } (x_{t}, \mathsf{x}_t)} \leq \exp \{ n \Delta_t^5  \} \,.
\end{align*}
Plugging Claims \ref{claim-2-norm-L} and \ref{claim-op-norm-W} and Corollaries \ref{corollary-2-HS-norm-W-inverse-minus-V-inverse} and \ref{corollary-op-norm-Sigma} into \eqref{equ-density-ratio-II}, we have under $\mathcal{C}$
\begin{align*}
    & \frac{ p_{\{ g_{t-1} \Tilde{Y}_t , g_{t-1} \Tilde{\mathsf{Y}}_t | \mathcal{F}_{t-1} \}}  (x_{t}, \mathsf{x}_t) }{ p_{\{ \Check{Y}_t ,  \Check{\mathsf{Y}}_t  \}} (x_{t}, \mathsf{x}_t) } \leq \exp \Big\{ n K_{t}^{29} \Delta_{t}^2 \\
    &+  ( L^{(t)} , \mathsf{L}^{(t)} ) \Tilde{\Sigma}_{t}^{-1} ( x_{t},\mathsf{x}_{t} )^{*} 
    + \frac{1}{2} \| ( x_{t} , \mathsf{x}_{t} ) \|^2_{ (\Tilde{\Sigma}_{t}^{-1} - \Check{\Sigma}_{t}^{-1}) } - \frac{1}{2} \mathbb{E} \big[ \| (X , \mathsf{X} ) \|^2_{ (\Tilde{\Sigma}_{t}^{-1} - \Check{\Sigma}_{t}^{-1}) } \big]  \Big\} \,,
\end{align*}
where $(X,\mathsf{X}) \sim (( g_{t-1} \Tilde{Y}_t, g_{t-1} \Tilde{\mathsf{Y}}_t ) | \mathcal{F}_{t-1})$. Altogether, we get that
\begin{align*}
    & \frac{ p_{\{ g_{t-1} Y_t , g_{t-1} \mathsf{Y}_t | \mathfrak{S}_{t-1} , \mathrm{BAD}_t = \mathrm{B}_t \}} (x_{t}, \mathsf{x}_t) }{ p_{\{ \Check{Y}_t , \Check{\mathsf{Y}}_t  \}} (x_{t}, \mathsf{x}_t) } \leq \exp \Big\{ 2 n K_{t}^{29} \Delta_{t}^2 \\
    &+  ( L^{(t)} , \mathsf{L}^{(t)} ) \Tilde{\Sigma}_{t}^{-1} ( x_{t},\mathsf{x}_{t} )^{*}
    + \frac{1}{2} \| ( x_{t} , \mathsf{x}_{t} ) \|^2_{ (\Tilde{\Sigma}_{t}^{-1} - \Check{\Sigma}_{t}^{-1}) } - \frac{1}{2} \mathbb{E} \big[ \| (X , \mathsf{X} ) \|^2_{ (\Tilde{\Sigma}_{t}^{-1} - \Check{\Sigma}_{t}^{-1}) } \big]  \Big\} \,.
\end{align*}
Thus, to estimate probability for $\mathcal B$ it suffices to show that the preceding upper bound is out of control only with probability $o(1)$. By  Lemma~\ref{lemma-bound-conditional-unbiased-density}, (as we will show) it suffices to control this probability under the measure $p_{\{ g_{t-1} \Tilde{Y}_t , g_{t-1} \Tilde{\mathsf{Y}}_t | \mathcal{F}_{t-1}  \}}$. To this end,  define 
\begin{align*}
    & \mathcal{U}^{(I)} = \{  (x,\mathsf{x}) : (L^{(t)},\mathsf{L}^{(t)}) \Tilde{\Sigma}_{t}^{-1} (x,\mathsf{x})^{*} \geq n K_{t}^{29} \Delta_{t}^{2} \} \,, \\
    & \mathcal{U}^{(II)} = \{ (x,\mathsf{x}) : \| (x,\mathsf{x}) \|_{(\Tilde{\Sigma}_{t}^{-1} - \Check{\Sigma}_{t}^{-1})}^2  - \mathbb{E} \big[ \| (X,\mathsf{X} ) \|^2_{ (\Tilde{\Sigma}_{t}^{-1} - \Check{\Sigma}_{t}^{-1}) } \big] \geq n K_{t}^{29} \Delta_{t}^{2} \} \,.
\end{align*}
By Claims~\ref{claim-2-norm-L} and \ref{claim-op-norm-W}, we have
\begin{align*}
    \mathrm{Var} \big( \big\langle (L^{(t)},\mathsf{L}^{(t)}) , (X,\mathsf{X}) \big\rangle_{\Tilde{\Sigma}_{t}^{-1}} \big) = (L^{(t)},\mathsf{L}^{(t)}) \Tilde{\Sigma}_{t}^{-1} (L^{(t)},\mathsf{L}^{(t)})^{*} \leq n K_{t}^{10} \Delta_{t}^2\,.   
\end{align*}
Since the mean is equal to the variance in this case and  $(X,\mathsf{X}) \sim (( g_{t-1} \Tilde{Y}_t, g_{t-1} \Tilde{\mathsf{Y}}_t ) | \mathcal{F}_{t-1})$, we then obtain from the tail probability of normal distribution that
\begin{align}
        \mathbb{P} ( ( g_{t-1} \Tilde{Y}_t, g_{t-1} \Tilde{\mathsf{Y}}_t ) \in \mathcal{U}^{(I)} | \mathcal{F}_{t-1} )\leq \exp \{  - n K_{t}^{29} \Delta_{t}^2 \} \,. \label{eq-bound-mathcal-U-(I)}
\end{align}
Next, we consider $ \mathcal{U}^{(II)} $. On the event $\mathcal{U}^{(II)}$, we have
\begin{align*}
    & \| (X,\mathsf{X}) \|^2_{ ( \Tilde{\Sigma}_{t}^{-1} - \Check{\Sigma}_{t}^{-1} ) } - \mathbb{E} \big[ \|(X,\mathsf{X}) \|^2_{ ( \Tilde{\Sigma}_{t}^{-1} - \Check{\Sigma}_{t}^{-1} ) } \big] > n K_{t}^{29} \Delta_{t}^2\,.
\end{align*}
Recalling the definitions of $L^{(t)},\mathsf{L}^{(t)}$, we have that $(X-L^{(t)},\mathsf{X}-\mathsf{L}^{(t)})$ is a mean zero Gaussian vector. This motivates us to write
\begin{align*}
    \| (X,\mathsf{X}) \|^2_{ ( \Tilde{\Sigma}_{t}^{-1} - \Check{\Sigma}_{t}^{-1} ) } - \mathbb{E} \big[ \|(X,\mathsf{X}) \|^2_{ ( \Tilde{\Sigma}_{t}^{-1} - \Check{\Sigma}_{t}^{-1} ) } \big] = \langle (L^{(t)},\mathsf{L}^{(t)}), (X-L^{(t)},\mathsf{X}-\mathsf{L}^{(t)}) \rangle_{ ( \Tilde{\Sigma}_{t}^{-1} - \Check{\Sigma}_{t}^{-1} ) } \\ 
    + \| (X-L^{(t)},\mathsf{X}-\mathsf{L}^{(t)}) \|^2_{ ( \Tilde{\Sigma}_{t}^{-1} - \Check{\Sigma}_{t}^{-1} ) } - \mathbb{E} \big[ \| (X-L^{(t)},\mathsf{X}-\mathsf{L}^{(t)}) \|^2_{ ( \Tilde{\Sigma}_{t}^{-1} - \Check{\Sigma}_{t}^{-1} ) } \big] \,.
\end{align*}
By Claim \ref{claim-2-norm-L} and Corollary \ref{corollary-2-HS-norm-W-inverse-minus-V-inverse}, $\langle (L^{(t)},\mathsf{L}^{(t)}), (X-L^{(t)},\mathsf{X}-\mathsf{L}^{(t)}) \rangle_{ ( \Tilde{\Sigma}_{t}^{-1} - \Check{\Sigma}_{t}^{-1} ) }$ is a Gaussian variable with mean 0 and variance $ O(n K_{t}^{12} \Delta_{t}^2))$. Then,
\begin{align}\label{eq-prob-U-II-reduction}
    \mathbb{P} ( (g_{t-1} \Tilde{Y}_t,  g_{t-1} \Tilde{\mathsf{Y}}_t) \in \mathcal{U}^{(II)} | \mathcal{F}_{t-1} ) \leq  ( \mathsf P_1 + \mathsf P_2 ) \,,
\end{align}
where $\mathsf P_1 =\mathbb{P} \big( \langle (L^{(t)},\mathsf{L}^{(t)}), (X-L^{(t)},\mathsf{X}-\mathsf{L}^{(t)}) \rangle_{ ( \Tilde{\Sigma}_{t}^{-1} - \Check{\Sigma}_{t}^{-1} ) } > n K_t^{28} \Delta_t^2 \big)  \leq \exp \{ - n K_t^{28} \Delta_t^2 \}$ and 
\begin{align*}
    \mathsf P_2 =  \mathbb P \big(
    \| (X-L^{(t)}, \mathsf{X} - \mathsf{L}^{(t)}) \|^2_{ ( \Tilde{\Sigma}_{t}^{-1} - \Check{\Sigma}_{t}^{-1} ) } - \mathbb{E} \big[ \| (X-L^{(t)}, \mathsf{X} - \mathsf{L}^{(t)}) \|^2_{ ( \Tilde{\Sigma}_{t}^{-1} - \Check{\Sigma}_{t}^{-1} ) } \big] \geq n K_t^{28} \Delta_t^2 \big) \,.  
\end{align*}
It remains to bound $\mathsf P_2$. To this end,  we see that there exist a linear transform $\mathbf{T}_{t}$ and a standard normal random vector $(U_{t},\mathsf{U}_{t})$ such that 
\begin{align*}
    (X-L^{(t)}, \mathsf{X} - \mathsf{L}^{(t)}) = \mathbf{T}_{t} (U_{t},\mathsf{U}_{t})
\end{align*}
and $\mathbf{T}_{t}^{*}\mathbf{T}_{t} = \Tilde{\Sigma}_{t}$ (so in particular $\| \Tilde{\Sigma}_{t} \|_{\mathrm{op}} = \| \mathbf{T}_{t} \|^2_{\mathrm{op}}$). Thus,
\begin{align*}
    \| (X-L^{(t)}, \mathsf{X} - \mathsf{L}^{(t)}) \|^2_{( \Tilde{\Sigma}_{t}^{-1} - \Check{\Sigma}_{t}^{-1} )} = \| (U_{t},\mathsf{U}_{t}) \|^2_{ \mathbf{T}_{t}(  \Tilde{\Sigma}_{t}^{-1} - \Check{\Sigma}_{t}^{-1} ) \mathbf{T}_{t}^{*} }
\end{align*}
is a quadratic form of a standard Gaussian vector. We also have the following estimate: 
\begin{align*}
    \| \mathbf{T}_{t}( \Tilde{\Sigma}_{t}^{-1} - \Check{\Sigma}_{t}^{-1} ) \mathbf{T}_{t}^{*} \|_{\mathrm{op}} \leq \| \mathbf{T}_{t} \|^2_{\mathrm{op}} \|  \Tilde{\Sigma}_{t}^{-1} - \Check{\Sigma}_{t}^{-1} \|_{\mathrm{op}} = \| \Tilde{\Sigma}_{t} \|_{\mathrm{op}} \|  \Tilde{\Sigma}_{t}^{-1} - \Check{\Sigma}_{t}^{-1} \|_{\mathrm{op}} \leq K_{t}^{20} \,,
\end{align*}
where the second inequality follows from Corollaries~\ref{corollary-2-HS-norm-W-inverse-minus-V-inverse} and \ref{corollary-op-norm-Sigma}.
In addition, we have 
\begin{align*}
    \| \mathbf{T}_{t}(  \Tilde{\Sigma}_{t}^{-1} - \Check{\Sigma}_{t}^{-1} ) \mathbf{T}_{t}^{*} \|_{\mathrm{HS}}^2 \leq \| \mathbf{T}_{t} \|^4_{\mathrm{op}} \|  \Tilde{\Sigma}_{t}^{-1} - \Check{\Sigma}_{t}^{-1} \|_{\mathrm{HS}}^2 = \| \Tilde{\Sigma}_t \|_{\mathrm{op}}^2 \|  \Tilde{\Sigma}_{t}^{-1} - \Check{\Sigma}_{t}^{-1} \|_{\mathrm{HS}}^2 \leq n K_{t}^{24} \Delta_{t}^2\,.
\end{align*}
where the first inequality follows from Lemma~\ref{lemma-HS-norm-product} and the second inequality follows from Corollaries~\ref{corollary-2-HS-norm-W-inverse-minus-V-inverse}, \ref{corollary-op-norm-Sigma}.
We can then apply Lemma~\ref{lemma-Hanson-Wright} and obtain that 
\begin{align*}
    \mathsf P_2 \leq 2 \exp \Big\{ - \Omega(1) \min \Big( \frac{ n K_{t}^{28} \Delta_{t}^2 }{ K_{t}^{20} }, \frac{ ( n K_{t}^{28} \Delta_{t}^2 )^2 }{n K_{t}^{24} \Delta_{t}^2} \Big) \Big\} \leq 2 \exp \{ - \Omega( n K_t^{8} \Delta_{t}^2 ) \} \,.
\end{align*}
Plugging the estimates of $\mathsf P_1,\mathsf P_2$ into \eqref{eq-prob-U-II-reduction} we get that the left hand side of \eqref{eq-prob-U-II-reduction} is bounded by $3 \exp \{ - \Omega( n K_t^{8} \Delta_{t}^2 ) \}$.
Combined with \eqref{eq-bound-mathcal-U-(I)} and Lemma~\ref{lemma-bound-conditional-unbiased-density} it yields that
\begin{align}
    & \mathbb{P} ( (g_{t-1} Y_t, g_{t-1} \mathsf{Y}_t) \in \mathcal{U}^{(I)} \cup \mathcal{U}^{(II)}; \mathcal{C} | \mathfrak{S}_{t-1}; \mathrm{BAD}_t = \mathrm{B}_t ) \leq \exp \{ n \Delta_t^5 - n K_t^{8} \Delta_t^2  \} \,. \label{equ-prob-subbad-event-2}
\end{align}
Combined with \eqref{eq-bound-prob-mathcal-C}, this yields that (by recalling \eqref{equ-def-delta} and noting that $\mathcal{B}_t^{c} \subset \mathcal{C}^{c} \cup \{ (g_{t-1} Y_t, g_{t-1} \mathsf{Y}_t) \in \mathcal U^I \cup \mathcal U^{II} \}$)
\begin{align}
    \mathbb{P} (\mathcal{B}_{t}^{c}; \mathcal{A}_{t-1}, \mathcal{B}_{t-1}, \mathcal{T}_{t}, \mathcal{E}_{t}, \mathcal{H}_{t} ) \leq  \exp \{ - \tfrac{1}{2} n^{ \frac{1}{\log \log \log n}} \} \,.
    \label{equ-prob-B-happens}
\end{align}

\subsubsection{Step 3: $\mathcal{A}_{t}$}
It is straightforward to bound the probability for $\mathcal{A}_t$ on the event $\mathcal{B}_t$. Indeed,
\begin{align*}
    & \mathbb{P} \Big(  \sum_{v \not \in \mathrm{BAD}_{t}} ( g_{t-1} Y_t(k,v) )^2 > 100 n ; \mathcal{B}_t \Big) \leq \exp \{ n K_{t}^{30} \Delta_t^2 \} \cdot \mathbb{P} \Big(  \sum_{v \not \in \mathrm{BAD}_{t}} \Check{Y}_t(k,v)^2 > 100n \Big)  \,,
\end{align*}
where the latter probability is bounded by $e^{- 2n}$ using Chernoff bound. Thus, applying union bound on $k$ we have 
\begin{align}
    \mathbb{P}( \mathcal{A}_{t}^{c} ; \mathcal{B}_t ) \leq K_t \exp \{ - n \} \,.
    \label{equ-prob-mathcal-A-t-happens}
\end{align}

\subsubsection{Step 4: $\mathcal{E}_{t+1}$}
Recall Definition~\ref{def-admissible} and \eqref{equ-def-admissible}.
The goal of this subsection is to prove
\begin{equation}\label{equ-mathcal-E-t+1-holds}
    \mathbb{P}( \mathcal{E}_{t+1}^{c}; \mathcal{A}_{t}, \mathcal{B}_{t}, \mathcal{E}_{t}, \mathcal{H}_{t}, \mathcal{T}_{t} ) \leq 2 K_t^2 \exp \{ - n \Delta_t^2 \} \,.
\end{equation}
To this end, we will verify Condition (i.)--(x.) in Definition~\ref{def-admissible}. Since (i.), (ii.) and (iii.) are controlled by \eqref{eq-mathcal-E-0-bound}, we then focus on the other conditions. In what follows, we always assume that $\mathcal A_{t}, \mathcal B_{t}, \mathcal{E}_{t}, \mathcal{H}_t$ and $\mathcal{T}_t$ hold. 
Crucially, we will reduce our analysis for events on 
$\{  W^{(t)}_v (k)+ \langle \eta^{(t)}_k, g_{t-1}{D}^{(t)}_v \rangle,  \mathsf W^{(t)}_{\pi(v)}(k)+ \langle \eta^{(t)}_k, g_{t-1}\mathsf{D}^{(t)}_{\pi(v)} \rangle : 1 \leq k \leq \frac{K_t}{12}, v \not \in \mathrm{BAD}_{t} \}$ under the conditioning of $\mathfrak{S}_{t-1}$ and $\mathrm{BAD}_t = \mathrm{B}_t$ to the same events on $\Check{\mathfrak F}_t$ (recalling \eqref{eq-def-Check-mathfrak-F-t}) thanks to $\mathcal B_t$; the latter would be much easier to estimate. To be more precise, note that
\begin{align}
    & \Big\{ \frac{|\Gamma^{(t+1)}_k \setminus \mathrm{BAD}_{t} |}{n} - \mathfrak{a} > \mathfrak{a} \Delta_{t+1} \Big\} \label{eq-recall-Gamma-Bad-sum} \\
    = & \Big\{ \frac{1}{n} \sum_{v \in V \setminus \mathrm{BAD}_{t}} \Big( \mathbf{1}_{ \{ | \frac{1}{\sqrt{2}} ( \sqrt{12/K_t} \langle \beta^{(t)}_k, W^{(t)}_v \rangle + \langle \sigma^{(t)}_k, D^{(t)}_v \rangle ) | \geq 10 \} } - \mathfrak{a} \Big) \geq \mathfrak{a} \Delta_{t+1} \Big\} \,. \nonumber
\end{align}
For $v \not \in \mathrm{BAD}_t$, by \eqref{equ-def-set-BIAS} we have $\| b_{t-1}{D}^{(t)}_v \| \leq K_t e^{-10 (\log \log n)^{10} } \leq \Delta_t^{10}$, and as a result $|\langle \sigma^{(t)}_k, D^{(t)}_v \rangle - \langle \sigma^{(t)}_k, g_{t-1}{D}^{(t)}_v \rangle| \leq K_t \Delta_t^{10} \ll \Delta_t^2$. Thus,  under the conditioning of $\mathfrak{S}_{t-1}$ and $\mathrm{BAD}_t = \mathrm{B}_t$ we have 
\begin{align*}
    \eqref{eq-recall-Gamma-Bad-sum}  \subset \Big\{ \frac{1}{n} \sum_{v \in V \setminus \mathrm{BAD}_{t}} \Big( \mathbf{1}_{ \{ | \frac{1}{\sqrt{2}} ( \sqrt{12/K_t} \langle \beta^{(t)}_k, W^{(t)}_v \rangle + \langle \sigma^{(t)}_k, g_{t-1}{D}^{(t)}_v \rangle ) | \geq 10 - \Delta_{t}^2  \} } - \mathfrak{a} \Big) \geq \frac{\mathfrak{a} \Delta_{t+1}}{2} \Big\} \,.
\end{align*}
Therefore, the conditional probability of \eqref{eq-recall-Gamma-Bad-sum}  can be bounded by the conditional probability of the right hand side in the preceding inequality under the conditioning of $\mathfrak{S}_{t-1}$ and $\mathrm{BAD}_t = \mathrm{B}_t$. We will tilt the measure to the same event on $\Check{\mathfrak F}_t$ (as explained earlier), and on $\mathcal{B}_{t}$ we know this tilting loses at most a factor of $\exp \{ n K_t^{30} \Delta_t^2 \}$. Also on $\mathcal T_t$ we have $\frac{|\mathrm{BAD}_{t}|}{n} \ll \mathfrak{a} \Delta_{t+1}$. Therefore, the conditional probability of \eqref{eq-recall-Gamma-Bad-sum} is bounded by the following probability up to a factor of $\exp \{ n K_t^{30} \Delta_t^2 \}$:
\begin{align}
    & \mathbb{P} \Big( \frac{1}{n} \sum_{v \in V \setminus \mathrm{B}_{t}} \Big( \mathbf{1}_{ \{ | \frac{1}{\sqrt{2}} ( \sqrt{12/K_t} \langle \beta^{(t)}_k, \Tilde{W}^{(t)}_v \rangle + \langle \sigma^{(t)}_k, \Check{D}^{(t)}_v \rangle) | \geq 10 - \Delta_t^2 \} } - \mathfrak{a} \Big) > \frac{ \mathfrak{a} \Delta_{t+1}}{2} \Big) \,. \label{eq-checker-probability-Gamma}
\end{align}
Since $\{ \frac{1}{\sqrt{2}} ( \sqrt{12/K_t} \langle \beta^{(t)}_k, \Tilde{W}^{(t)}_v \rangle + \langle \sigma^{(t)}_k, \Check{D}^{(t)}_v \rangle ) : v \in V \setminus \mathrm{B}_{t} \}$ is a collection of i.i.d.\ standard normal variables, we have $\eqref{eq-checker-probability-Gamma} \leq  \exp \{ - \frac{ n \mathfrak{a}^2 \Delta_{t+1}^2 }{10}\}$, and thus 
\begin{align*}
    \mathbb P(\eqref{eq-recall-Gamma-Bad-sum} \mid \mathfrak{S}_{t-1};\mathrm{BAD}_t = \mathrm{B}_t) \leq  \exp \{ - \frac{ n \mathfrak{a}^2 \Delta_{t+1}^2 }{20} \} \,.
\end{align*}
Similarly a lower deviation for $\frac{|\Gamma^{(t+1)}_k|}{n}-\mathfrak{a}$ can be derived, completing the verification for (iv.). 
The bounds on $\frac{|\Gamma^{(t+1)}_k \cap \Gamma^{(t+1)}_l|}{n}$, $\frac{|\Pi^{(t+1)}_k \cap \Pi^{(t+1)}_l|}{n} $ and $ \frac{| \pi(\Gamma^{(t+1)}_k) \cap \Pi^{(t+1)}_l|}{n}$ (which correspond to (iv.), (v.) and (vi.) respectively) can be proved similarly.

Furthermore, we bound $\frac{ | \pi(\Gamma^{(t+1)}_k) \cap \Pi^{(s)}_l |}{n}, \frac{ | \Gamma^{(t+1)}_k \cap \Gamma^{(s)}_l |}{n}, \frac{ | \Pi^{(t+1)}_k \cap \Pi^{(s)}_l |}{n} $ (which correspond to (vii.), (viii.), (ix.) and (x.) respectively). Note that under $\mathfrak{S}_{t-1}$, we have that $\Pi^{(s)}_l$'s are fixed subsets for $s\leq t$. In addition, on the event $\mathcal{E}_{t}$ we have $\Big| \frac{| \Pi^{(s)}_l |}{n} - \mathfrak{a}_s \Big| < \mathfrak{a}_s \Delta_s$. Thus,
\begin{align*}
    \frac{ | \Gamma^{(t+1)}_k \cap \Pi^{(s)}_l |}{n} - \mathfrak{a}_{t+1} \mathfrak{a}_s = \mathfrak{a}_{s} \Big( \frac{1}{\mathfrak{a}_s n}  \sum_{ u \in \Pi^{(s)}_l } \Big( \mathbf{1}_{ u \in \Gamma^{(t+1)}_k } - \mathfrak{a}_{t+1} \Big) \Big) + \mathfrak{a}_{t+1} \Big( \frac{|\Pi^{(s)}_l|}{n} - \mathfrak{a}_{s} \Big)\,.
\end{align*}
Since $\Big| \mathfrak{a}_{t+1} \Big( \frac{|\Pi^{(s)}_l|}{n} - \mathfrak{a}_{s} \Big) \Big| \leq \mathfrak{a}_{t+1} \mathfrak{a}_s \Delta_s$ on the event $\mathcal{E}_{t}$, the above can be bounded similarly to that for $\frac{|\Gamma^{(t+1)}_k |}{n} - \mathfrak{a}$. The same applies to the other two items here. We omit further details since the modifications are minor.
 
Putting all above together, we finally complete the proof of \eqref{equ-mathcal-E-t+1-holds}.

\subsubsection{Step 5: $\mathcal{H}_{t+1}$}
We assume that $\mathcal{A}_{t}, \mathcal{B}_{t}, \mathcal{H}_{t}, \mathcal{T}_t, \mathcal{E}_{t+1}$ hold throughout this subsection without further notice. Thus, we have
\begin{align*}
    & \sum_{v \not \in \mathrm{BAD}_{t+1}} \big| {\mathrm{PROJ}}^{\prime} (\langle \eta^{(t+1)}_k, g_t D^{(t+1)}_v \rangle) - \mathrm{PROJ}(\langle \eta^{(t+1)}_k, g_t D^{(t+1)}_v \rangle) \big|^2 \\
    \leq \ & \Delta_{t+1}^2 \sum_{s,k} \sum_{v \not \in \mathrm{BAD}_{t}} \langle \eta^{(s)}_k, g_{s-1} D^{(s)}_v \rangle^2 + n \Delta_{t+1}^2 \leq 2 n K_{t+1}^3 \Delta_{t+1}^2  \,,
\end{align*}
where the first inequality follows from Lemma \ref{lemma-projection-replace} and the second inequality relies on our assumption that $\mathcal{A}_t$ holds.
In light of this, to show $\mathcal{H}_{t+1}$ it suffices to bound for each $k$ the \emph{conditional} probability given $\mathfrak{S}_{t-1}$ and $\mathrm{BAD}_t$ of the event
\begin{align}
    \sum_{v \not \in \mathrm{BAD}_t} \big| {\mathrm{PROJ}}^{\prime} (\langle \eta^{(t+1)}_k, g_t D^{(t+1)}_v \rangle) \big|^2 > \frac{1}{4} n K_{t+1}^{6} \Delta_{t+1}^2\,.
    \label{equ-tail-sum-proj}
\end{align}
For notation  convenience, we will write $\mathbb P$ and $\mathbb E$ as $\mathbb P(\cdot \mid \mathfrak{S}_{t-1}; \mathrm{BAD}_t)$ and $\mathbb E(\cdot \mid \mathfrak{S}_{t-1}; \mathrm{BAD}_t)$ in the rest of the subsection. In order to bound \eqref{equ-tail-sum-proj}, we expand the matrix product in \eqref{equ-modified-projection-form} into a summation as follows (we write $[g Y,g\mathsf{Y}]_{t} (r,l,w) = [g Y]_{t} (r,l,w)$ and $[g Y,g\mathsf{Y}]_{t} (r,l,\pi(w)) = [g \mathsf{Y}]_{t} (r,l,\pi(w))$ for $w \not \in \mathrm{BAD}_t$ below):
\begin{equation}
\begin{aligned}
    \sum_{s,r=1}^{t} \sum_{l=1}^{K_r} \sum_{m=1}^{K_s} \sum_{\tau_1,\tau_2} \mathbf{J}_{t+1}( (k,v);(s,m,\tau_1) ) \mathbf{P}_t((s,m,\tau_1);(r,l,\tau_2)) [gY,g\mathsf{Y}]_{t} (r,l,\tau_2) 
    \label{equ-explicit-projection}
\end{aligned}
\end{equation}
where the summation is taken over $\tau_1, \tau_2 \in \big( V \setminus \mathrm{BAD}_t \big) \cup \big( \mathsf{V} \setminus \pi(\mathrm{BAD}_t) \big)$. So we know that $\sum_{v} \big| \mathrm{PROJ}^{\prime} (\langle \eta^{(t+1)}_k, g_t D^{(t+1)}_v \rangle) \big|^2$ is bounded up to a factor of $4K_t^2$ by the maximum of
\begin{align*}
    & \sum_{v} \Big|  \sum_{ \tau_1 \in \mathcal{V}_1, \tau_2 \in \mathcal{V}_2} \mathbf{J}_{t+1} ((k,v);(s,m,\tau_1)) \mathbf{P}_t((s,m,\tau_1);(r,l,\tau_2)) [gY,g\mathsf{Y}]_{t} (r,l,\tau_2)  \Big|^2\,,
\end{align*}
where the maximum is taken over $s,r \leq t, 1 \leq m \leq K_s, 1 \leq l \leq K_r , \mathcal{V}_i \in \big\{ V \setminus \mathrm{BAD}_t , \mathsf{V} \setminus \pi(\mathrm{BAD}_t) \big\}$. Thus, it suffices to bound each term in the maximum. For simplicity, we only demonstrate how to bound terms of the form:
\begin{align}
    \sum_{v} \Big| \sum_{u,w \in V \setminus \mathrm{BAD}_t} \mathbf{J}_{t+1} ((k,v);(s,l,\pi(u))) \mathbf{P}_t((s,m,\pi(u));(r,l,w)) [gY,g\mathsf{Y}]_{t} (r,l,w)  \Big|^2 \,.
    \label{equ_one_part_projection}
\end{align}
\begin{Lemma}{\label{lemma-one-part-projection}}
    We have for all $s,r,m,l$
    \begin{align*}
        \mathbb{P} \Big( \eqref{equ_one_part_projection} \geq \frac{1}{2} n K_{t+1}^{5} \Delta_{t+1}^2; \mathcal{A}_t, \mathcal{B}_t, \mathcal{H}_t, \mathcal{T}_t, \mathcal{E}_{t+1} \Big) \leq \exp \big\{ - \tfrac{1}{2} n \Delta_{t+1}^2 \big\} \,.
    \end{align*}
\end{Lemma}
Since the bounds for other terms are similar, by Lemma~\ref{lemma-one-part-projection} and a union bound, we get that $\mathbb{P}( \eqref{equ-tail-sum-proj} ; \mathcal{A}_t,\mathcal{B}_t,\mathcal{H}_t, \mathcal{T}_t, \mathcal{E}_{t+1} ) \leq  10 K_t^2 \exp \{ -n \Delta_{t+1}^2 /2\} $. Applying a union bound then yields that
\begin{equation}
    \label{equ-prob-mathcal-H-t+1}
    \mathbb{P}( \mathcal{H}_{t+1}^{c} ; \mathcal{A}_t, \mathcal{B}_t, \mathcal{H}_t, \mathcal{T}_{t},  \mathcal{E}_{t+1} ) \leq 20 K_{t+1}^4 \exp \big\{ - \tfrac{1}{2} n \Delta_{t+1}^2 \big\}  \,.
\end{equation}

\begin{proof}[Proof of Lemma~\ref{lemma-one-part-projection}]
Recall \eqref{eq-def-J}. We first divide the summation in $\eqref{equ_one_part_projection}$ into three parts $\mathcal S_1, \mathcal S_2, \mathcal S_3$, where $\mathcal S_1$ accounts for the summation over $u=v$ and can be written as
\begin{align*}
    \sum_{v} \Big| \sum_{w} \Hat{\mathbb{E}} \Big[ \langle \eta^{(t+1)}_k, g_t\Tilde{D}^{(t+1)}_v \rangle \langle \eta^{(s)}_m , g_{s-1}\Tilde{\mathsf{D}}^{(s)}_{\pi(v)} \rangle \Big] \mathbf{P}_t ((s,m,\pi(v));(r,l,w)) g_{r-1} Y_r(l,w)   \Big|^2\,,
\end{align*}
and $\mathcal S_2$ accounts for the summation over $u \neq v, r=t$ and can be written as
\begin{align*}
    \sum_{v} \Big| \sum_{u,w} \frac{(\mathbf{1}_{ \pi(v) \in \Pi^{(s)}_j }-\mathfrak{a}_s) ( \mathbf{1}_{u \in \Gamma^{(t+1)}_i } -\mathfrak{a}_{t+1}) }{ n \sqrt{(\mathfrak{a}_{t+1} - \mathfrak{a}_{t+1}^2) (\mathfrak{a}_s-\mathfrak{a}_s^2)} } \mathbf{P}_t((s,m,\pi(u));(t,l,w)) g_{t-1} Y_t(l,w) \Big|^2 \,,
\end{align*}
and $\mathcal S_3$ accounts for the summation over $u \neq v, r<t$ and can be written as
\begin{align*}
    \sum_{v} \Big| \sum_{u,w}  \frac{( \mathbf{1}_{ \pi(v) \in \Pi^{(s)}_j}-\mathfrak{a}_s) ( \mathbf{1}_{u \in \Gamma^{(t+1)}_i} -\mathfrak{a}_{t+1}) }{ n \sqrt{(\mathfrak{a}_{t+1}- \mathfrak{a}_{t+1}^2) (\mathfrak{a}_s - \mathfrak{a}_s^2)} } \mathbf{P}_t((s,m,\pi(u));(r,l,w)) g_{r-1} Y_r(l,w) \Big|^2 \,.
\end{align*}
By Cauchy-Schwartz inequality, we have that $\eqref{equ_one_part_projection} \leq 3 (\mathcal S_1 + \mathcal S_2 + \mathcal S_3)$. We first bound $\mathcal S_1$. Using \eqref{equ-degree-correlation-2}, on the event $\mathcal{E}_{t+1}$ we have 
\begin{align*}
    \Hat{\mathbb{E}} \left[ \langle \eta^{(t+1)}_k, g_t\Tilde{D}^{(t+1)}_v \rangle \langle \eta^{(s)}_m, g_{s-1}\Tilde{\mathsf{D}}^{(s)}_{\pi(v)} \rangle \right] = \eta^{(t+1)}_k \mathrm{P}_{\Gamma,\Pi}^{(t+1,s)} \left(\eta^{(s)}_m \right)^{*} + o(\Delta_{t+1}) \leq K_{t+1} \Delta_{t+1}   \,.
\end{align*}
Thus, recalling $\| \mathbf{P}_t \|_{\mathrm{op}} \leq 100$ we have
\begin{align}
    \mathcal S_1 \leq \ & \sum_{v} K_{t+1}^2 \Delta_{t+1}^2 \Big| \sum_{w} \mathbf{P}_t ((s,m,\pi(v));(r,l,w)) g_{r-1} Y_r(l,w) \Big|^2 \nonumber \\
    \leq \ & K_{t+1}^2 \Delta_{t+1}^2 \| \mathbf{P}_t \|^2_{\mathrm{op}} \sum_{w} \big| g_{r-1} Y_r(l,w) \big|^2 \overset{\mathcal{A}_t}{\leq} 10^4 K_{t+1}^4 \Delta_{t+1}^2 n \,. \label{equ-bound-proj-v=u}
\end{align}
Next we bound $\mathcal{S}_2$. A straightforward calculation yields that
\begin{align*}
    \mathcal S_2
    = \ & \frac{  K_{t+1}^2 ((1-2\mathfrak{a}_s) |\Pi^{(s)}_j \setminus \mathrm{BAD}_{t}| +\mathfrak{a}_s^2 n)} {(\mathfrak{a}-\mathfrak{a}^2) (\mathfrak{a}_s-\mathfrak{a}_s^2)n^2}  \\
    & * \Big( \sum_{u,w} ( \mathbf{1}_{u \in \Gamma^{(t+1)}_i} -\mathfrak{a})  \mathbf{P}_t((s,m,\pi(u));(t,l,w)) g_{t-1} Y_t(l,w) \Big)^2  \\
    \leq \ & \frac{ K_{t+1}^2 }{\mathfrak{a} n} \Big( \sum_{u,w}  ( \mathbf{1}_{u \in \Gamma^{(t+1)}_i} -\mathfrak{a} )  \mathbf{P}_t((s,m,\pi(u));(t,l,w)) g_{t-1} Y_t(l,w) \Big)^2 \,. 
\end{align*}
We tilt the measure on $\{ g_{t-1} Y_t | \mathfrak{S}_{t-1}; \mathrm{BAD}_t \}$ to $\{ \Check{Y}_t  \}$ again. Write $X_u = \Check{Y}_t(l,u) $ and $\mathtt{b}_u =  ( \mathbf{1}_{ \{ | \frac{1}{\sqrt{2}} ( \sqrt{12/K_t} \langle \beta^{(t)}_i, \Tilde{W}^{(t)}_u \rangle +  \langle \sigma^{(t)}_i, \Check{D}^{(t)}_u \rangle + \langle \sigma^{(t)}_i, b_{t-1}{D}^{(t)}_u \rangle ) | \geq 10 \}} -\mathfrak{a} )$. We will first bound
\begin{align}
    \mathbb{P} \Big( \big| \sum_{u,w} \mathtt{b}_u \mathbf{P}_t((s,m,\pi(u));(t,l,w)) X_w \big| > \frac{1}{10} \mathfrak{a} K_{t+1} n \Delta_{t+1} \Big) \,.  \label{equ-tail-quadratic}
\end{align}
Note that
\begin{align}
    \eqref{equ-tail-quadratic} \leq \ & \mathbb{P} \Big( \big| \sum_{u} \mathbf{P}_t((s,m,\pi(u));(t,l,u)) \mathtt{b}_u X_u \big| > \frac{1}{30} \mathfrak{a} K_{t+1} n \Delta_{t+1} \Big) \label{equ-tail-quadratic-part-I} \\
    + \ & \mathbb{P} \Big( \big| \sum_{u \neq w} \mathbf{P}_t((s,m,\pi(u));(t,l,w)) \mathbb{E}[\mathtt{b}_u] X_w \big| > \frac{1}{30} \mathfrak{a} K_{t+1} n \Delta_{t+1} \Big) \label{equ-tail-quadratic-part-II} \\
    + \ & \mathbb{P} \Big( \big| \sum_{u \neq w} \mathbf{P}_t((s,m,\pi(u));(t,l,w)) (\mathtt{b}_u - \mathbb{E}[\mathtt{b}_u] ) X_w \big| > \frac{1}{30} \mathfrak{a} K_{t+1} n \Delta_{t+1} \Big) \label{equ-tail-quadratic-part-III} \,.
\end{align}
To bound \eqref{equ-tail-quadratic-part-I}, note that $ \mathtt{b}_u X_u$'s for different $u$ are independent, and the entries of $\mathbf{P}_t$ are bounded by $100$ (since $\| \mathbf{P}_t \|_{\mathrm{op}} \leq 100$). Also, we have $|\langle \sigma^{(t)}_i, b_{t-1}{D}^{(t)}_u \rangle| \leq K_t \Delta_t^{10}$ for $u \not \in \mathrm{BIAS}_t$, and thus $| \mathbb{E}[ \mathtt{b}_u X_u ] | \lesssim K_t \Delta_t^{10}$ (this is the place where we use the symmetry from taking the absolute values as discussed below \eqref{equ-def-iter-sets}; in fact $| \mathbb{E}[ \mathtt{b}_u X_u ] |$ would have been 0 if $\langle \sigma^{(t)}_i, b_{t-1}{D}^{(t)}_u \rangle$ were 0). We then get that $\eqref{equ-tail-quadratic-part-I}  \leq \exp\{ - \Omega( n \mathfrak{a}^2 K_{t+1}^2 \Delta_{t+1}^2) \}$ by Chernoff bound. To bound \eqref{equ-tail-quadratic-part-II}, note that $\sum_{u \neq w} \mathbf{P}_t((s,m,\pi(u));(t,l,w)) \mathbb{E}[\mathtt{b}_u] X_w$ is a mean-zero Gaussian variable, with variance bounded by
\begin{align*}
    \sum_{w} \Big( \sum_{u \neq w} \mathbf{P}_t((s,m,\pi(u));(t,l,w)) \mathbb{E}[\mathtt{b}_u] \Big)^2 \leq \| \mathbf{P}_t \|_{\mathrm{op}}^2 \sum_{u} (\mathbb{E}[\mathtt{b}_u])^2 \leq n K_t^2 \Delta_t^{20}
\end{align*}
using $|\mathbb{E}[\mathtt{b}_u]| \leq K_t \Delta_t^{10}$ again. Thus we get $\eqref{equ-tail-quadratic-part-II} \leq \exp \{ - \frac{ (n \Delta_{t+1})^2 }{ n K_{t+1}^2 \Delta_t^{20} } \} \leq \exp \{ - n \}$.
We now bound \eqref{equ-tail-quadratic-part-III}. Note that $\{ \mathtt{b}_u - \mathbb{E}[ \mathtt{b}_u ], X_u \}$ are mean-zero sub-Gaussian variables, and are independent with $\{ \mathtt{b}_w - \mathbb{E}[ \mathtt{b}_{w} ], X_w : w \neq u \}$. In addition, we have $\| \mathbf{P}_t \|_{\mathrm{op}} \leq 100, \| \mathbf{P}_t \|^2_{\mathrm{HS}} \leq n K_t \| \mathbf{P}_t \|_{\mathrm{op}}^2 \leq K_t^2 n$. Thus, we can apply Lemma~\ref{lemma-modify-Hanson-Wright} and get that
\begin{align*}
    & \eqref{equ-tail-quadratic-part-III} \leq 2 \exp \Big\{ -\Omega(1)  \min \Big(  \frac{(\mathfrak{a} K_{t+1}n\Delta_{t+1})^2}{\| \mathbf{P}_t \|^2_{\mathrm{HS}}}, \frac{\mathfrak{a} K_{t+1} n\Delta_{t+1}}{ \| \mathbf{P}_{t} \|_{\mathrm{op}} } \Big) \Big\} \leq 2 \exp \{ - \Omega( n \Delta_{t+1}^2 K_{t+1} ) \} \,.
\end{align*}
Combining bounds on \eqref{equ-tail-quadratic-part-I} and \eqref{equ-tail-quadratic-part-II}, we have $\eqref{equ-tail-quadratic} \leq O(e^{-\Omega(n \Delta_{t+1}^2 K_{t+1} )})$. Thus, recalling definition of $\mathcal B_t$ and averaging over $\mathfrak S_{t-1}$ and $\mathrm{BAD}_t$ we have that
\begin{align}
    \mathbb{P}( \mathcal S_2 > \tfrac{1}{10} K_{t+1}^{4} n \Delta_{t+1}^2 ) \leq \exp \{ - n \Delta_{t+1}^2  \} \,. \label{equ-bound-proj-r=t}
\end{align}
It remains to bound $\mathcal S_3$. Again, using Cauchy-Schwartz inequality we have
\begin{align*}
    \mathcal S_3 \leq \frac{ K_{t+1}^2 }{ \mathfrak{a}n } \Big( \sum_{u,w} ( \mathbf{1}_{u \in \Gamma^{(t+1)}_i} -\mathfrak{a}) \mathbf{P}_t((s,m,\pi(u));(r,l,w)) g_{r-1} Y_r(l,w) \Big)^2 \,.  
\end{align*}
Again, by tilting the measure to $\{ \Check{Y}_t \}$, we get that $\mathbb{P}( \mathcal{S}_3 > \frac{1}{10} K_{t+1}^{4} n \Delta_{t+1}^2 )$ is bounded by (recall from above that $\mathtt{b}_u =  ( \mathbf{1}_{ \{ | \frac{1}{\sqrt{2}} ( \sqrt{12/K_t} \langle \beta^{(t)}_i, \Tilde{W}^{(t)}_u \rangle +  \langle \sigma^{(t)}_i, \Check{D}^{(t)}_u \rangle + \langle \sigma^{(t)}_i, b_{t-1}{D}^{(t)}_u \rangle ) | \geq 10 \}} -\mathfrak{a})$)
\begin{align*}
    \exp \{ n K_t^{30} \Delta_t^2 \} \cdot
    \mathbb{P} \Big( & \sum_{u,w} \mathtt{b}_u \mathbf{P}_t ((s,m,\pi(u));(t,l,w)) g_{r-1} Y_r(l,w) > \frac{1}{10} \mathfrak{a} n K_{t+1} \Delta_{t+1} \Big) \,.
\end{align*}
We can write $ \sum_{u,w} \mathtt{b}_u \mathbf{P}_t ((s,m,\pi(u));(t,l,w)) g_{r-1} Y_r(l,w) =  \sum_{u} \lambda_{u} \mathtt{b}_u $, where $\lambda_u$ is given by $\lambda_u = \sum_{w} \mathbf{P}_t((s,m,\pi(u));(t,l,w)) g_{r-1} Y_r(l,w)$. We see that $\lambda_u$'s satisfy that
\begin{align*}
    \sum_{u} \lambda_u^2 &= \sum_{u} \Big( \sum_{w} \mathbf{P}_t((s,m,\pi(u));(t,l,w)) g_{r-1} Y_r(l,w) \Big)^2 \\
    & \leq \| \mathbf{P}_t \|_{\mathrm{op}}^2 \sum_{w} \big( g_{r-1} Y_r(l,w) \big)^2 \overset{\mathcal{A}_t}{\leq} 10^4 n  \,.
\end{align*}
Using $| \langle \sigma^{(t)}_i, b_{t-1}{D}^{(t)}_u \rangle | \leq K_{t} \Delta_{t}^{10}$ again, we have $\mathbb{E}[\mathtt{b}_u] = O(K_t \Delta_t^{10})$. Thus, we get that $|\mathbb E[\sum_{u} \lambda_u \mathtt{b}_u] | \leq \sum_{u \not \in \mathrm{BAD}_{t-1}} |\lambda_{u}| K_t \Delta_{t}^{10} \ll n K_{t} \Delta_t^2$. Combined with Azuma-Hoeffding inequality, it yields that
\begin{align*}
    & \mathbb{P} \big( \sum_{u} \lambda_{u} \mathtt{b}_u > \tfrac{1}{10} \mathfrak{a} K_{t+1} n \Delta_{t+1} \big) \leq 2 \exp \Big\{ - \frac{ ( \frac{1}{10} \mathfrak{a} K_{t+1} n \Delta_{t+1} )^2 }{ 4\sum_{u} \lambda_u^2} \Big\} \leq 2 \exp \{ - \Omega(n)  K_{t+1} \Delta_{t+1}^2 \} \,.
\end{align*}
This then implies that (by using $\mathcal B_t$ and averaging again)
\begin{align}
    \mathbb{P}( \mathcal S_3 > \frac{1}{10} K_{t+1}^{4} n \Delta_{t+1}^2 ) \leq \exp \{ - n \Delta_{t+1}^2  \} \,.
    \label{equ-bound-proj-r<t}
\end{align}
Combined with \eqref{equ-bound-proj-v=u} and  \eqref{equ-bound-proj-r=t}, this completes the proof of Lemma \ref{lemma-one-part-projection}.
\end{proof}

\subsubsection{Conclusion}
By putting together \eqref{equ-prob-mathcal-T-happens}, \eqref{equ-prob-B-happens}, \eqref{equ-prob-mathcal-A-t-happens}, \eqref{equ-mathcal-E-t+1-holds} and \eqref{equ-prob-mathcal-H-t+1}, we have proved {\bf Step 1}--{\bf Step 5} listed at the beginning of this subsection. In addition, since $t^*\leq \log \log n$, our quantitative bounds imply that all these hold simultaneously for $t = 0, \ldots, t^*$ with probability $1-o(1)$. By the inductive logic explained at the beginning of this subsection, we complete the proof of Proposition~\ref{prop-cardinality-BAD}. We also point out that in addition we have shown that 
\begin{equation}\label{eq-prob-mathcal-E-diamond}
\mathbb P(\mathcal{A}_{t^*}, \mathcal B_{t^*}, \mathcal{H}_{t^*}, \mathcal{T}_{t^*} \cap \mathcal E_{t^*}) = 1-o(1)\,,
\end{equation}
 which will be used in Section \ref{sec:almost-exact-matching}.

\section{Almost exact matching} \label{sec:almost-exact-matching}

In this section we show that on the event $\mathcal{E}_{\diamond}= \mathcal{A}_{t^*} \cap \mathcal{B}_{t^*} \cap \mathcal{E}_{t^*} \cap \mathcal{H}_{t^*} \cap \mathcal{T}_{t^*}$, our algorithm matches all but a vanishing fraction of vertices with probability $1-o(1)$, thereby proving Proposition~\ref{prop-almost-exact-matching} (recall \eqref{eq-prob-mathcal-E-diamond}). For notational convenience, we will drop $t^*$ from subscripts (unless we wish to emphasize it). That is,  we will write $\varepsilon, K, \Delta, \eta_l, D_v, W_v,  \Check{Y}(k,v), \mathrm{BAD}$ instead of $\varepsilon_{t^*}, K_{t^*}, \Delta_{t^*}, \eta^{(t^*)}_l, D^{(t^*)}_v, W^{(t^*)}_v, \Check{Y}_{t^*} (k,v), \mathrm{BAD}_{t^*}$. 

In light of \eqref{equ-def-matching-ver}, we define $\mathtt U$ to  be the collection of $v\in V$ such that
\begin{align*}
    \sum_{k=1}^{\frac{1}{12}K} \big( W_v(k) +  \langle \eta_k, D_v \rangle \big) \big( \mathsf{W}_{\pi(v)}(k) + \langle \eta_k, \mathsf{D}_{\pi(v)} \rangle \big) < \frac{1}{100} K \varepsilon \,,
\end{align*}
and we define $\mathtt{E}$ to be the collection of directed edges $(u, w)\in V\cap \mathrm{BAD}^c \times V\cap \mathrm{BAD}^c$ (with $u\neq w$) such that 
\begin{align*}
    \sum_{k=1}^{\frac{1}{12}K} \big( W_u(k)+ \langle \eta_k, D_u \rangle \big) \big( \mathsf{W}_{\pi(w)}(k) + \langle \eta_k, \mathsf{D}_{\pi(w)} \rangle \big) \geq \frac{1}{100} K \varepsilon \,.
\end{align*}
It is clear that $\mathtt U$ and $\mathtt E$ will potentially lead to mis-matching for our algorithm in the finishing stage. As a result, our proof requires bounds on them.

\begin{Lemma}\label{lem-self-loops}
We have $\mathbb P(|\mathtt U| \geq n/\log n; \mathcal E_\diamond) = o(1)$.
\end{Lemma}
\begin{proof}
We assume $\mathcal E_\diamond$. If $|\mathtt U| > \frac{n}{\log n}$,  we have $|\mathtt U \cap \mathrm{BAD}^{c} | > \frac{n}{2 \log n}$ using $|\mathrm{BAD}| \ll \frac{n}{(\log n)^2}$. Let $U$ be a realization for $\mathtt U \cap \mathrm{BAD}^{c}$. Then $\langle \eta_k, b_{t^*-1} D_v \rangle \ll \Delta^{10}$  for $v \in U$. Thus,
\begin{align}
    \sum_{k=1}^{\frac{1}{12}K} \big( g_{t^*-1} Y(k,v) + o(\Delta^{10}) \big) \big( g_{t^*-1} \mathsf{Y}(k,\pi(v)) + o(\Delta^{10}) \big) < \frac{1}{100} K \varepsilon \mbox{ for } v \in U \,.
    \label{equ-event-self-edge}
\end{align}
We again use the tilted measure. By the definition of $\{ \langle \eta_k, \Check{D}_v \rangle, \langle \eta_k, \Check{\mathsf{D}}_{\mathsf{v}} \rangle \}$, the events 
\begin{align*}
    \Big\{  \sum_{k=1}^{\frac{1}{12}K} \big( \Check{Y}(k,v) + o(\Delta^{10}) \big) \big( \Check{\mathsf{Y}}(k, \pi(v))  + o(\Delta^{10}) \big) < \frac{1}{100} K \varepsilon \Big\}
\end{align*}
are independent for different $v$, and each has probability at most (by $\Delta \leq e^{-(\log \log n)^8} \ll \varepsilon$ and Lemma~\ref{lemma-Hanson-Wright} again)
\begin{align*}
    2 \exp \big\{ -\frac{ (K \varepsilon)^2 }{ K } \big\} \leq \exp \{ - (\log n)^{1.8} \} \,.
\end{align*}
Recalling the definition of $\mathcal B_{t^*}$ and recalling that $\mathcal E_\diamond \subset \mathcal B_{t^*}$, we derive that
\begin{align*}
   \mathbb P(\eqref{equ-event-self-edge}; \mathcal E_\diamond\mid \mathfrak S_{t^* - 1}, \mathrm{BAD}_{t^*}) \leq \exp \{ n \Delta^2 \}  \cdot \exp \{ - n (\log n)^{0.8}  \} \ll \exp \{ - \tfrac{1}{2} n(\log n)^{0.8}  \} \,,
\end{align*}
Since the enumeration for possible realizations of $\mathtt U$ is at most $2^n$, this completes the proof by a simple union bound.
\end{proof}

\begin{Lemma}\label{lem-bad-matching}
On $\mathcal E_\diamond$ with probability $1-o(1)$ we have the following: any subset of $\mathtt E$ has cardinality at most $\frac{n}{\log n}$ if each vertex is incident to at most one edge in this subset.
\end{Lemma}
\begin{proof}
Suppose otherwise there exists $U = \{ v_1, \ldots, v_M  \} \subset V \cap \mathrm{BAD}^{c} $  with $M=\frac{2n}{\log n}$ such that for all $1\leq i \leq M/2$
\begin{align}
    \mbox{$\sum_{k=1}^{K/12}$} \big( g_{t^*-1} Y(k,v_{2i-1}) + o(\Delta^{10}) \big) \big(  g_{t^*-1} \mathsf{Y}(k,\pi(v_{2i})) + o(\Delta^{10}) \big) \geq \tfrac{1}{100} K \varepsilon \,.
    \label{equ-event-edge-matching}
\end{align}
We again tilt the measure. Note that the events 
\begin{align*}
    \Big\{  \mbox{$\sum_{k=1}^{K/12}$} \Big( \Check{Y}(k,v_{2i-1}) + o(\Delta^{10}) \big) \big( \Check{\mathsf{Y}} (k,\pi(v_{2i}))  + o(\Delta^{10}) \Big) \geq \tfrac{1}{100} K \varepsilon \Big\} \mbox{ for } 1 \leq i \leq M/2
\end{align*}
are independent and each occurs with probability at most $\exp \{ - (\log n)^{1.8} \}$. Therefore, by similarly applying $\mathcal B_{t^*}$ and applying a union bound, we complete the proof of the lemma.
\end{proof}
We are now ready to provide the proof of  Proposition~\ref{prop-almost-exact-matching}.
\begin{proof}[Proof of  Proposition~\ref{prop-almost-exact-matching}]
By Lemmas~\ref{lem-self-loops} and \ref{lem-bad-matching}, we assume without loss of generality that the events described in these two lemmas both occur. Let $V_{\mathrm{fail}} = \{ v \in V : \Hat{\pi}(v) \neq \pi(v) \}$. Suppose in the finishing step (i.e, for the step of computing \eqref{equ-def-matching-ver}) our algorithm processes $V_{\mathrm{fail}} \setminus \mathtt U$ in the order of $w_1, w_2, \ldots, w_m$. For $w_k \not \in \mathtt U$, in order to have $\Hat{\pi}(w_k) \neq \pi(w_k)$, either our algorithm assigns a wrong matching to $w_k$ or at the time when processing $w_k$ the vertex $\pi(w_k)$ has already been matched to some other vertex. We then construct a directed graph $\overrightarrow{H}$ on vertices $\{ w_1, w_2, \ldots, w_m  \} \cup \mathtt U$ as follows: for each $v \in \{ w_1, w_2, \ldots, w_m  \} \cup \mathtt U$, if the finishing step puts $v$ into $\mathrm{SUC}$ and matches $v$ to some $\pi(u)$ with $\pi(u) \neq \pi(v)$, then we connect a directed edge from $v$ to $u$. Note our algorithm will not match a vertex twice, so all vertices have in-degree and out-degree both at most 1. Also, for $1 \leq k \leq m$, if $w_k$ has out-degree 0, then $\pi(w_k)$ must have been matched to some $u$ and thus there is a directed edge $(u,w_k) \in \overrightarrow{H}$. Thus, the directed graph $\overrightarrow{H}$ is a collection of non-overlapping directed chains. Since there are at least $\frac{m}{2}$ edges in $\overrightarrow{H}$ (recall that each $w_k$ is incident to at least one edge in $\overrightarrow{H}$), we can get a matching with cardinality at least $\frac{m}{4}$. Since $|\mathrm{BAD}| \ll n/(\log n)^2$, we can then get a matching restricted on $V \cap \mathrm{BAD}^c$ with cardinality at least $m/4 - n/(\log n)^2$. By the event in Lemma~\ref{lem-bad-matching}, we see that $m/4 - n/(\log n)^2 \leq n/\log n$, completing the proof.
\end{proof}

\appendix
\section{Index of notation}
\label{appendix:index}
 
Here we record some commonly used symbols in the paper, along with their meaning and the location where they are first defined. Local notations will not be included.

\begin{multicols}{2}
\begin{itemize}
\item $\overrightarrow{G},\overrightarrow{\mathsf G}$: pre-processed graphs; Subsection~\ref{sec:preprocessing}. 
\item $\Hat{q}$: pre-processed edge density; Subsection~\ref{sec:preprocessing}. 
\item $\Hat{\rho}$: pre-processed edge correlation; Subsection~\ref{sec:preprocessing}. 
\item $\iota_{\mathrm{lb}}$, $\iota_{\mathrm{ub}}$: lower and upper bounds for the increment of $\phi$; \eqref{eq-def-iota-ub-lb}.
\item $\kappa$: number of sets generated in initialization; \eqref{eq-kappa-choice}.
\item $\chi$: depth of neighborhood in initialization; \eqref{eq-def-chi}.
\item $\aleph^{(a)}_k,\Upsilon^{(a)}_k$: $a$-neighborhood of the seeds in initialization; \eqref{equ-def-initial-aleph-Upsilon}, \eqref{equ-def-iter-aleph-Upsilon}.
\item $\vartheta_a,\varsigma_a$: fraction of $a$-neighborhood and their interactions; \eqref{equ-def-iter-vartheta-varsigma}.
\item $\Gamma^{(t)}_k,\Pi^{(t)}_k$: sets generated at time $t$ in iteration; \eqref{equ-def-initial-set}, \eqref{equ-def-iter-sets}.
\item $\Phi^{(t)},\Psi^{(t)}$: matrices denote the overlapping structures of sets; \eqref{equ-initial-matrix}, \eqref{equ-def-iter-matrix}.
\item $\mathfrak{a}_t$: fraction of sets generated at time $t$; \eqref{eq-def-mathfrak-a-t}.
\item $K_t$: number of sets generated at time $t$ in iteration; \eqref{equ-def-iter-K}.
\item $\varepsilon_t$: signal contained in each pair at time $t$; \eqref{equ-def-iter-varepsilon}.
\item $\eta^{(t)}_k$: basis of projection spaces at time $t$; \eqref{equ-vector-unit}, \eqref{equ-vector-orthogonal}.
\item $\sigma^{(t)}_k$: direction of projections at time $t$; \eqref{equ-def-sigma}.
\item $t^*$: time when the iteration stops; \eqref{eq-def-t-*}.
\item $D^{(t)}_v,\mathsf{D}^{(t)}_{\mathsf v}$: degree of vertices to the sets at time $t$; \eqref{equ-def-degree}.
\item $\Delta_t$: targeted approximation error at time $t$; \eqref{equ-def-delta}. 
\item $\mathrm{M}_{\Gamma}, \mathrm{M}_{\Pi}, \mathrm{P}_{\Gamma,\Pi}$: matrices recording the correlation among iteration; \eqref{equ_martix_M_P}.
\item $\mathrm{REV}$: the vertices revealed in initialization; \eqref{eq-def-REV}.
\item $W_v^{(t)},\mathsf W_{\mathsf v}^{(t)}$: Gaussian vectors introduced for smoothing; Section~\ref{sec:iteration}.
\item $\mathfrak{S}_t$: the information used up to time $t$; \eqref{eq-def-mathfrak-S-t}.
\item $b_t D_v^{(t)}, g_t {\mathsf D}_{\mathsf v}^{(t)}$: the bias part and the good part of degree. \eqref{eq-def-b-t-D-t}.
\item $\overrightarrow{Z}, \overrightarrow{\mathsf Z}$: independently sampled Gaussian matrices; Subsection~\ref{sec:proof-outline}.
\item $\mathfrak{F}_t$: the Gaussian process up to time $t$; \eqref{eq-def-mathcal-F-t}.
\item $\mathcal{F}_t$: the information generated by the Gaussian process up to time $t$; \eqref{eq-def-mathcal-F-t}.
\item $\mathrm{BAD}_t$: the collection of bad vertices up to time $t$; \eqref{equ-def-set-BAD}.
\item $\mathrm{BIAS}_t$: the collection of vertices with large bias; \eqref{equ-def-set-BIAS}.
\item $\mathrm{LARGE}_t$: the collection of vertices with large degree; \eqref{equ-def-set-LARGE}.
\item $\mathrm{PRB}_t$: the collection of vertices with bad projection; \eqref{equ-def-set-PRB}.
\item $g_t \Tilde{D}_{v}^{(t)}, g_t \Tilde{\mathsf D}_{\mathsf v}^{(t)}$: the degree replaced by Gaussian; \eqref{equ-def-g-Tilde-D}.
\item $\mathbf Q_t$: covariance matrix of Gaussian variables at time $t$; Remark~\ref{remark-conditional-Gaussian}.
\item $\mathbf P_t$: covariance matrix of revised Gaussian variables at time $t$; Subsection~\ref{sec:Gaussian-analysis}.
\item $\mathbf H_t$: coefficient matrix of Gaussian variables at time $t$; Remark~\ref{remark-conditional-Gaussian}.
\item $\mathbf J_t$: coefficient matrix of revised Gaussian variables at time $t$;  Subsection~\ref{sec:Gaussian-analysis}.
\item $\Check{D}^{(t)}_v, \Check{\mathsf D}^{(t)}_{\mathsf v}$: Gaussian process with simple covariance structure; \eqref{eq-def-Check-Gaussian-process}.
\item $\mathfrak{F}_t'$: the revised Gaussian process up to time $t$; \eqref{eq-def-mathfrak-F-t'}.
\item $\mathcal{F}_t'$: the information generated by revised Gaussian process; \eqref{eq-def-mathfrak-F-t'}.

\end{itemize}
\end{multicols}

\end{document}